\documentclass[twocolumn]{svjour3} 
\usepackage{url}
\usepackage{graphicx}
\usepackage{balance}
\usepackage{subfigure}
\usepackage{color}
\usepackage{times,amsmath,epsfig}
\usepackage{amssymb}
\usepackage{textcomp}
\usepackage[linesnumbered,algoruled,boxed,lined]{algorithm2e}
\usepackage{cite}
\usepackage{todonotes} 
\usepackage{multirow}

\newcommand{\yy}[1]{{\color{magenta} #1}\normalcolor}

\newcommand{\commentnote}[1]{\todo[color=white,inline=true]{{ICDE comment: \footnotesize \textcolor{blue}{#1}\par}}\xspace}

\newcommand{\remove}[1]{}

\newcommand{\ignore}[1]{}

\usepackage{booktabs} 

\newtheorem{observation}{Observation}

\begin{document}
	
\title{Fast Data Series Indexing for In-Memory Data}
%




\author{Botao Peng         \and
	Panagiota Fatourou\and
	Themis Palpanas
}


\institute{B. Peng \at
	Institute of Computing Technology, Chinese Academy of Sciences \\
	\email{ pengbotao@ict.ac.cn}           
	\and
	P. Fatourou \at
	FORTH ICS
	\email { faturu@csd.uoc.gr}
	\and
	T. Palpanas \at
		LIPADE, Universit{\'e} de Paris \& French University Institute (IUF) \\
	\email{ themis@mi.parisdescartes.fr} 
}

\date{Received: date / Accepted: date}

\maketitle

\begin{abstract}
Data series similarity search is a core operation for several data series analysis applications across many different domains. 
However, the state-of-the-art techniques fail to deliver the time performance required for interactive exploration, 
or analysis of large data series collections.
In this work, we propose MESSI, the first data series index designed for in-memory operation on modern hardware.
Our index takes advantage of the modern hardware parallelization opportunities (i.e., SIMD instructions, multi-socket and multi-core architectures), in order to accelerate both index construction and similarity search processing times. 
Moreover, it benefits from a careful design in the setup and coordination of the parallel workers and data structures, so that it maximizes its performance for in-memory operations.
MESSI supports similarity search 
using both the Euclidean and Dynamic Time Warping (DTW) distances.
Our experiments with synthetic and real datasets demonstrate that overall MESSI is up to 4x faster at index construction, and up to 11x faster at query answering than the state-of-the-art parallel approach. 
MESSI is the first to answer exact similarity search queries on 100GB datasets in $\sim$50msec (30-75msec across diverse datasets), which enables real-time, interactive data exploration on very large data series collections.
\end{abstract}

%




\section{Introduction}

Several applications across many diverse domains (e.g., finance, astrophysics, etc.)
, such as in finance, astrophysics, neuroscience, engineering, multimedia, and others~\cite{DBLP:journals/sigmod/Palpanas15, fulfillingtheneed,DBLP:journals/dagstuhl-reports/BagnallCPZ19,Palpanas2019}, 
continuously produce big collections of data series\footnote{A data series, or data sequence, 
is an ordered sequence of data points. If the ordering dimension is time then we talk 
about time series, though, series can be ordered over other measures (e.g., angle in astronomical radial profiles, frequency in infrared spectroscopy, mass in mass spectroscopy, position in genome sequences, 
etc.).}
which need to be processed and analyzed.
Often times, this is part of an exploratory process, where users ask a query, review the results, and then decide what their subsequent queries, or analysis tasks should be.
The most common type of query that different analysis applications need to answer on these collections of data series 
is similarity search~\cite{DBLP:journals/sigmod/Palpanas15, lernaeanhydra, lernaeanhydra2}, which is at the core of several data series analysis tasks, such as classification and anomaly detection~\cite{lernaeanhydra, DBLP:conf/sdm/TanWP17, norma, normajournal, series2graph, valmodjournal, sand, sanddemo, eenergy21}.

The continued increase in the rate and volume 
of data series production 
with collections that grow to several terabytes in size~\cite{DBLP:journals/sigmod/Palpanas15,url:adhd,url:sds,DBLP:journals/pvldb/PelkonenFCHMTV15}, 
renders existing data series indexing technologies inadequate. 
For example, ADS+~\cite{zoumpatianos2016ads}, the state-of-the-art sequential (i.e., non-parallel) indexing technique, 
requires more than 2min to answer exactly a single 1-NN (Nearest Neighbor) query on a (moderately sized) 100GB sequence dataset.

Given the evolution of CPU performance, 
where the processor clock speed is not increasing due to the power wall constraint, 
algorithmic speedups can now mainly come by exploiting parallelism~\cite{zhou2002implementing,ZeuchFH14,gepner2006multi,LometN15,GowanlockC16,tang2016exploit,xiao2013parallelizing,ailamaki2015databases,DBLP:conf/ieeehpcs/Palpanas17,DBLP:conf/cidr/Blanas17,DBLP:conf/sigmod/PolychroniouRR15,DBLP:conf/damon/PolychroniouR14}. 
This involves (i) parallelism across compute nodes (e.g., using Spark)~\cite{dpisaxjournal,DBLP:journals/kais/LevchenkoKYAMPS21}, 
where the main goal is to scale to datasets that cannot be easily handled by a single node, 
and (ii) parallelism inside a single compute node (e.g., exploiting the Multi-Socket and Multi-Core (MSMC) architectures)~\cite{peng2018paris,parisplus,sing}, 
where the main goal is to minimize latency. 

In this study, we focus on parallelization inside a single node. 
The state-of-the-art approach, ParIS+~\cite{parisplus}, is a disk-based data series parallel indexing scheme that exploits the parallelism capabilities 
provided by MSMC architectures. 
Experiments showed that ParIS+ answers queries $10$x faster 
than ADS+, and more than $1000$x faster
than the optimized serial scan method. 
Still, ParIS+ is designed for disk-resident data and its performance is dominated by the I/O cost. 
For instance, ParIS+ answers a 1-NN (Nearest Neighbor) exact query on a 100GB dataset in 15sec, 
which is above the limit for keeping the user's attention (i.e., 10sec), 
and for supporting interactive analysis (i.e., 100msec)~\cite{Fekete:2016}.

\noindent{\bf{[Application Scenario]}}
In this work, we focus on designing an efficient parallel indexing and query answering scheme 
for \emph{in-memory} data series processing. Our work is motivated and inspired by the following
real scenario. Airbus\footnote{\scriptsize\url{http://www.airbus.com/}}, 
currently stores petabytes of data series, describing the behavior 
over time of various aircraft components (e.g., the vibrations of the bearings in the engines), 
as well as that of pilots (e.g., the way they maneuver the plane through the fly-by-wire system)~\cite{Airbus}. 
The experts need to access these data in order to run different analytics algorithms. 
However, these algorithms usually operate on a subset of the data (e.g., only the data relevant to landings 
from Air France pilots), which fit in memory. 
In order to perform complex analytics operations (such as searching for similar patterns, or classification) fast, in-memory data series indices must be built. 
Thus, the time cost of both index creation and query answering become important factors. 
Apart from engineering, similar needs appear in other domains and applications, as well~\cite{Palpanas2019,DBLP:journals/dagstuhl-reports/BagnallCPZ19}, such as astrophysics and neuroscience, where different, \emph{adhoc} subsets of data need to be analyzed, and for which we need to build indexes and then perform similarity search operations.

\noindent{\bf{[MESSI Approach]}}
We present MESSI, an in-MEmory data SerieS Index that
incorporates the state-of-the-art techniques in sequence indexing\footnote{A preliminary version of this work has appeared elsewhere~\cite{peng2020messi}.}, and inherently takes advantage of modern hardware parallelization in order to accelerate processing times.
MESSI supports similarity search queries on both z-normalized and non z-normalized data, using both the Euclidean and the Dynamic Time Warping (DTW) distance measures.

MESSI 
uses MSMC architectures 
in order to concurrently perform both index construction 
and query answering, and it exploits the Single Instruction Multiple Data (SIMD) capabilities of modern CPUs, in order to further parallelize the execution of individual instructions (mainly distance computations) 
inside each core. 
%
More importantly though, MESSI features a novel solution for answering exact 1-NN queries
which is 6-11x faster than an in-memory version of ParIS+ across the datasets (of size 100GB) we tested, achieving for the first time interactive exact query answering times, at $\sim$50msec.
It also provides redesigned algorithms that lead to a further $\sim$4x speedup in index construction time, in comparison to (in-memory) ParIS+.

The design decisions in ParIS+ were heavily influenced by the fact that 
the cost was mainly I/O bounded. 
Since MESSI copes with in-memory data series, no CPU cost can be hidden under I/O.
Therefore, MESSI required more careful design choices 
and coordination of the parallel workers. 
This led to the development of a more subtle design for the index construction 
and new algorithms for answering similarity search queries on this index. 

For query answering in particular, we showed that adaptations of alternative solutions,
which have proven to perform the best in other settings (i.e., disk-resident data~\cite{parisplus}), 
are not optimal in our case, so we designed a novel solution that achieves a good balance 
between the amount of communication among the parallel worker threads, and the effectiveness of each individual worker.

For instance, the new scheme uses concurrent priority queues for storing the data series 
that cannot be pruned, and for processing these series in order, starting from those whose iSAX representations have the smallest distance
to the iSAX representation of the query data series. 
In this way, the parallel query answering threads achieve better pruning on the data series they process.
Moreover, the new scheme uses the index tree to decide which data series to insert 
into the priority queues for further processing. In this way, the number of distance calculations performed
between the iSAX summaries of the query and data series is significantly reduced (ParIS+
performs this calculation for all data series in the collection).

To achieve load balancing, we had to come up with a scheme where
all priority queues had about the same number of elements.
ParIS+ had to perform this calculation for the entire collection of data series. 
We also experimented with several designs
for reducing the synchronization cost among different workers that access 
the priority queues and for achieving load balancing.
We ended up with a scheme where 
workers use randomization to choose the priority queues they will work on.
Consequently, MESSI answers exact 1-NN queries on 100GB datasets within 30-70msec across diverse synthetic and real datasets.

The index construction phase of MESSI differentiates from ParIS+ in several ways.
For instance, ParIS+ was using a number of buffers to temporarily store pointers to the iSAX summaries
of the raw data series before constructing the tree index~\cite{parisplus}.
MESSI allocates smaller such buffers per thread and stores in them the iSAX summaries themselves.
In this way, it completely eliminates the synchronization cost in accessing the iSAX buffers.
To achieve load balancing, MESSI splits the array storing the raw data series into small blocks, and assigns blocks to threads
in a dynamic 
fashion. 
We applied the same technique when assigning to threads the buffers containing the iSAX summary of the data series.
Overall, the new design and algorithms 
of MESSI led to $\sim$4x improvement in index construction time when compared to (in-memory) ParIS+.
Still, the main contribution of the paper is our novel query answering scheme,
which results in up to 11x better performance than ParIS+.
This scheme supports similarity search on both Z-normalized and non Z-normalized data, 
and can be used with either the Euclidean, or the Dynamic Time Warping (DTW) distance.

\noindent{\bf [Contributions]} 
Our contributions are summarized below. 
\begin{itemize}

\item
We propose MESSI, the first in-memory data series index designed for modern hardware, 
which can answer similarity search queries in a highly efficient manner.

\item
We implement a novel, tree-based exact query answering algorithm for both the Euclidean and Dynamic Time Warping (DTW) distances, 
which minimizes the number of 
distance calculations 
(both lower bound distance calculations for pruning true negatives, 
and real distance calculations for pruning false positives).

\item
We also design an index construction algorithm that effectively balances 
the workload among the index creation workers by using 
a parallel-friendly index framework with low synchronization cost.

\item 
We provide proofs of correctness for our parallel algorithms. 
These proofs guarantee that both the index creation and query answering algorithms will always produce correct results.

\item
We conduct an experimental evaluation with several synthetic and real datasets, 
which demonstrates the efficiency of the proposed solution. 
The results show that MESSI is up to 4.2x faster at index construction and up to 11.2x faster at query answering than the state-of-the-art parallel index-based competitor, up to 109x faster at query answering than the state-of-the-art parallel serial scan algorithm, and thus can significantly reduce the execution time of complex analytics algorithms (e.g., \emph{k-NN} classification by more than 1 order of magnitude). 
\end{itemize}

\noindent{\bf [Paper Structure]} 
The rest of this paper\footnote{A preliminary version of this paper has appeared elsewhere~\cite{peng2020messi}.} is organized as follows.
In Section~\ref{sec:prelim}, we provide the necessary background material. 
The MESSI approach is described in Section~\ref{sec:parmis}. 
Section~\ref{secproof} is the proof of correctness of our index creation and query answering algorithms.
Section~\ref{sec:experiments} contains our experimental analysis.
We review the related work in Section~\ref{sec:rel},
and conclude in Section~\ref{sec:conclusions}.

\section{Background}
\label{sec:prelim}

We now provide some necessary definitions, and introduce 
background knowledge on state-of-the-art data series indexing.

\subsection{Data Series and Similarity Search}

\noindent{\bf [Data Series]}
A data series, $S=\{p_1, ..., p_n\}$, is defined as a sequence of points, 
where each point $p_i=(v_i,t_i)$, $1 \le i \le n$\yy{,} 
is associated to a real value $v_i$ and a position $t_i$.
The position corresponds to the order of this value in the sequence (in the case of time series, positions are expressed in terms of time, i.e., they are timestamps).
We call $n$ the \emph{size}, or \emph{length} of the data series.
All discussions in this work are applicable to general high-dimensional vectors, too.

\noindent{\bf [Similarity Search]}
Analysts perform a wide range of data mining tasks on data series 
including clustering~\cite{keogh1998,liao2005,rodrigues2008,rakthanmanon2011}, 
classification and deviation detection~\cite{Shieh2009,Shandola2009}, 
and frequent pattern mining~\cite{DBLP:journals/datamine/MueenKZCWS11,DBLP:journals/tkdd/GrabockaSS16}. 
Existing algorithms for executing these tasks rely on performing fast similarity search 
across the different series.
Thus, efficiently processing Nearest Neighbor (NN) queries is crucial 
for speeding up the above tasks.
NN queries are defined as follows: given a query series $S_q$ of length $n$,  
and a collection $\mathcal{S}$ of sequences of the same length, $n$, 
we want to identify the series $S_c \in \mathcal{S}$ 
that has the smallest distance to $S_q$ among all the series in the collection $\mathcal{S}$.
(In the case of streaming series, we first create subsequences of length $n$ using a sliding window, and then index those.)

Common distance measures for comparing data series are Euclidean Distance (ED)~\cite{Agrawal1993} 
and Dynamic Time Warping (DTW)~\cite{rakthanmanon2012searching}, which performs better for data mining tasks (e.g., classification~\cite{DBLP:journals/datamine/BagnallLBLK17}).
Euclidean distance is computed as the sum of distances between the pairs of corresponding points in the two sequences.
Note that minimizing ED on z-normalized data (i.e., a series whose values have mean 0 and standard deviation 1) is equivalent to maximizing their Pearson's correlation coefficient~\cite{MueenNL10}.

\noindent{\bf [Distance calculation in SIMD]}
Single-Instruction Multiple-Data (SIMD) refers to a parallel architecture that allows the execution of the same operation on multiple data simultaneously~\cite{lomont2011introduction}. 
Using SIMD, we can reduce the latency of an operation, because the corresponding instructions are fetched once, and then applied in parallel to multiple data.
All 
modern CPUs support 256-bit wide SIMD vectors, 
which means that certain floating point (or other 32-bit data) computations can be up to 8 times faster. 
%

In the data series context, SIMD has been employed for the computation of the Euclidean distance functions~\cite{tang2016exploit}, 
as well as in the ParIS+ index, for the conditional branch calculations during the computation of the lower bound distances~\cite{parisplus}.


\begin{figure}[tb]

\begin{minipage}[b]{0.41\columnwidth}
\subfigure[raw data series\label{fig:saxa}] {
	\hspace{-0.5em}
	\includegraphics[page=1,width=0.92\columnwidth]{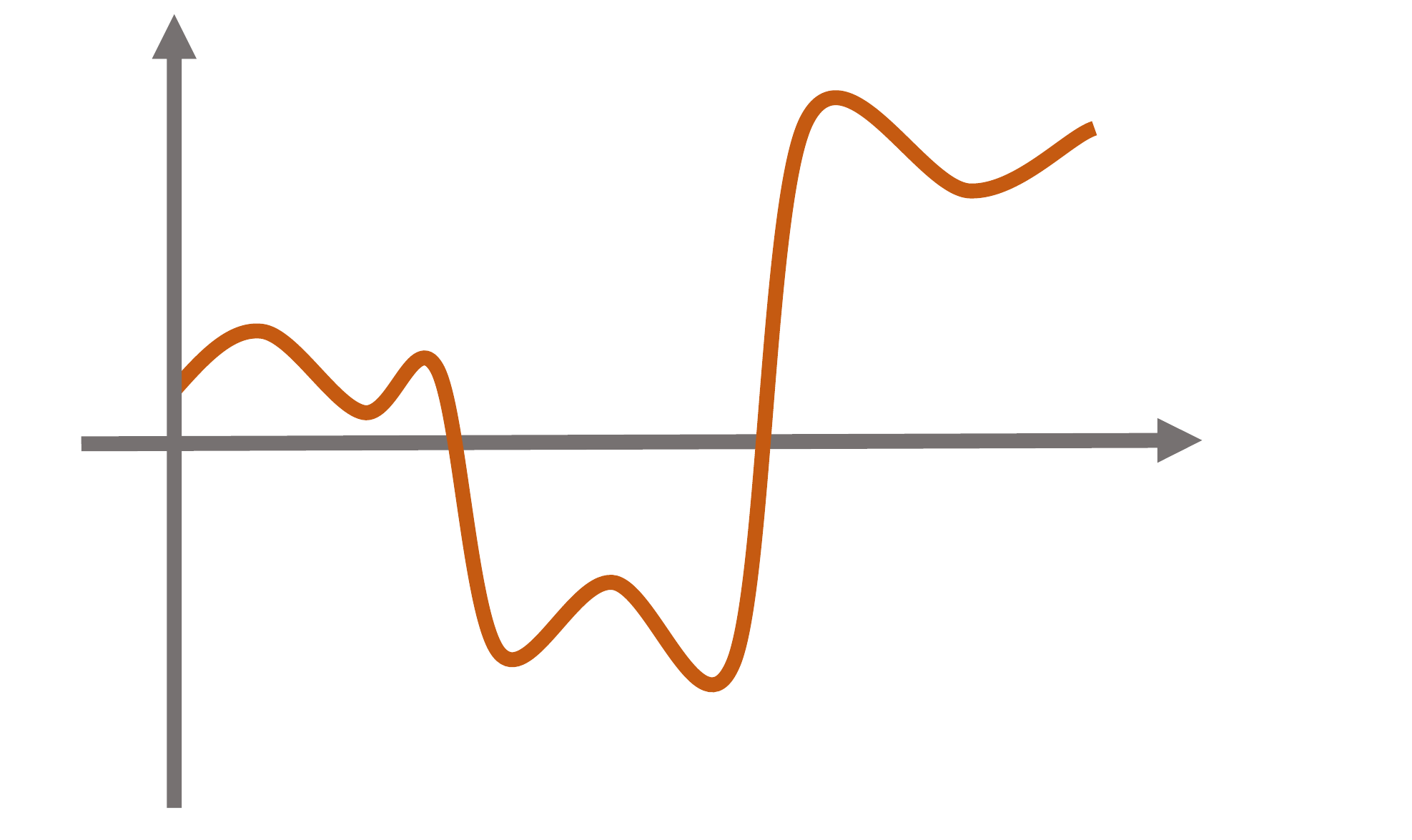}
}

\subfigure[PAA representation\label{fig:saxb}] {
	\hspace{-0.5em}
	\includegraphics[page=2,width=0.95\columnwidth]{picture/sax2}
}
\subfigure[iSAX representation\label{fig:saxc}]{
	\hspace{-0.5em}
	\includegraphics[page=3,width=0.97\columnwidth]{picture/sax2}
}
\end{minipage}
\hspace*{0.2cm}
\subfigure[ParIS+ index\label{fig:ads}] {
	\hspace{-1em}
	\includegraphics[width=0.57\columnwidth]{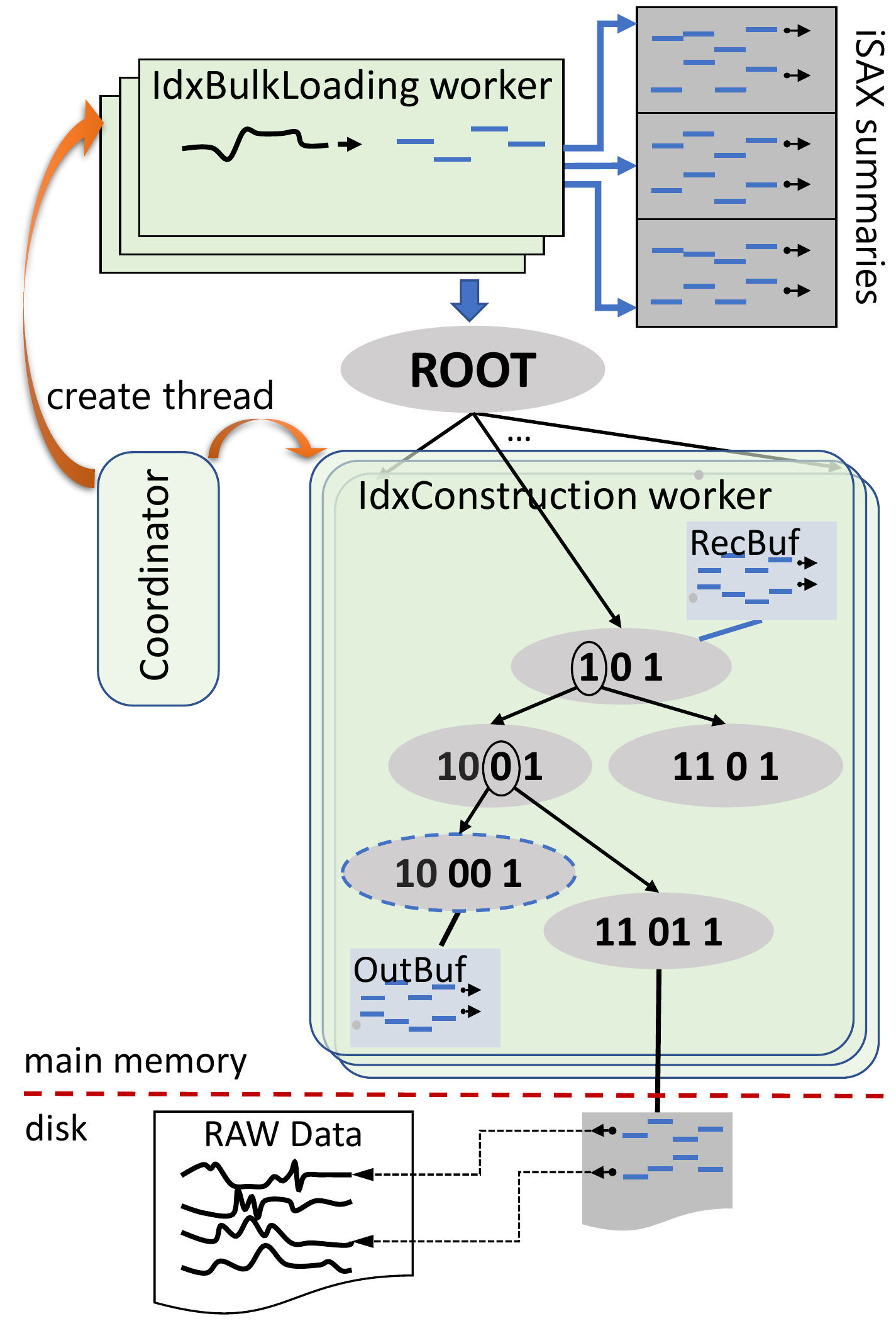}
}
\caption{The iSAX representation, and the ParIS+ index}
\end{figure}

\subsection{iSAX Representation and the ParIS+ Index}

\noindent{\bf [iSAX Representation]}
The iSAX representation (or summary) is based on the 
Piecewise Aggregate Approximation (PAA) representation~\cite{keogh2001dimensionality}, which divides the data series 
in $w$ segments of equal length, and uses the mean value of the points in each segment in order to summarize a data series. 
Figure~\ref{fig:saxb} illustrates an example of PAA representation with three segments (shown with the black horizontal lines), for the data series depicted in Figure~\ref{fig:saxa}. 

Based on PAA, the indexable Symbolic Aggregate approXimation (iSAX) representation was proposed~\cite{shieh2008sax} (and later used in several different data series indices~\cite{Shieh2009, zoumpatianos2016ads, DBLP:journals/pvldb/KondylakisDZP18, peng2018paris, ulissevldb}).
This method first divides the (y-axis) space in different regions, and assigns a bit-wise symbol to each region.
In practice, the number of symbols is small: 
previous work has shown that 
iSAX achieves very good approximations with just 256 symbols, the maximum alphabet cardinality, $|alphabet|$ represented by $8$ bits~\cite{isax2plus}.
It then represents each of the $w$ segments of the series
not by the real value of the PAA, but 
with the symbol of the region the PAA falls into, forming the word $10_2 00_2 11_2$ shown in Figure~\ref{fig:saxc} (subscripts denote the number of bits used to represent the symbol of each segment). 

Therefore, iSAX further reduces the size of the data series summarization, and more importantly it leads to a bit-wise representation. 
Note that, even though the iSAX representation was invented several years ago, it remains one of the most popular data series summarization methods.
Recent studies have shown that data series indices based on iSAX achieve state-of-the-art performance in various similarity search tasks~\cite{lernaeanhydra,lernaeanhydra2}.
{For an overview of iSAX-based indices, see~\cite{evolutionofanindex}.


\noindent{\bf [ParIS+ Index]}
Based on the iSAX representation, the ParIS+ index was developed~\cite{parisplus}, which proposed techniques and algorithms specifically designed for modern hardware and disk-based data. 

ParIS+ makes use of variable cardinalities for the iSAX summaries (i.e., variable degrees of precision for the symbol of each segment) 
in order to build a hierarchical tree index (see Figure~\ref{fig:ads}), consisting of three types of nodes:
(i) the root node points to several children nodes, $2^w$ in the worst case (when the series in the collection 
cover all possible iSAX summaries); (ii) each inner node contains the iSAX summary of all the series
below it, and has two children; and (iii) each leaf node contains the iSAX summaries of all the series inside it, and pointers to the raw data (in order to be able to prune false positives and produce exact, correct answers), which reside on disk. 
When the number of series in a leaf node becomes greater than the maximum leaf capacity, the leaf splits: 
it becomes an inner node and creates two new leaves, by increasing the cardinality of the iSAX summary 
of one of the segments (the one that will result in the most balanced split of the contents of the node 
to its two new children~\cite{isax2plus,zoumpatianos2016ads}). 
The two refined iSAX summaries (new bit set to \textit{0} and \textit{1}) are assigned to the two new leaves. 
In our example, the series of Figure~\ref{fig:saxc} will be placed in the outlined node of the index (Figure~\ref{fig:ads}).
The distance of a query 
to a node is the distance between the query (raw values, or iSAX summary) and the node's iSAX summary.

In the index construction phase (see Figure~\ref{fig:ads}), ParIS+ uses a coordinator worker that reads raw data series from disk and transfers them into 
a raw data buffer in memory. 
A number of index bulk loading workers compute the iSAX summaries of these series, 
and insert $<$iSAX summary, file position$>$ pairs in an array. 
They also insert a pointer to the appropriate element of this array 
in the receiving buffer of the corresponding subtree of the index root.
When main memory is exhausted, the coordinator worker creates a number of index construction worker threads, each one assigned to one subtree of the root and responsible for further building that subtree (by processing the iSAX summaries stored in the corresponding receiving buffer). 
This process results in each iSAX summary being moved to the output buffer of the leaf it belongs to. 
When all iSAX summaries in the receiving buffer of an index construction worker have been processed, the output buffers of all leaves in that subtree are flushed to disk. 

For query answering, ParIS+ offers a parallel implementation of the SIMS exact search  algorithm~\cite{zoumpatianos2016ads}. 
It first computes an approximate answer by calculating the real distance between the query and the best candidate series, 
which is in the leaf with the smallest lower bound distance to the query. 
ParIS+ uses the index tree only for computing this approximate answer.
Then, a number of lower bound calculation workers compute the lower bound distances 
between the query and the iSAX summary of each data series in the dataset, which are stored in the SAX array,
and prune the series whose lower bound distance is larger than the approximate real distance computed earlier.
The data series that are not pruned, are stored in a candidate list for further processing.
Subsequently, a number of real distance calculation workers operate on different parts of this array
to compute the real distances between the query and the series stored in it
(for which the raw values need to be read from disk).
For details see~\cite{parisplus}.

In the in-memory version of ParIS+, 
the raw data series are stored in an in-memory array.
Thus, there is no need for a coordinator worker. 
The bulk loading workers operate directly on this array (split to as many chunks as the workers). 
In the rest of the paper, we use ParIS+ to refer to this in-memory version. 

\section{The MESSI Solution}
\label{sec:parmis}

\begin{figure*}[tb]
	\centering
	\includegraphics[page=1,width=0.65\textwidth]{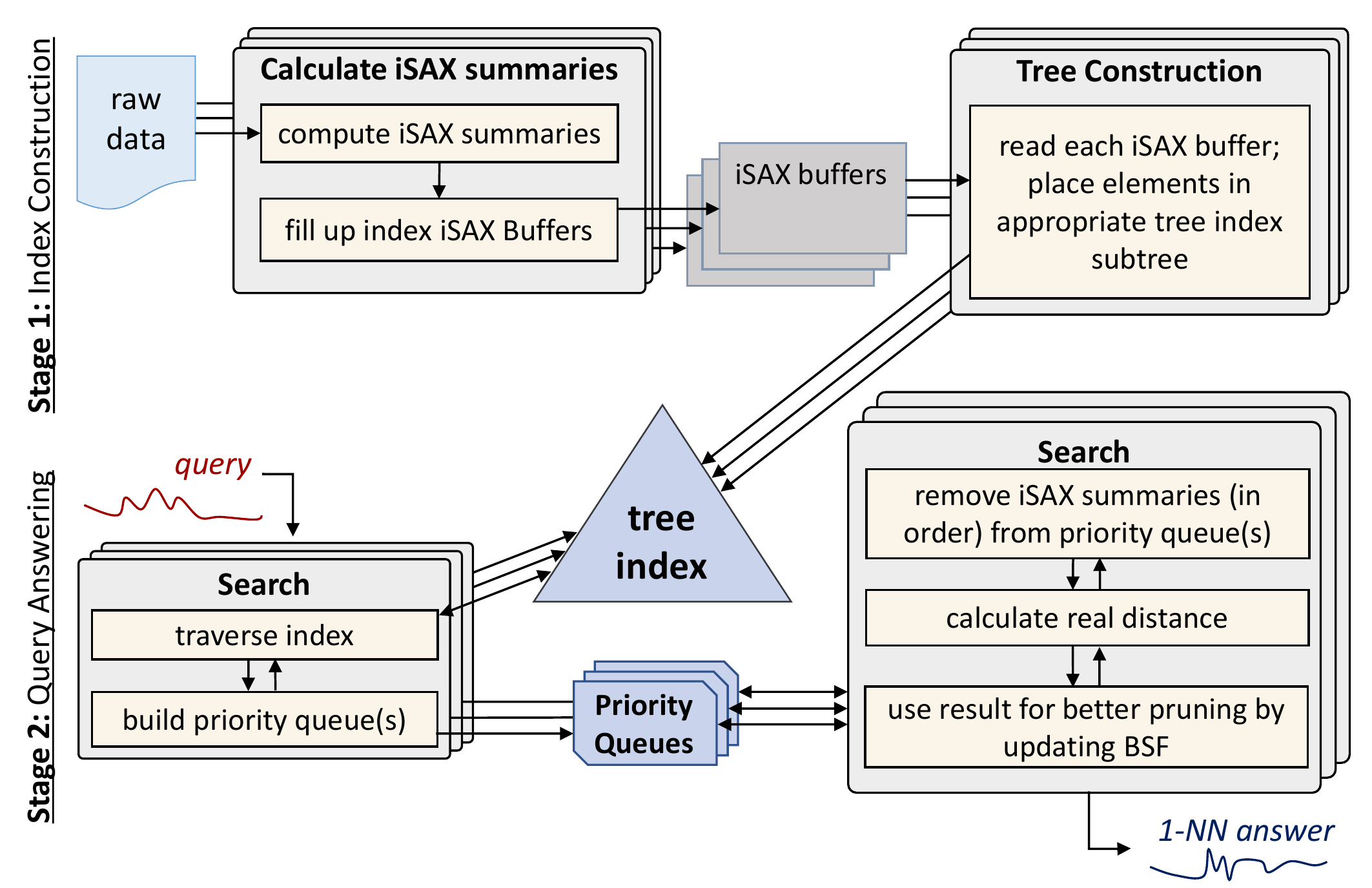}
	\caption{MESSI index construction and query answering}
	\label{fig:workflow}
	
\end{figure*}

The parallelism approach we employ in MESSI is governed by two main principles : 
1) eliminate synchronization overheads as much as possible, 
and 2) balance the load of the index workers. These two principles
often require contradicting design choices, 
so the design of MESSI is based on extensive experimentation to find the best compromise whenever needed.

We first outline the main ideas 
of MESSI. 
Figure~\ref{fig:workflow} depicts the MESSI index construction and query answering pipeline.
MESSI uses an index tree, comprised of several root subtrees. 
A number of index workers are responsible to construct the index tree. 
To avoid synchronization overheads and exploit locality, 
each subtree is built by a distinct worker.
To achieve load balancing, workers are assigned subtrees on the fly, 
with different threads possibly processing different numbers of subtrees
(depending on the work necessary on each subtree),
so that they are all busy most of the time.
To avoid synchronization overheads, 
workers are assigned to work on disjoint data subsets.
This way 
workers never interfere with one another. 

For query answering, 
workers traverse the tree pruning nodes
whenever possible. 
MESSI uses a number of shared priority queues
to store leaf nodes that are not pruned. 
For reducing the synchronization cost and exploit locality, 
different workers traverse different subtrees of the index tree. 
For ensuring load balancing, each worker 
adds elements in the queues in a round-robin fashion; this way,
all queues end up having approximately the same number of elements.
After the queues have been populated, the workers  process the nodes in the priority queues. 
The priority of a queue node is its lower bound distance from the query series,
so if a DeleteMin operation returns a node whose distance is larger than the current best distance,
all nodes in the queue can be pruned (i.e., the worker gives up the entire queue).
This scheme allows MESSI to perform additional pruning when processing queue nodes, 
and results in a reduced number of real-distance computations. 
In our implementation, more than one threads work on each queue, so that the real distance calculations
on a node's series (which is a time-consuming task) overlaps with the deletion of additional nodes from 
the queue. 
However, we have chosen the number of threads to work on each queue with care to avoid high synchronization overheads.

\subsection{Preliminaries}

We proceed with the details of MESSI. 
The raw data are stored 
into the $RawData$ array,
which is split into a predetermined number of chunks. 
A number, $N_w$, of {\em index worker} threads process the chunks
to calculate the iSAX summaries of the raw data series they store.
The number of chunks is not necessarily the same as $N_w$.
Chunks are assigned to index workers the one after the other using Fetch and Increment (Fetch\&Inc).
Based on the iSAX representation, we compute 
in which subtree of the index an iSAX summary will be stored.

Each index worker stores the iSAX summaries it computes in the appropriate iSAX buffers. 
Each iSAX buffer is split into $N_w$ parts and each worker works on its own part\footnote{
	We also tried an alternative design, where buffers were not split, 
	so many threads could try to update each element of a buffer concurrently. 
	Therefore, each buffer had to be protected by a lock.
        This design resulted in worse performance due to 
	the contention in accessing the iSAX buffers. 
	}.
The number of iSAX buffers is usually a few tens of thousands and at most $2^w$, 
where $w$ is the number of segments in the iSAX summaries 
of each data series ($w$ is fixed to $16$ in this paper, 
as in previous studies~\cite{zoumpatianos2016ads,peng2018paris}).

When the iSAX summaries for all data series have been computed, 
the index workers proceed in the construction of the tree index. 
Each worker is assigned an iSAX buffer to work on (this is done
again using Fetch\&Inc). 
Each worker reads the data stored in (all parts of) 
its assigned buffer and builds the corresponding index subtree. 
Therefore, all index workers process distinct subtrees of the index, 
and work in parallel and independently from one another\footnote{
Parallelizing the processing inside each one of the index root subtrees would require 
a lot of synchronization due to node splitting. 
When a node is split, two new leaf nodes are created and the data of the original leaf are moved to the new leaves.
	}.
When an index worker finishes with the current iSAX buffer it works on, 
it continues with the next iSAX buffer that has not yet been processed.

When the series in all iSAX buffers have been processed, 
the tree index has been built and can be used to answer similarity search queries, as depicted in 
the query answering phase of Figure~\ref{fig:workflow}. 
To answer a query, we first perform a search for the query iSAX summary in the tree index.
This returns a leaf whose iSAX summary has the closest distance
to the iSAX summary of the query.
We calculate the real distance of the (raw) data series pointed to by the elements
of this leaf to the query series, and store the minimum of these distances 
into a shared variable, called BSF (Best-So-Far). 
Then, the index workers start traversing the index subtrees
(the one after the other) using BSF to decide 
which subtrees will be pruned. 
The leaves of the subtrees that cannot be pruned are placed 
into (a fixed number of) minimum priority queues, using the lower bound distance between the raw values of the query series and the iSAX summary of the leaf node, in order to be further examined. 
Each thread inserts elements in the priority queues in a round-robin fashion
so that load balancing is achieved (i.e., all queues contain about 
the same number of elements). 

As soon as the necessary elements have been placed in the priority queues,
each index worker chooses a priority queue to work on, and repeatedly calls DeleteMin() on it to get a leaf node, on which it performs the following operations. 
It first checks whether the lower bound distance stored in the priority queue is larger than the current BSF: if it is then we are certain that the leaf node does not contain any series that can be part of the answer, and we can prune it; otherwise, the worker needs to examine the series contained in the leaf node, by first computing lower bound distances using the iSAX summaries, and if necessary also the real distances using the raw values. 
During this process, we may discover a series with a smaller distance to the query, in which case we also update the BSF.
When a worker reaches a node whose distance
is bigger than the BSF, it gives up this priority queue
and starts working on another, since all other elements in the abandoned queue have a higher distance to the query. 
This process is repeated until
all priority queues have been processed, 
and the BSF is updated along the way.
At the end of the calculation,
the value of BSF is returned as the query answer. 

Note that, similarly to ParIS+, MESSI uses SIMD (Single-Instruction Multiple-Data) 
for calculating the distances of both the index iSAX summaries 
from the query iSAX summary ({\em lower bound distance calculations}), 
and the raw data series from the query data series ({\em real distance calculations})~\cite{parisplus}.

	\begin{figure*}[tb]
	\centering
	\subfigure[CalculateiSAXSummaries\label{fig:inc2a}]{
		\includegraphics[page=1,width=0.65\columnwidth]{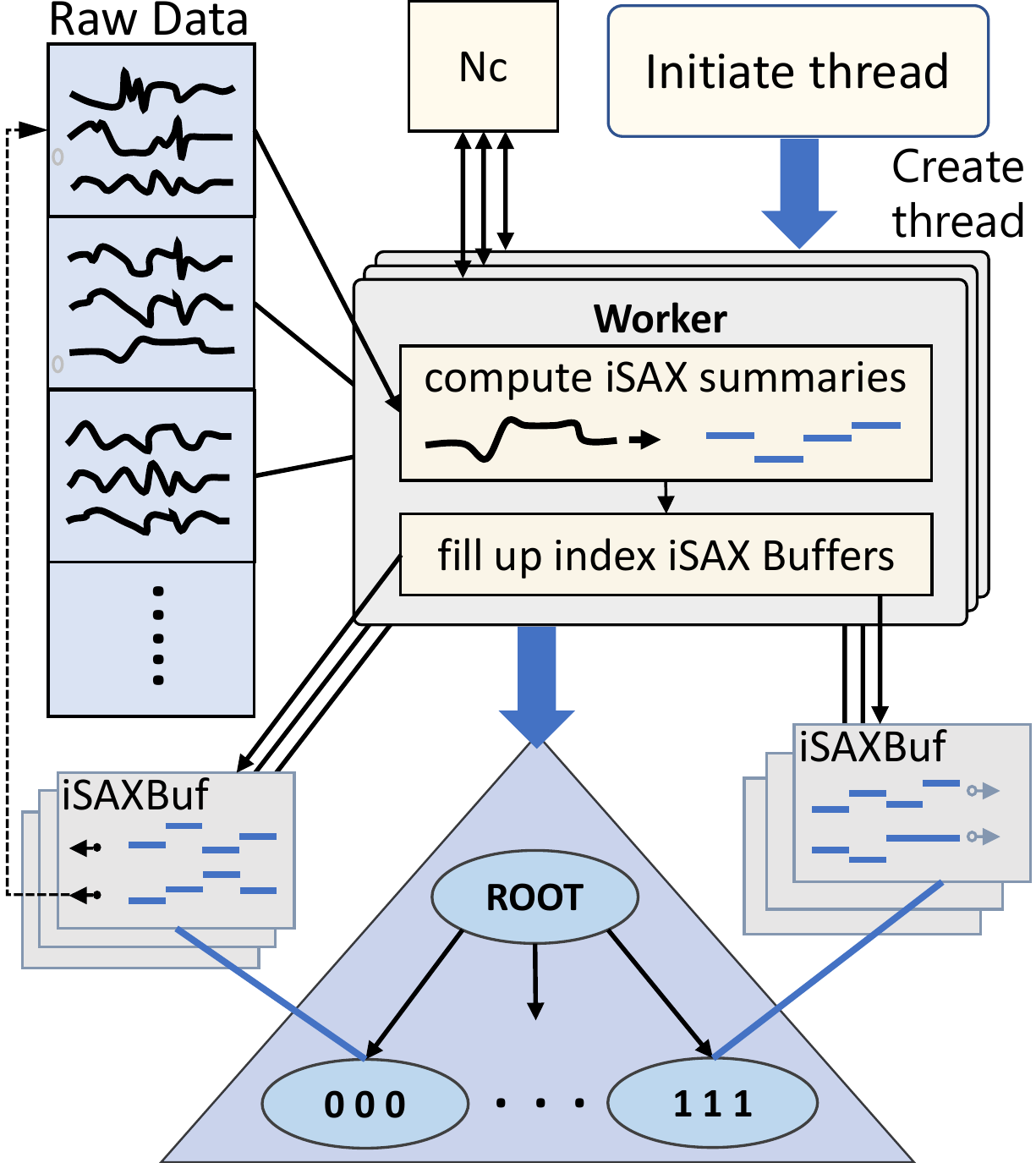}
	}
\hspace*{2cm}
	\subfigure[TreeConstruction\label{fig:inc2b}]{
		\includegraphics[page=1,width=0.73\columnwidth]{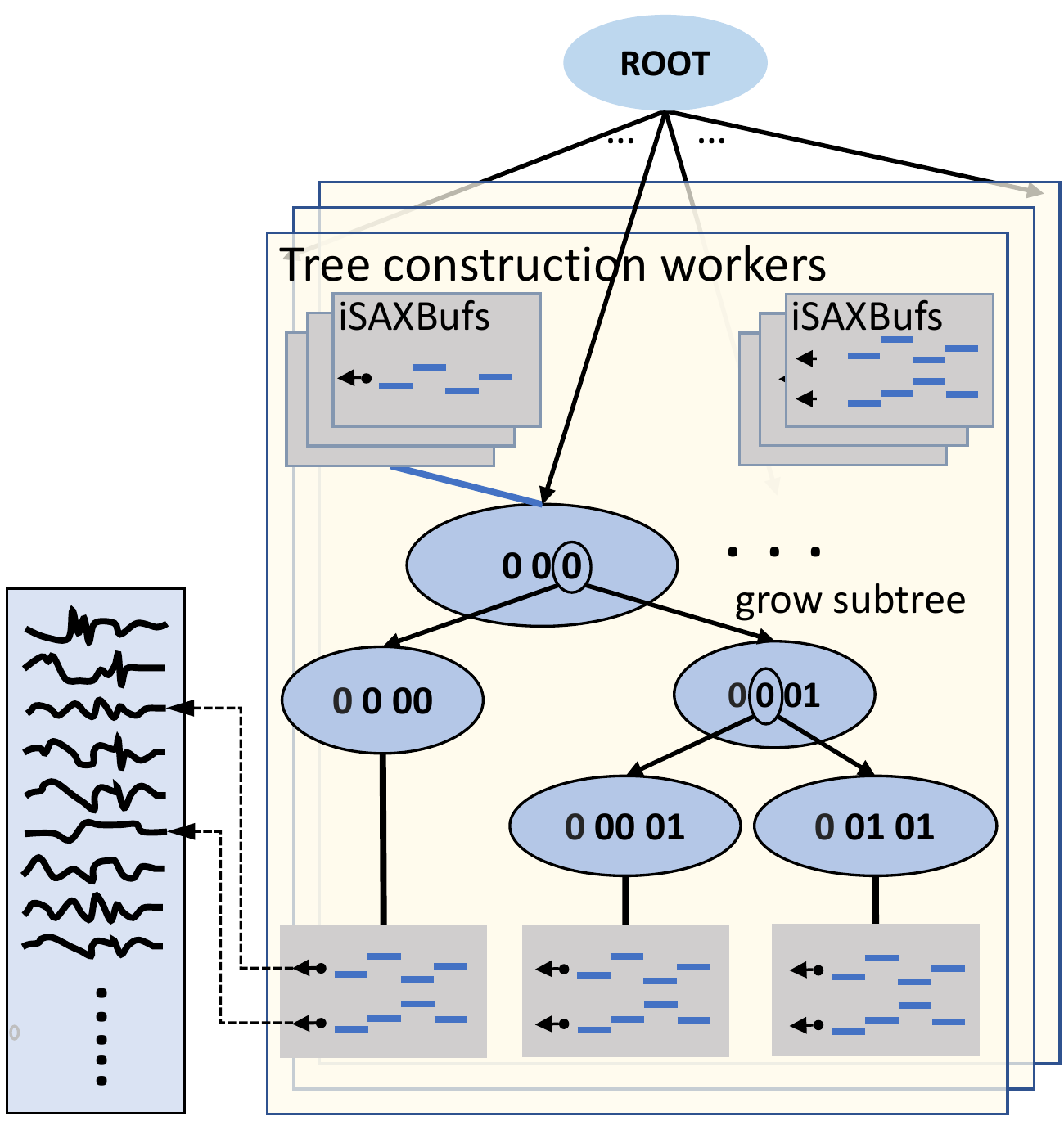}
	}
	\caption{Workflow and algorithms for MESSI index creation}
	\label{fig:inc2}
\end{figure*}

\subsection{Index Construction}
Algorithm~\ref{creatindex-bp} presents the pseudocode for the {\em initiator} thread.
The initiator creates $N_w$ index worker threads 
to execute the index construction phase (line~\ref{crb:worker}). 
As soon as these workers finish their execution, 
the initiator returns (line~\ref{crb:return}). 
We fix $N_w$ to be $24$ threads (Figure~\ref{fig:pRecBuf} in \textsection~\ref{sec:experiments} justifies this choice).
We assume that the $index$ variable is a structure (struct) containing the $RawData$ array,
all iSAX buffers, and a pointer to the root of the tree index.
Recall that MESSI splits $RawData$ into chunks of size $chunk\_size$. 
We assume that the size of $RawData$ is a multiple of
$chunk\_size$ (if not, standard padding techniques can be applied).

\begin{algorithm}[tb]
	{
		\SetAlgoLined
		\KwIn{\textbf{Index} $index$, \textbf{Integer} $N_w$, \textbf{Integer} $chunk\_size$	}	\vspace*{.1cm}
		\vspace*{.1cm}
		
		\For{$i$ $\leftarrow$ $0$ \emph{\KwTo} $N_w-1$} {
			create a thread to execute an instance of IndexWorker($index$, $chunk\_size$,$i$, $N_w$)\label{crb:worker}\; 		
		}
               wait for all these threads to finish their execution\;\label{crb:return}

	}
	\caption{$CreateIndex$}
	\label{creatindex-bp}
\end{algorithm}

\begin{algorithm}[tb]
	{
		\SetAlgoLined
		\KwIn{\textbf{Index} $index$, \textbf{Integer} $chunk\_size$, \textbf{Integer} $pid$, \textbf{Integer} $N_w$}
		\vspace*{.1cm}
		CalculateiSAXSummaries($index$, $chunk\_size$,$pid$)\;\label{iw:callbl}
		barrier to synchronize the IndexWorkers with one another\;\label{iw:barrier}
		TreeConstruction($index$, $N_w$)\;\label{iw:calltc}
		exit()\;
	}
\caption{$IndexWorker$}
\label{indexworker}
\end{algorithm}

The pseudocode for the index workers is in Algorithm~\ref{indexworker}. 
The workers first call the $CalculateiSAXSummaries$ function (line~\ref{iw:callbl})
to calculate the iSAX summaries of the raw data series and store
them in the appropriate iSAX buffers. 
As soon as the iSAX summaries of all the raw data series have been computed
(line~\ref{iw:barrier}), 
the workers call $TreeConstruction$  
to construct the index tree. 

The pseudocode of $CalculateiSAXSummaries$ is shown in Algorithm~\ref{bulkloading-bp} 
and is schematically illustrated in Figure~\ref{fig:inc2a}. 
Each index worker repeatedly does the following. 
It first performs a Fetch\&Inc to get assigned a chunk of raw data series to work on
(line~\ref{al2b:part}). Then, it calculates the offset in the $RawData$ array that this
chunk resides (line~\ref{al2b:initb}) and starts processing the relevant data series (line~\ref{al2b:for}). 
For each of them, it computes its iSAX summary by calling the ConvertToiSAX function (line~\ref{al2b:cover}),
and stores the result in the appropriate iSAX buffer of $index$ (lines~\ref{al2b:mask}-\ref{al2b:insert}). 
Recall that each iSAX buffer is split into $N_w$ parts, one for each thread;
therefore,\\ $index.iSAXbuffer$ is a two dimensional array. 

Each part of an iSAX buffer
is allocated dynamically when the first element to be stored in it 
is produced.  
The size of each part has an initial small value (5 series in this work, as we discuss in the experimental evaluation) 
and it is adjusted dynamically based on how many elements are inserted in it 
(by doubling its size each time). 
%
We note that we also tried a design of MESSI with no iSAX buffers, but this led to slower performance (due to the worse cache locality). 
Thus, we do not discuss this alternative further.

\begin{algorithm}[tb]
	{
		\SetAlgoLined
		\KwIn{\textbf{Index} $index$, \textbf{Integer} $chunk\_size$, \textbf{Integer} $pid$}
		\textbf{Shared integer} $F_c=0$\;  
		\vspace*{.1cm}
		\While{(TRUE)\label{al2b:loop}}
		{
			$b\leftarrow${\em Atomically} fetch and increment $F_c$\;\label{al2b:part}
			$b$ = $b*chunk\_size$\;\label{al2b:initb}
			\textbf{if} ($b \geq$ size of the $index.RawData$ array) \textbf{then} break\label{al2b:if} \;
			\For{$j$ $\leftarrow$ $b$ \emph{\KwTo} $b + chunk\_size$\label{al2b:for}}
			{
				$isax$ = $ConvertToiSAX$($index.RawData[j]$)\;\label{al2b:cover}
    			        $\ell$ = find appropriate root subtree where $isax$ must be stored\;\label{al2b:mask}
				$index.iSAXbuf[\ell][pid] = \langle isax, j \rangle$\;\label{al2b:insert}
			}
		}
	}
\caption{$CalculateiSAXSummaries$}
\label{bulkloading-bp}
\end{algorithm}

\begin{algorithm}[tb]
	{
		\SetAlgoLined
		\KwIn{\textbf{Index} $index$, \textbf{Integer} $N_w$} 
		\vspace*{.1cm}
		\textbf{Shared integer} $F_{b}=0$\;  
		\vspace*{.1cm}
		\While{(TRUE)\label{conp:loop} }
		{
			$b\leftarrow${\em Atomically} fetch and increment $F_{b}$\;\label{al3:gotnode}
			\textbf{if} ($b \geq 2^w$) \textbf{then} break \label{conp:break} \tcp*{\footnotesize{root has <= $2^w$ children}} 
			\For{$j$ $\leftarrow$  $0$ \emph{\KwTo} $N_w$\label{conp:lay}}
			{
				\For{\textbf{{\em every}} $\langle isax, pos \rangle$ {\em pair} $\in index.iSAXbuf[b][j]$\label{conp:passts}}
				{
					$targetLeaf \leftarrow$ Leaf of $index$ tree to insert $\langle isax, pos \rangle$\;\label{conp:insertinleaf}
					\While{$targetLeaf$ {\em is full}}
					{
						SplitNode($targetLeaf$)\;\label{al3:split}
						$targetLeaf \leftarrow$ New leaf to insert $\langle isax, pos \rangle$\;
					}
					Insert $\langle isax, pos \rangle$ in $targetLeaf$\;\label{al3:output}
				}
			}
		}
	}
\caption{$TreeConstruction$}
\label{construction-p}
\end{algorithm}

As soon as the computation of the iSAX summaries is over, 
each index worker starts executing the $TreeConstruction$ function.
Algorithm~\ref{construction-p} shows the pseudocode for this function and 
Figure~\ref{fig:inc2b} schematically describes how it works.
In $TreeConstruction$, a worker repeatedly executes the following actions.
It accesses $F_b$ (using Fetch\&Inc) to get assigned an iSAX buffer to work on (line~\ref{al3:gotnode}). 
Then, it traverses all parts of the assigned buffer (lines~\ref{conp:lay}-\ref{conp:passts}) 
and inserts every pair $\langle \mbox{iSAX summary}, \mbox{pointer to relevant data series} \rangle$ 
stored there in the index tree (line~\ref{conp:insertinleaf}-\ref{al3:output}). 
Recall that the iSAX summaries contained
in the same iSAX buffer will be stored in the same subtree of the index tree.
So, no synchronization is needed among the index workers during this process. 
If a tree worker finishes its work on a subtree, 
a new iSAX buffer is (repeatedly) assigned to it, until all iSAX buffers have been processed. 

\subsection{Query Answering with Euclidean Distance}

\begin{algorithm}[tb]
	{	
		\SetAlgoLined
		\textbf{Shared float} $BSF$\;
		\KwIn{\textbf{QuerySeries} $QDS$, \textbf{Index} $index$, \textbf{Integer} $N_q$ }
		QDS\_iSAX = calculate iSAX summary for QDS\;\label{eshq:calcuqisax}
		BSF = approxSearch($QDS\_iSAX$, $index$)\;\label{eshq:appro}
		\For{$i$ $\leftarrow$ $0$ \emph{\KwTo} $N_q-1$\label{eshq:scq}} {
			$queue[i]$ = Initialize the $i$th priority queue\;
		}\label{eshq:ecq}

		\For{$i$ $\leftarrow$ $0$ \emph{\KwTo} $N_s-1$\label{eshq:scw}} {
			create a thread to execute an instance of  		
			SearchWorker($QDS$, $index$, $queue[]$, $i$, $N_q$)\;\label{eshq:cw}
		}\label{eshq:ecw}
		Wait for all threads to finish\;\label{eshq:finish}
		\Return ($BSF$)\;\label{eshq:return}
	}
\caption{$Exact Search$}
\label{eshq}
\end{algorithm}

The pseudocode for executing an exact search query with Euclidean distance is shown 
in Algorithm~\ref{eshq}. 
We first calculate the iSAX summary of the query (line~\ref{eshq:calcuqisax}),
and execute an approximate search (line~\ref{eshq:appro}) to find 
the initial value of BSF, i.e., a first upper bound
on the actual distance between the query and the series
indexed by the tree. This process is illustrated in Figure~\ref{fig:MESSIappro}.  

\begin{figure*}[t]
	\subfigure[Approximate search and first BSF\label{fig:MESSIappro}]
	{
		\includegraphics[page=1,height=5.3cm]{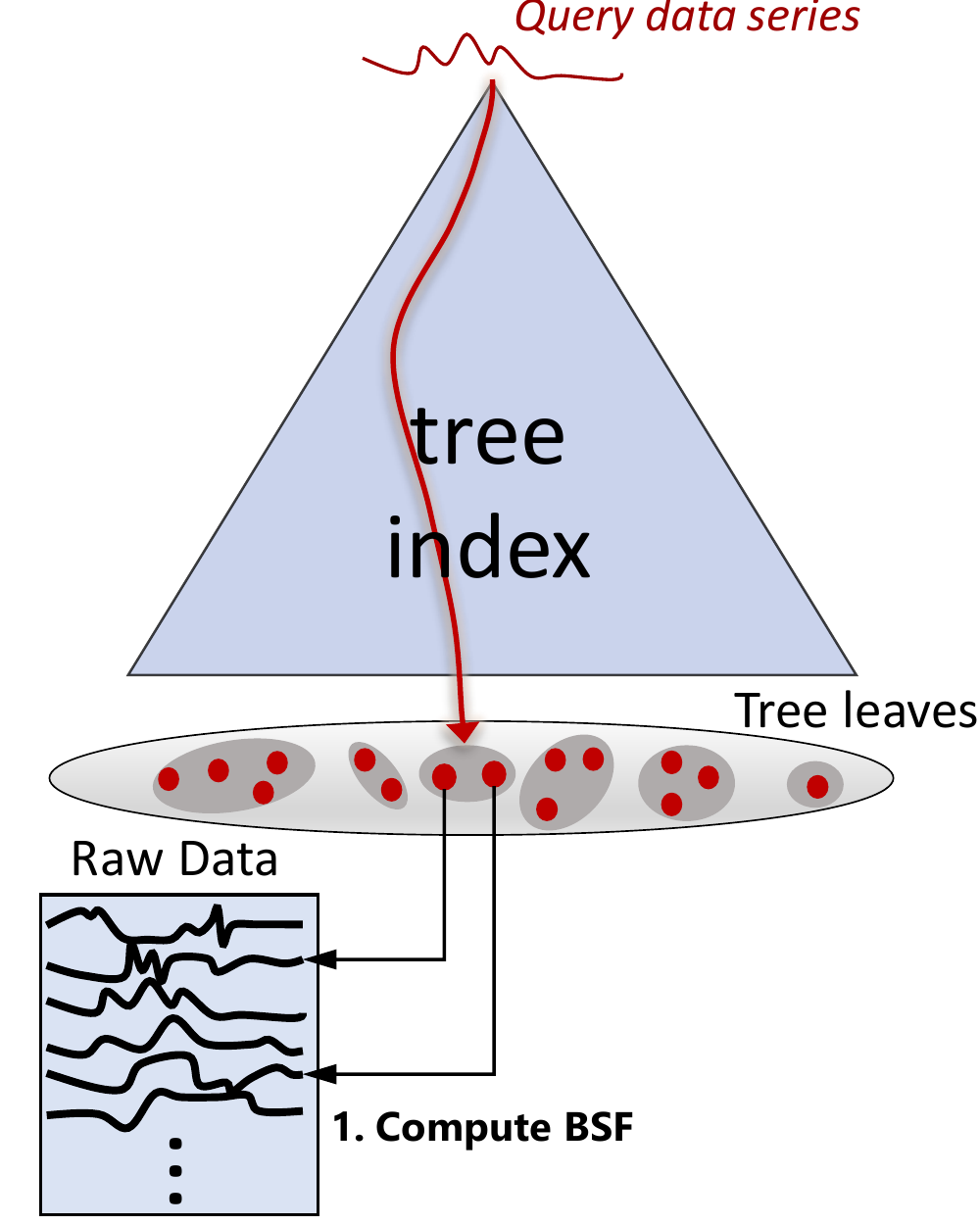}
	}
	\hspace*{0.7cm}
	\subfigure[Tree traversal and priority queues\label{fig:MESSIphq}]
	{
		\includegraphics[page=1,height=5.3cm]{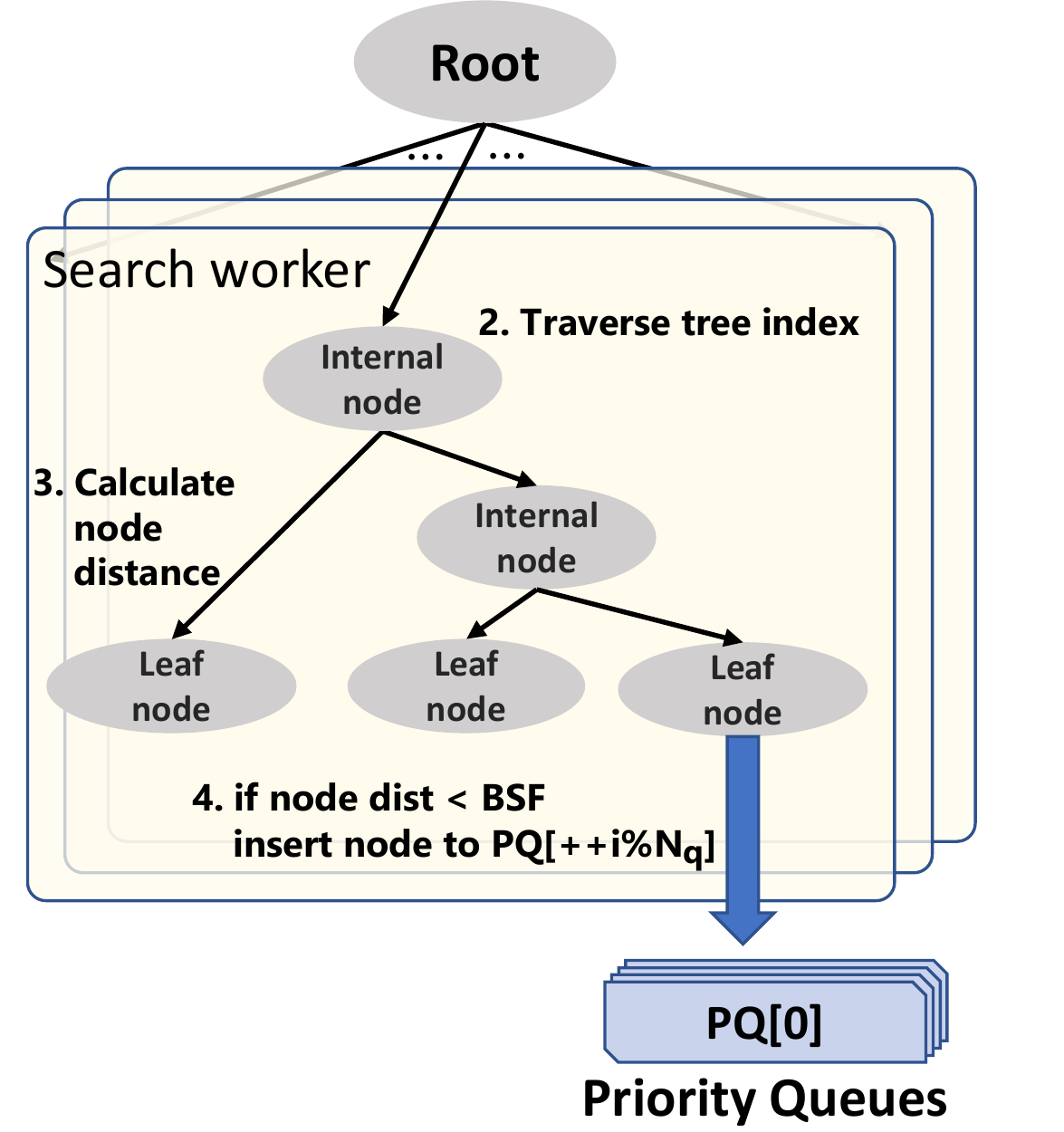}
	}
	\subfigure[Node distance calculation from priority queues\label{fig:MESSIpophq}]
	{
		\includegraphics[page=1,height=5.3cm]{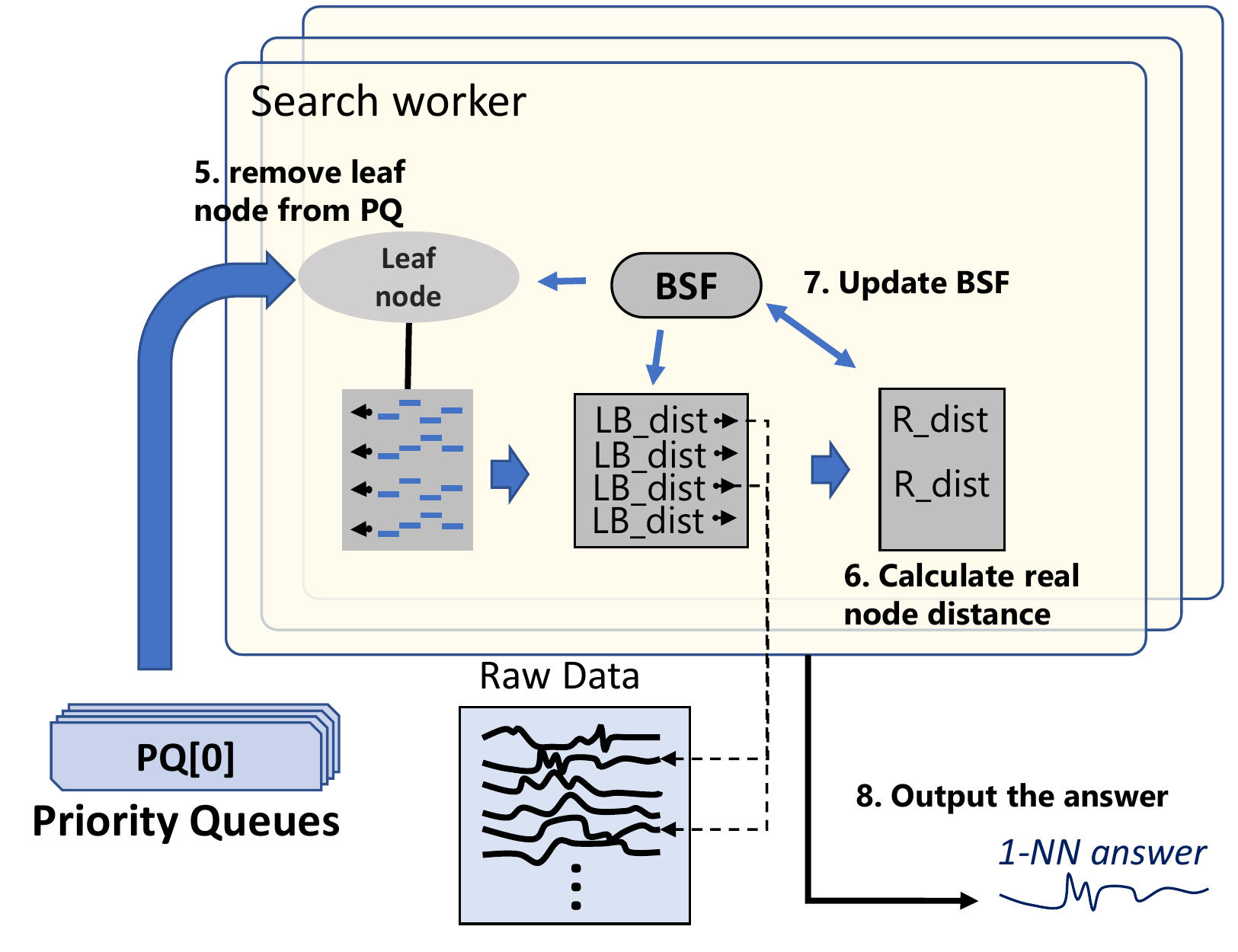}
	}
	\caption{Workflow and algorithms for MESSI query answering}
	\label{fig:MESSIquerychart}
\end{figure*}

During a search query, the index tree is traversed and the distance of the iSAX summary
of each of the visited nodes to the iSAX summary of the query is calculated. 
If the distance of 
the iSAX summary of a node, $nd$, to the query iSAX summary is higher than BSF, then we are certain that the distances
of all data series indexed by the subtree rooted at $nd$
are higher than BSF. 
So, the entire subtree can be pruned. 
Otherwise, we go down the subtree, and the leaves with a distance to the query smaller 
than the BSF, are inserted in the priority queue.

The technique of using priority queues 
maximizes the pruning degree, thus resulting in a relatively 
small number of raw data series whose real distance to the query series 
must be calculated. As a side effect, BSF converges fast
to the correct value. Thus, the number of iSAX summaries 
that are tested against the iSAX summary of the query series 
is also reduced. 

Algorithm~\ref{eshq} creates $N_s = 48$ threads, called the {\em search workers} (lines~\ref{eshq:scw}-\ref{eshq:ecw}), 
which perform the computation described above by calling $SearchWorker$. 
It also creates $N_q \geq 1$ priority queues (lines~\ref{eshq:scq}-\ref{eshq:ecq}),
where the search workers place those data series that are potential candidates
for real distance calculation. After all search workers have finished (line~\ref{eshq:finish}), 
$ExactSearch$ returns the current value of $BSF$ (line~\ref{eshq:return}). 

We have experimented with two different settings regarding the number of 
priority queues, $N_q$, that the search workers use. 
The first, called {\em Single Queue} ($SQ$), 
refers to $N_q = 1$, whereas the second focuses in the Multiple-Queue ($MQ$) case where $N_q > 1$.
Using a single shared queue imposes a high synchronization overhead, 
whereas using a local queue per thread results in load imbalance, 
since, depending on the workload, the size of the different queues may vary significantly. 
Thus, we choose to use $N_q$ shared queues, 
where $N_q$ is a fixed number
(in our analysis $N_q=24$, as experiments showed that this is the
best choice).

The pseudocode of search workers is shown in Algorithm~\ref{eshqw},
and the work they perform 
in Figures~\ref{fig:MESSIphq} and~\ref{fig:MESSIpophq}.
At each point in time, each thread works on a single queue. 
Initially, each queue is shared 
by two threads. 
Each search worker first identifies the queue where it will 
perform its first insertion (line~\ref{eshqw:startq}). 
Then, it repeatedly 
chooses (using Fetch\&Inc)
a root subtree of the index tree to work on 
by calling $TraverseRootSubtree$ (line~\ref{eshqw:traverse}). 
After all root subtrees have been processed (line~\ref{eshqw:barrier}), 
it repeatedly chooses 
a priority queue
(lines~\ref{eshqw:sgetqueue},~\ref{eshqw:egetqueue})
and works on it by calling\\ $ProcessQueue$ (line~\ref{eshqw:getqueue}).  
Each element of the $queue$ array has a $finished$ field indicating if the processing of the corresponding priority queue has 
finished. As soon as a search worker determines that all priority queues
have been processed (line~\ref{eshqw:finish}), it terminates.

\begin{algorithm}[tb]
	{
		\SetAlgoLined
		\KwIn{\textbf{QuerySeries} $QDS$, \textbf{Index} $index$, \textbf{Queue} $queue[]$, \textbf{Integer} $pid$, \textbf{Integer} $N_q$}
		\textbf{Shared integer} $N_{b}=0$\;  
		\vspace*{.1cm}
		$q = pid \mod  N_q$\;\label{eshqw:startq}
		\vspace*{.1cm}
		\While{(TRUE)\label{eshqw:loop}}
		{
			$i\leftarrow${\em Atomically} fetch and increment $N_{b}$\label{eshqw:ato}\;
			\textbf{if} ($i \geq 2^w$) \textbf{then} break\label{eshqw:break}\;
			$TraverseRootSubtree$($QDS$, $index.rootnode[i]$, $queue[]$, $\&q$, $N_q$)\;\label{eshqw:traverse}
			
		}
		\vspace*{.2cm}
		Barrier to synchronize the search workers with one another;\label{eshqw:barrier}\\
                $q = pid \mod N_q$\;
		\vspace*{.2cm}
		
		\While{(true)\label{eshqw:sgetqueue}}
		{
			$ProcessQueue(QDS, index, queue[q])$\label{eshqw:getqueue}\;
			\If{all queue[].finished=true}
			{
			break\;	\label{eshqw:finish}
			}
			$q \leftarrow$ index such that $queue[q]$ has not been processed yet\;
		}\label{eshqw:egetqueue}
	}
\caption{$SearchWorker$}
\label{eshqw}
\end{algorithm}

We continue to describe the pseudocode for \\
$TraverseRootSubtree$,
which is presented in Algorithm~\ref{TraverseRootSubtree} and illustrated in Figure~\ref{fig:MESSIphq}. 
$TraverseRootSubtree$ is recursive. On each internal node, $nd$, it checks whether the (lower bound) 
distance of the iSAX summary of $nd$ to the raw values of the query (line~\ref{insrq:mindist})
is smaller than the current $BSF$, and if it is,
it examines the two subtrees of the node using recursion (lines~\ref{tr:insl}-\ref{tr:insr}). 
If the traversed node is a leaf node and its distance to the 
iSAX summary of the query series is smaller than the current BSF (lines~\ref{tr:sleaf}-\ref{tr:eleaf}), it places it in the appropriate
priority queue (line~\ref{tr:insert}). 
Recall that the priority queues
are accessed in a round-robin fashion (line~\ref{tr:rq}). 
This strategy maintains the size of the queues balanced,
and reduces the synchronization cost of 
node insertions to the queues. 
We implement this strategy by (1) passing a pointer 
to the local variable $q$ of $SearchWorker$
as an argument to \\$TraverseRootSubtree$, 
(2) using the current value of $q$ for choosing the next queue
to perform an insertion (line~\ref{tr:insert}), and (3) updating the value
of $q$ (line~\ref{tr:rq}). 
Each queue
may be accessed by more than one threads, so a lock per queue is
used to protect its concurrent access by multiple threads.

\begin{algorithm}[tb]
	{
		\SetAlgoLined
		\KwIn{\textbf{QuerySeries} $QDS$, \textbf{Node} $node$, \textbf{queue} $queue[]$, \textbf{Integer} $*pq$, \textbf{Integer} $N_q$}
		\vspace*{.2cm}
		$nodedist$ = FindDist($QDS$, $node$)\;\label{insrq:mindist}
		\uIf{$nodedist$ $>$ $BSF$\label{tr:ifnodedist}} 
		{
			break\;
		}
		\uElseIf{$node$ is a leaf\label{tr:sleaf}} 
		{
			acquire $queue[*pq]$ lock\label{tr:lq}\;
			Put $node$ in $queue[*pq]$ with priority $nodedist$\;\label{tr:insert}
			release $queue[*pq]$ lock\label{tr:ulq}\;
			\emph{// next time, insert in the subsequent queue\\}
			$*pq\leftarrow (*pq+1) \mod N_q$\;\label{tr:rq}
		}\label{tr:eleaf}
		\Else
		{
			TraverseRootSubtree($node.leftChild,queue[],pq,N_q$)\;\label{tr:insl}
			TraverseRootSubtree($node.rightChild,queue[],pq,N_q$)\label{tr:insr}
		}
	}
	\caption{$TraverseRootSubtree$}
	\label{TraverseRootSubtree}
\end{algorithm}

We next describe how $ProcessQueue$ works (see Algorithm~\ref{ProcessQueue} and Figure~\ref{fig:MESSIpophq}). 
The search worker repeatedly 
removes the (leaf) node, $nd$, with the highest priority 
from the priority queue, 
and checks whether the corresponding distance stored in the queue 
is still less than the BSF. 
We do so, because the
BSF may have changed since the time that the leaf node was inserted in the priority queue. 
If the distance is less than the BSF, then $CalculateRealDistance$ (line~\ref{process:rd1}) is called to identify if any series in the leaf node (pointed to by $nd$)
has a real distance to the query that is smaller than the current BSF. 
If we discover such a series (line~\ref{process:if}), 
$BSF$ is updated to the new value (line~\ref{process:upbsf}). 
We use a lock to protect BSF from concurrent update efforts (lines~\ref{process:loc},~\ref{process:unloc}).
Previous experiments showed that the initial value of BSF is very close to its final value~\cite{gogolou2019progressive,DBLP:conf/sigmod/GogolouTEBP20}. 
Indeed, in our experiments, the BSF is updated only
10-12 times (on average) per query. 
So, the synchronization cost for updating the BSF is negligible. 

\begin{algorithm}[tb]
	{	
		\SetAlgoLined
		\KwIn{\textbf{QuerySeries} $QDS$, \textbf{Index} $index$, \textbf{Queue} $Q$}

		\While{TRUE)\label{process:popnode}}
		{
			acquire Q's lock\label{process:dloc}\;
			$node$ = DeleteMin($Q$)\label{process:dlmin}\;
			release Q's lock\label{process:duloc}\;
			\If{node == NULL} { return\;}
			\uIf{$node.dist$ $<$ $BSF$\label{process:1if}}
			{
				$realDist$ = CalculateRealDistance($QDS$, $index$, $node$)\label{process:rd1}\;
				\If{\textbf{$realDist$} $<$ $BSF$\label{process:if}} 
				{
					acquire $BSFLock$\;\label{process:loc}
					\If{$realDist < BSF$\label{process:2if}} {
						$BSF$ = $realDist$\;\label{process:upbsf} }
					release $BSFLock$\;\label{process:unloc}
				}
			}
			\Else
			{\label{process:else}
				$Q.finished$ = true\;
				return\;	
			}	
		}
	}
	\caption{$ProcessQueue$}
\label{ProcessQueue}
\end{algorithm}

$CalculateRealDistance$ is shown in Algorithm~\ref{realdist}. 
Note that both the lower bounding (line~\ref{rd:ld}) and the real (line~\ref{rd:rd}) distance calculations use SIMD~\cite{peng2018paris}. 
However, this does not lead to the same significant
impact in performance as in ParIS+. 
This is because pruning is much more effective in MESSI for two reasons: (i) MESSI performs much less lower bounding distance calculations since many of them are pruned during the traversal of the tree 
(see Algorithm~\ref{TraverseRootSubtree}); 
(ii) MESSI also performs a smaller number of real distance calculations since examining the raw data series in the order defined by the priority queue (see Algorithm~\ref{ProcessQueue}), rather than in the order of the raw file that ParIS+ uses, means that the \emph{BSF} gets updated earlier and converges earlier to the value of the nearest neighbor, leading to better pruning.

\begin{algorithm}[tb]
{
	\SetAlgoLined
	\KwIn{\textbf{QuerySeries} $QDS$, \textbf{Index} $index$, \textbf{node} $node$, \textbf{float} $BSF$}

	\For{every ($isax$, $pos$) pair $\in$ $node$} 
	{
		\If{$LowerBound\_SIMD$($QDS$, $isax$) $<$ $BSF$\label{rd:ld}}
		{
			$dist = RealDist\_SIMD(index.RawData[pos],QDS)$\;\label{rd:rd}
			\If{$dist < BSF$}
			{
				 $BSF = dist$\;
			}
		}
	}
	\Return{($BSF$)}
}
\caption{$CalculateRealDistance$}
\label{realdist}
\end{algorithm}

\subsection{Query Answering with Dynamic Time Warping}
Not only can MESSI accelerate similarity search based on Euclidean distance, 
but it also can be used to perform similarity search using the Dynamic Time Warping (DTW) distance.
No changes are required in the index structure for this: the index we build can answer both Euclidean and DTW similarity search queries. 
Supporting DTW queries requires modifying the query answering algorithm only, and using LB\_Keogh~\cite{keogh2005exact}, which is a tight lower bound of the DTW distance\footnote{We note that other lower bounds for DTW can be used as well, such as LB\_Improved~\cite{DBLP:journals/pr/Lemire09}. Even though LB\_Improved can produce tighter bounds, in our experiments it also resulted in higher query answering times due to the additional computations it involves.}.
Recall that a lower bound for the DTW distance between the query and a candidate series can be computed by considering the distances between the corresponding points of the candidate series and the points of the LB\_Keogh envelope of the query (see Figure~\ref{fig:envelope}; if some points of the candidate fall inside the query envelope, then their distance is zero). 

\begin{figure*}[tb]
	\centering
	\hspace*{-0.3cm}
	\includegraphics[page=1,width=0.64\textwidth]{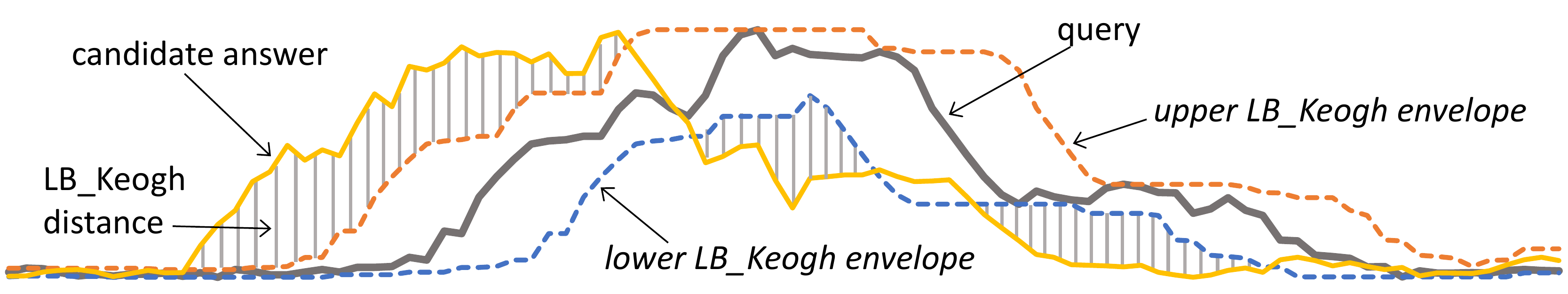}
	\caption{Envelope (dotted lines representing the upper, $U$, and lower, $L$, series that define the envelope) of query $Q$ (thick solid black line) with warping size 10\%. Vertical lines represent LB\_Keogh lower bound between query envelope and candidate answer $C$ (thin solid yellow line).}
	\label{fig:envelope}
\end{figure*}

Assuming the reach, or allowed range of (the constrained) warping, is $r$, we define two new sequences, $U$ and $L$, corresponding to the upper and lower parts of the LB\_Keogh envelope: 
$$
\small
\\U_i=\max \left(q_{i-r}:q_{i+r}\right)\\
$$
$$
\small
\\L_i=\min \left(q_{i-r}:q_{i+r}\right)\\
$$
Having defined $U$ and $L$, we now use them to define a lower bounding measure
for DTW between a query sequence $Q$ and a candidate answer $C$~\cite{keogh2005exact}:
$$
\small
\\LB\_Keogh(Q,C)=\sqrt{ \sum_{i=1}^{n}
	\left\{
	\begin{array}{lr}
		(c_i-U_i)^2 \;\; if \; c_i>U_i  \\
		(c_i-L_i)^2 \;\; if \; c_i<L_i\\
		0 \qquad\qquad otherwise &  
	\end{array}
\right.
}\\
$$

Intuitively, when the query series arrives, we compute the LB\_Keogh envelope of this series, as shown in Figure~\ref{fig:envelope}. 
We then probe the index using the envelope as the query (instead of the series itself). 
The distances we compute using the LB\_Keogh envelope are guaranteed to be lower bounds of the true DTW distances~\cite{keogh2005exact}. 
Therefore, this operation correctly prunes the search space, and (as in the case of Euclidean distance) we then simply need to remove the false positives by computing the DTW distance on the raw data values of the (small set of) candidate answers.
Overall, the process of query answering using DTW follows the same steps as those described in Algorithms~\ref{eshq}-\ref{ProcessQueue}, except that in line~\ref{eshq:calcuqisax} of Algorithm~\ref{eshq} we compute the PAA of the LB\_Keogh envelope of the query, in line~\ref{insrq:mindist} of Algorithm~\ref{TraverseRootSubtree} we compute the (lower bounding) distance between the query envelope PAA and the node iSAX summarization, and in line~\ref{process:rd1} of Algorithm~\ref{ProcessQueue} we call the function that computes the DTW real distance (Algorithm~\ref{realdtwdist}).

\begin{algorithm}[tb]
	{		
		\SetAlgoLined
		\KwIn{\textbf{QuerySeries} $QDS$, \textbf{Envelope of QuerySeries} $EQDS$, \textbf{Index} $index$, \textbf{node} $node$, \textbf{float} $BSF$}
		\For{every ($isax$, $pos$) pair $\in$ $node$} 
		{
			\If{$LowerBound\_DTW\_SIMD$($EQDS$, $isax$) $<$ $BSF$\label{rdtwd:ld}}
			{
				\If{$LB\_Keogh(EQDS,index.RawData[pos])$ $<$ $BSF$\label{rdtwd:lbkd}}
				{
					$dist = RealDist\_SIMD(index.RawData[pos],QDS)$\;\label{rdtwd:rd}
					\If{$dist < BSF$}
					{
						$BSF = dist$\;
					}
				}
			}
		}
		\Return{($BSF$)}
	}
	\caption{$CalculateRealDistanceDTW$}
	\label{realdtwdist}
\end{algorithm}}

\begin{figure*}[tb]
	\centering
	\hspace*{-0.3cm}
	\includegraphics[page=1,width=0.7\textwidth]{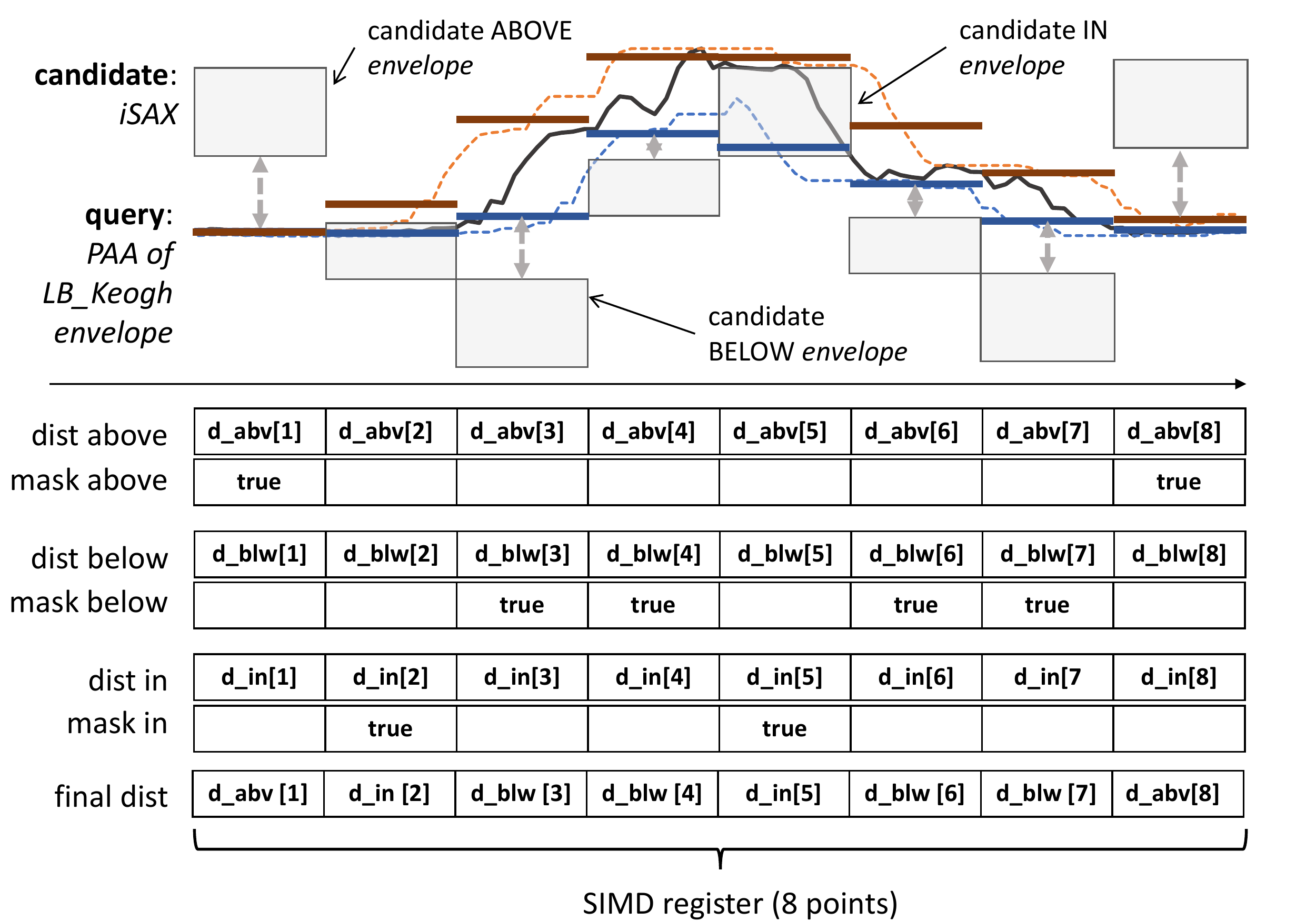}
	\caption{SIMD conditional branch DTW lower bound distance calculation.}
	\label{fig:dtwsimd}
\end{figure*}

More specifically, we first perform an approximate search in order to get an initial solution that is close to the actual answer, which will serve as our BSF.
In order to prune the index tree, we then calculate the lower bound distance between the query envelope PAA representation  and the iSAX summarization of the leaf nodes (see Figure~\ref{fig:dtwsimd} \emph{top}).
We insert the leaves that we cannot prune (i.e., the DTW lower bound is less than the BSF) in the priority queue. 
When MESSI pops a leaf from the priority queue, it calculates the DTW lower bound distance between the query envelope PAA representation and the iSAX summarization of each series in the leaf (Algorithm~\ref{realdtwdist}, line~\ref{rdtwd:ld}). 
For the series that survive this second filter, we have to access the raw data.
We start by computing the DTW lower bound distance between the query envelope raw values and the raw values of the series (Algorithm~\ref{realdtwdist}, line~\ref{rdtwd:lbkd}).
If this step cannot prune the series, either, then we finally compute the true DTW distance between the raw values of the query and the series (Algorithm~\ref{realdtwdist}, line~\ref{rdtwd:rd}).

In order to speed up the execution of the DTW lower bound distance calculations, we develop a SIMD solution.
Note that, in contrast to the simple case of a Euclidean distance calculation on the real data, developing a SIMD solution for the LB\_Keogh lower bound is not straight-forward.
The Euclidean distance calculated on the real values of two series involves exactly the same operations (i.e., first a subtraction and then a power of two operation) for all the points of the series. 
This leads to a simple SIMD solution, where the entire SIMD register performs the same operations, all useful and necessary for the final result.
On the other hand, the algorithm for computing the DTW lower bound involves branching. 
As we discuss below, we need to perform different operations for the candidate series segments, depending on whether their values are larger, smaller, or lie within the the LB\_Keogh envelope values. 
We therefore need to translate these branches of the operation into an efficient SIMD implementation.

Our DTW lower bound SIMD solution is illustrated in Figure~\ref{fig:dtwsimd}.
In the top part of this figure, boxes represent the iSAX summary of each segment of the candidate series,
and red and blue horizontal lines represent the PAA representations of each segment
of the upper and the lower LB\_Keogh envelope of the query, respectively. 
The bottom part of Figure~\ref{fig:dtwsimd} illustrates how the lower bound DTW distance between the iSAX summary of the candidate series 
and the PAA representations of  the query LB\_Keogh envelopes, is calculated using SIMD. 
Our architecture supports the computation of the lower bound DTW distances for eight of these segments concurrently. 

This per-segment computation  needs to capture three different cases,  
namely, the PAA representation of the candidate being \emph{ABOVE}, \emph{BELOW}, or \emph{IN} the query envelope PAA representation.
The code to run is different in each of these cases.
The ABOVE distance is calculated between the lower edge of the candidate series segment's iSAX box and the corresponding segment of the PAA representation of the \emph{Upper} LB\_Keogh envelope (shown as a red line in Figure~\ref{fig:dtwsimd}).
The BELOW distance is calculated between the upper edge of the iSAX box of the candidate series and the the PAA representation of the \emph{Lower} LB\_Keogh envelope (shown as a blue line in Figure~\ref{fig:dtwsimd}).
Finally, the IN distance is simply zero. 

Since we cannot know beforehand which of the above three cases is true in each case, we use SIMD to calculate all three distances for each segment. 
These are depicted in Figure~\ref{fig:dtwsimd} as \emph{dist above}, \emph{dist below}, and \emph{dist in}. 
We then 
choose the correct distance among these three distances. 
To do so, we use SIMD to compute three masks, setting the appropriate positions to \emph{true} depending on the observed situation, that is, depending on the position (\emph{ABOVE}, \emph{BELOW}, or \emph{IN}) of a candidate series' segment with respect to the corresponding segments of the LB\_Keogh envelope.
In Figure~\ref{fig:dtwsimd} for example, the first candidate iSAX representation is above the corresponding query envelope PAA, which means that only the \emph{ABOVE} mask will be \emph{true} for this position; consequently we will choose
the ABOVE distance value for this position of the SIMD vector. 

The final result is then computed using SIMD by summing up the right distances for each segment, i.e., those for which the corresponding mask position was set to \emph{true} (shown as \emph{final dist} in Figure~\ref{fig:dtwsimd} bottom). 
This operation is efficiently executed by using the appropriate SIMD instructions (AVX, AVX2 and SSE3)~\cite{coorporation2009intel}. 

\subsection{Complexity Analysis}

We now provide a best- and worst-case
time analysis, which contrasts the time needed in a concurrent setting with that required in a single-thread environment.
That is, we compare the performance of MESSI when multiple workers are
active with its performance when just a single thread is active. 
Note that the best- and worst-case scenaria are mainly (but not exclusively) driven by the data characteristics: as we detail below, different datasets may lead to index trees with significantly different properties.

\noindent{\bf [Index Construction]} 
	Index construction in MESSI is comprised of two phases.
	During the first phase (Algorithm~\ref{bulkloading-bp}, Figure~\ref{fig:inc2a}), 
	the index workers calculate the iSAX summaries of the raw data series
	and store them into the iSAX buffers. During the second phase (Algorithm~\ref{construction-p}, Figure~\ref{fig:inc2b}), 
	the index workers process the iSAX buffers and build the index tree. We analyze each of these phases separately. 
	Assume that the time needed by a single thread to execute phase 1 is $T_1$, the time needed to execute phase 2 is $T_2$,
	and the total sequential time for index creation is $T = T_1 + T_2$. 
	In MESSI, every index worker processes about the same number of data series.
	Note that processing data series of the same length takes the same amount of time. 
	Given that each index worker works on its own part of the iSAX buffers,
	the amount of time spent to allocate the buffers in the concurrent setting 
	is at most a factor of $2$ larger than that needed in the sequential case. 
	Therefore, 
	if we exclude the contention that a worker experiences when accessing the Fetch\&Increment object,
	all workers require about the same amount of time to finish the first phase.
	
	\noindent
	{\it \underline{Best Case}:} The best case scenario occurs when 
	threads experience no contention in accessing the Fetch\&Inc objects
	(i.e., no two threads access a Fetch\&Inc object at exactly the same time).
	Since the threads do not experience any contention when storing elements 
	in the iSAX buffers either, MESSI then requires $O(T_1/N_w)$ time for executing the first phase,
	which is optimal. 
	%
	We now focus on the second phase. In the best case scenario, it additionally holds that 
	the iSAX buffers are assigned to threads in a way that all threads finish the second stage
	at about the same time. 
	Note that 
	this does not necessarily require that all subtrees contain the same number of nodes (e.g., if bigger subtrees are processed earlier than smaller ones, then the processing of big subtrees overlaps with that of smaller ones). 
	Then, MESSI requires $O(T_2/N_w)$ time to execute phase 2, which is also optimal. 
	Therefore, in the best case, MESSI exhibits optimal speedup 
	for index creation.
	
	\noindent
	{\it \underline{Worst Case}:} In the worst case, all $N_w$ workers access the Fetch\&Inc
	objects at the same time. We assume that the system serializes these accesses and the
	$i$th worker in this serialization order will get a response from the object after $i$ time units. 
	Therefore, each block of $N_w$ concurrent accesses to the Fetch\&Inc object (by all threads)
	adds a total of $N_w$ time units due to contention. 
	We will have $(N_c+ N_b)/N_w$ such blocks of accesses, where $N_c$ is the number of chunks
	and $N_b \leq 2^w$ is the number of root subtrees in the index tree.  In the worst case,
	MESSI will also experience the following: while a thread will be processing the last chunk  of the raw data array,
	all other threads will be sitting idle. The same will occur with the last subtree. So, if $T_c$ is the sequential time for processing a chunk
	and $T_{s}$ is the sequential time for processing the biggest subtree of the root, the 
	worst case execution time of MESSI will be $O((T-T_c-T_s)/N_w + T_c + T_{s} + N_c + N_b)$.
	In the extreme scenario, where all data series are stored in only one of the subtrees,
	this time may be no less than $T$ (the sequential time). 
	Note that this is a pathological case that would 
	happen when all series in the dataset are very similar to one another\footnote{In such a case, indexing and similarity search would not be useful anyways.}.

	\noindent
	{\bf [Query Answering]} 
	Query answering in MESSI is comprised of three phases.
	During the first phase, approximate search is executed. 
	During the second phase ({\em tree traversal}), 
	the search workers traverse the index tree and populate the priority queues. 
	During the third phase ({\em queue processing}), the search workers process the elements of the priority 
	queues to produce the final result.  We analyze each of these phases separately. 
	Assume that the time needed by a single thread to execute phase 1 is $T_1$, the time needed to execute phase 2 is $T_2$,
	the time needed to execute phase 3 is $T_3$, 
	and the total sequential time for the last two phases is $T = T_2 + T_3$. 
	The approximate search is executed by a single thread in MESSI and therefore
	this time is also $T_1$ in a multi-threaded environment. 
	We therefore focus on the other two phases.

	\noindent
	{\underline{Best case.}} For the tree traversal phase, the best case
	occurs when threads never access the Fetch\&Inc object concurrently
	and never find the lock of a queue taken. 
	Moreover, each thread must add the same number of nodes in the priority queues,
	so that all threads perform about the same amount of work. Then, 
	the time needed to execute phase 2 is $O(T_2/N_s)$. 
	For phase 3, the best case occurs when no thread has to ever wait on a lock
	and each thread performs about the same amount of computation. 
	Note that after acquiring a node from some queue, a thread has to perform computation
	(i.e., real distance calculations). These computations could be overlapping
	with the deletion of additional elements from the queue. 
	Thus, the time needed for phase 3 in the best case is $O(T_3/N_s)$. 
	Therefore, the total time is $O(T_1 + T/N_s)$.
	
	We observe that in the concurrent case, it may happen that the final value of
	BSF is reached faster than in the sequential case, since all threads update the value of BSF in parallel. 
	This may result in better pruning  
	than in the single-thread case, where the thread may process the subtree 
	(or the queue) that contains the node which results in the final value of BSF
	towards the end of the tree traversal (or the processing of the queues). 
	
	\noindent
	{\underline{Worst case.}}
	Let $T_a$ be the sequential time needed for performing 
	the insertions to the priority queues. Since both the cost for an insertion 
	and the cost for a deletion are logarithmic on the size of the priority queue,
	the sequential time needed for performing the deletions from the priority queue is 
	also in $\Theta(T_a)$. 
	Let $T_b$ be the sequential time needed for updating the BSF.
	Due to the use
	of locks, these times are still sequential in the concurrent setting. 
	Assuming queue locks~\cite{10.5555/2385452}, the steps needed to acquire or release a lock is $O(1)$.
	Note that the time a thread waits for the lock to be released is overlapping 
	with the critical sections of other threads, and therefore we do not count 
	waiting times on the locks.  
	
	Regarding the second phase, in the worst case, all $N_s$ workers access the Fetch\&Inc object
	at the same time. Thus, each block of $N_s$ concurrent accesses to the
	Fetch\&Inc object adds a total of $N_s$ time units (due to contention).
	We have $N_b/N_s$ such blocks of accesses, 
	thus resulting in a total cost of $O(N_b)$ time units. Therefore, the worst-case time
	for the second phase is $O((T_2-T_a)/N_s + T_a + N_b)$.
	Since $T_b \in O(T_a)$, the time to execute the third phase is $O((T_3-T_a - T_d)/N_s + T_a + T_d)$.
	Therefore, the total worst-case time is $O(T_1 + (T - T_a - T_d)/N_s + T_a + T_d + N_b)$.

\section{Proof of Correctness}
\label{secproof}
We note that in concurrent algorithms, the non-deterministic nature of parallel execution may lead to errors that are not detected during testing.
In this section, we provide proofs that the proposed algorithms for index creation and query answering always produce correct results, irrespective of the peculiarities of parallel execution.

\noindent
{\bf [Index construction phase] }
MESSI builds the tree index with minor synchronization,
i.e., by using two Fetch\&Inc objects and a barrier. 
This makes the correctness proof for index creation relatively simple.
However, for completeness, we include it below.

We say that a data series $S$ of the \emph{RawData} array {\em is processed} 
whenever a thread calculates its iSAX summary (i.e., executes line~\ref{al2b:cover}
of Algorithm~\ref{bulkloading-bp}). 
We say that a chunk of the \emph{RawData} array {\em is processed}
if a thread processes data series stored in it.


The use of the Fetch\&Inc object, $F_c$, in Algorithm~\ref{bulkloading-bp}, 
ensures that for each $0 \leq i \leq size/chunk\_size$,
(where $size$ is the size of the \emph{RawData} array and $chunk\_size$ is the size of each of its chunks),
there exists exactly one thread $p$ that gets number $i$ when accessing $F_c$ (line~\ref{al2b:loop} of Algorithm~\ref{bulkloading-bp}), and no thread other than $p$ processes chunk $i$. 
By inspection of the code of Algorithm~\ref{bulkloading-bp} (lines~\ref{al2b:loop} and~\ref{al2b:initb} and condition of 
the \emph{if} statement of line~\ref{al2b:if}), it follows that $F_c$ is accessed until its value becomes as large as the 
number of chunks of the \emph{RawData} array. Therefore, for every chunk of the \emph{RawData} array, 
there is exactly one thread to which this chunk is assigned.
Line~\ref{al2b:for} ensures that once a chunk is assigned to a thread, all the data series it contains are processed by this thread.
These (and the pseudocode) imply the following:

\begin{lemma}
\label{iSAX buffers}
For every data series, $S$, contained in the \emph{RawData} array, the following hold: (1) $S$ is processed exactly once, i.e. there is a single thread $p$ that calculates the iSAX summary for $S$;
(2) there exists exactly one iSAX buffer that contains an entry $e$ corresponding to $S$, 
and this entry appears in the part of the buffer that is assigned to $p$.
\end{lemma}

\remove{
\begin{proof}
The use of the Fetch\&Increment object, $F_c$, ensures that for each $i \geq 0$, $0 \leq i \leq size/chunk\_size$,
(where $size$ is the size of the \emph{RawData} array and $chunk\_size$ is the size of each of its chunks),
there exists exactly one thread $p$ that gets number $i$ when accessing $F_c$ (line~\ref{al2b:loop} of Algorithm~\ref{bulkloading-bp}), and no thread other than $p$ processes chunk $i$. 
By inspection of the code of Algorithm~\ref{bulkloading-bp} (lines~\ref{al2b:loop} and~\ref{al2b:initb} and condition of 
the \emph{if} statement of line~\ref{al2b:if}), it follows that $F_c$ is accessed until its value becomes as large as the 
number of chunks of the \emph{RawData} array. Therefore, for every chunk of the \emph{RawData} array, 
there is exactly one thread to which this chunk is assigned.
Line~\ref{al2b:for} ensures that once a chunk is assigned to a thread, all the data series it contains are processed by this thread. 
Thus, claim~1 follows. 

Claim~1 and lines~\ref{al2b:mask},~\ref{al2b:insert} (Algorithm~\ref{bulkloading-bp}) imply that claim~2 holds.
\end{proof}
}

We use Lemma~\ref{iSAX buffers} to argue that the constructed tree index is correct.

\begin{theorem}
\label{tree construction}
The data structure $T$ constructed by executing the IndexConstruction phase (Algorithms~\ref{bulkloading-bp} and~\ref{construction-p})
is a tree that contains a distinct element for every data series of the \emph{RawData} array
and no more elements. 
\end{theorem}

\begin{proof}
Initially, $T$ is a tree with a root node and $c\le2^w$ leaf children. 
By inspection of the code, additional elements can be added into $T$ 
by executing Algorithm~\ref{construction-p}.
The barrier on line~\ref{iw:barrier} of Algorithm~\ref{indexworker}
ensures that no thread starts inserting additional elements 
in $T$, as long as there exist threads that 
still process data series stored in the \emph{RawData} array (i.e., they still execute Algorithm~\ref{bulkloading-bp}). 
Thus, Lemma~\ref{iSAX buffers} implies that a thread
calls Algorithm~\ref{construction-p} only after the
iSAX summaries of all data series stored in the \emph{RawData} array 
have been placed in the iSAX buffers. 

Recall that the number of iSAX buffers is also $c$ (the same as the number of the root children of $T$).
The use of the Fetch\&Increment object, $F_b$, and lines~\ref{conp:loop} and~\ref{conp:break} ensure 
that for each $i$, $0 \leq i \leq c$,
exactly one thread $p$ gets number $i$ by accessing $F_b$
(i.e., by executing line~\ref{al3:gotnode} of Algorithm~\ref{construction-p}). 
Recall that every calculated iSAX summary is placed in 
the appropriate iSAX buffer, i.e., in  
the iSAX buffer that corresponds to the root subtree of $T$ in which the iSAX summary should be stored. 
Thus, iSAX buffer $i$ contains only those data series that 
are to be stored in $T$'s root subtree numbered $i$.
It follows that the task to build the 
entire subtree has been assigned solely to process $i$. 
So, different threads work on different
subtrees of $T$ (and no synchronization is needed between them).
It follows that $T$ ends up to be an index tree.


Lines~\ref{conp:loop}-\ref{conp:break}  ensure that all $c$ iSAX buffers will be examined.
The for loop of line~\ref{conp:lay} ensures that $p$ will examine all parts of iSAX buffer $i$,
and the for  loop of line~\ref{conp:passts} guarantees that all iSAX summaries stored in each of these parts
will be inserted in $T$ (lines~\ref{conp:insertinleaf}-\ref{al3:output}).
Therefore, Lemma~\ref{iSAX buffers} (claim~2) implies that 
the constructed tree contains a distinct element for every data series stored in the $RawData$ array
and no more elements.
\end{proof}

\noindent
{\bf Query Answering Phase.}
To argue that the response of a 1-NN query, $QR$, is correct, 
we need the following properties from~\cite{shieh2008sax}. 

\begin{property}
\label{iSAX property}
The distance between the PAA of $QR$ and the iSAX summary
of a node $nd$ of the index lower bounds the real distance 
between $QR$ and any of the series in the leaves of $nd$'s subtree.
\end{property}

\begin{property}
\label{iSAX property2}
Consider two leaf nodes $nd$ and $nd'$ of the index tree.
Let $d$ be the minimum real distance between $QR$ and any series in $nd$. 
If $d$ is smaller than the distance between the PAA of $QR$ and the iSAX summary of $nd'$,
then all real distances between $QR$ and every series in $nd'$ are greater than $d$. 
\end{property}

\begin{lemma}
\label{lem:traverse exactly once}
$TraverseRootSubtree$ is invoked exactly once for each of the root subtrees of the index tree.
\end{lemma}

\begin{proof}
The use of the Fetch\&Increment object, $N_b$, and lines~\ref{eshqw:loop} and~\ref{eshqw:break} 
of Algorithm~\ref{eshqw} ensure that for every $i$, $0 \leq i \leq c$,
exactly one thread $p$ gets number $i$ when accessing $N_b$
(by executing line~\ref{eshqw:ato} of Algorithm~\ref{eshqw}). 
It follows that the function\\ $TraverseRootSubtree$ is invoked (line~\ref{eshqw:traverse}, Algorithm~\ref{eshqw})
exactly once for each of the root children of the index tree
(i.e., $p$ is the only thread that traverses the 
subtree numbered $i$ of the index tree). 
\end{proof}

Let $t$ be the point in time when the last search worker meets the barrier at line~\ref{eshqw:barrier}
of Algorithm~\ref{eshqw}. 

\begin{lemma}
\label{lem:priority queues}
Consider any $i$, $0 \leq i < N_q$. 
The following hold at $t$:
(1) queue[i] is a heap (i.e., it implements a priority queue);
(2) every element of $queue[i]$ is a distinct leaf of the index tree; thus, 
for every $j$, $0 \leq j < N_q$, $j \neq i$, the set of elements stored in $queue[i]$ 
and the set of elements stored in $queue[j]$ are disjoint.
\end{lemma}

\begin{proof}
By inspection of the code, it follows that an insertion of an element in $queue[i]$ 
can be performed only when line~\ref{tr:insert} of Algorithm~\ref{TraverseRootSubtree} is executed,
whereas no deletions are performed on $queue[i]$ by $t$.
Concurrent insertions on $queue[i]$ (executed by multiple search workers)
are serialized by acquiring and releasing the lock for $queue[i]$
in lines~\ref{tr:lq} and~\ref{tr:ulq} of Algorithm~\ref{TraverseRootSubtree}.
Thus, line~\ref{tr:insert} (Algorithm~\ref{TraverseRootSubtree}), which performs
an insertion of a leaf node in $queue[i]$, is executed in mutual exclusion.
Specifically, line~\ref{tr:insert} executes the sequential code for a heap insertion 
with parameter a tree node that has as its priority the distance calculated in line~\ref{insrq:mindist}.
These imply that $queue[i]$ is a heap, so claim~1 holds. 

By Lemma~\ref{lem:traverse exactly once}, $TraverseRootSubtree$
is invoked exactly once for each subtree of the index tree. 
$TraverseRootSubtree$ is a recursive algorithm that visits each tree node at most once. 
In particular,
line~\ref{tr:insert} (of Algorithm~\ref{TraverseRootSubtree}) is executed at most once for each node.
The condition of the \emph{else if} statement of line~\ref{tr:sleaf} 
ensures that line~\ref{tr:insert} is executed only for leaf nodes, so 
only leaf nodes are inserted in $queue[i]$. 
These imply that claim~2 holds.
\end{proof}

We say that an instance of $TraverseRootSubtree$ (Algorithm~\ref{TraverseRootSubtree}) 
{\em visits} a leaf node $nd$ if it executes lines~\ref{tr:lq}-\ref{tr:rq} 
with $node$ being equal to $nd$.

\begin{lemma}
\label{lem: pruning}
Let $BSF_t$ be the value of shared variable $BSF$ at time $t$. 
For every leaf node, $nd$, of the index tree 
that is not stored in the heaps of array $queue$ at $t$,
it holds that the real distance between the query and each of the data series 
stored in $nd$ is larger than $BSF_t$.
\end{lemma}

\begin{proof}
By inspection of the code (Algorithms~\ref{eshqw} and~\ref{TraverseRootSubtree}), 
it is easy to see that the value of $BSF$ does not change 
from the point that it is first set (on line~\ref{eshq:appro} of Algorithm~\ref{eshq}) until $t$.
Therefore, the value of $BSF$ is equal to $BSF_t$ 
during the execution of every instance of $TraverseRootSubtree$ . 

Consider any leaf node $nd$ that is not stored in the heaps of the $queue$ array at $t$. This can happen only 
if no instance of $TraverseRootSubtree$ visits this node. 
By Lemma~\ref{lem:traverse exactly once}, $TraverseRootsubtree$
is invoked exactly once for each subtree of the index tree. 
By inspection of the code, it follows that $nd$ belongs to the subtree
of a node $nd'$ (that might be $nd$ or one of its proper ancestors) 
on which the condition of the \emph{if} statement of line~\ref{tr:ifnodedist}
(Algorithm~\ref{TraverseRootSubtree}) is evaluated to \emph{false}, so that $TraverseRootSubtree$
is not called recursively on the nodes of $nd'$'s subtree (including $nd$),
and therefore all these nodes are not visited. 
Lines~\ref{insrq:mindist}-\ref{tr:ifnodedist} (Algorithm~\ref{TraverseRootSubtree}) ensure that 
the distance between the PAA of the query and the iSAX representation
of $nd'$ is greater than $BSF_t$. Thus, Property~\ref{iSAX property}
implies that the real distance between the query and each of the
data series stored in $nd$ is larger than $BSF_t$, as needed.
\end{proof}

The following observation is a simple consequence of the fact that 
the BSF variable is protected by a distinct lock, and that the value
of $BSF$ is updated only if the \emph{if} statement of line~\ref{process:2if} (Algorithm~\ref{ProcessQueue}) is evaluated to \emph{true}. 

\begin{observation}
\label{obs:strictly dec}
The sequence of values stored in shared variable BSF is strictly decreasing.
\end{observation}

\begin{theorem}
\label{thm:query correctness}
The response of $QR$ is correct. 
\end{theorem}

\begin{proof}
Fix any $i$, $0 \leq i < N_q$ and let $Q = queue[i]$. 
Deletions from $Q$ may occur only by executing line~\ref{process:dlmin} of Algorithm~\ref{ProcessQueue}.
Concurrent deletions from $Q$ are serialized by acquiring and releasing the lock for $Q$
(lines~\ref{process:dloc}-\ref{process:duloc}, Algorithm~\ref{ProcessQueue}). Note that this lock is distinct for each queue.
This and Lemma~\ref{lem:priority queues} imply that $Q$ respects the semantics of a priority queue. 
Therefore, when a node $nd$ with its $dist$ field being equal to $d$ is deleted from $Q$, 
all other nodes of $Q$  have higher values than $d$ in their $dist$ field.

Let $t_f$ be the first point in time at which $Q.finished$ is set to \emph{true}. 
Then, lines~\ref{process:1if} and~\ref{process:else} imply that the distance between the PAA of the query $QR$
and the iSAX summary of the last node $nd$ deleted from $Q$ is greater than or equal to $BSF$.
Property~\ref{iSAX property2} then implies that for every leaf node $nd$ contained in the queue at $t_f$, 
the minimum real distance between the query $QR$ and all the data series stored in $nd$ is
larger than the value of $BSF$ at $t_f$ (let this be $BSF_f$). This and Observation~\ref{obs:strictly dec} imply that 
none of the data series of $nd$ may result in a real distance to $QR$ smaller than $BSF_f$ (or future values of $BSF$),
and therefore none of them needs to be further examined. Note that as soon as the $finished$ bit
of $Q$ changes to \emph{true}, any future update of this field does not change its value (i.e.,
it simply re-writes \emph{true} to it); this is why writes into this field are not protected
by a lock.

Lemma~\ref{lem: pruning} implies that by processing just the leaf nodes in the heaps
of the $queue$ array (and not all leaf nodes of the index tree), 
the correctness of $QR$'s response is not jeopardized. Every such heap
is processed by at least one search worker. This is ensured by the fact that a search worker stops 
processing heaps of the $queue$ array only if it discovers that the $finished$ bits
of all of them have been set to \emph{true} (lines~\ref{eshqw:sgetqueue}-\ref{eshqw:egetqueue}, Algorithm~\ref{eshqw}).
\end{proof}

\remove{
\commentnote{R2D6:As the paper is mainly for optimizingthe parallel execution. It will be better to point out how many operations need to be locked for concurrency control. e.g., updating BSF, "stores the result in the appropriate iSAX buffer of index". It is be clear for system guys to understand the cost.
	It is also better to give some complexity analysis or cost modeling.  
	e.g., when creating time, the cost is reduce to T/w where w is the number of workers, and T is the time cost when creating the time series using one thread. But some operations may be T/w + f(w) where
	f(w) is the cost of contention.
}

{\bf ??? maybe we need a new section/subsection for this? ???}

for complexity of query answering, see end of page 4 here:
\url{http://helios.mi.parisdescartes.fr/~themisp/publications/sigmod20-progressive.pdf}

see also sec5 here:
\url{http://helios.mi.parisdescartes.fr/~themisp/publications/vldbj16-ads.pdf}

}

\section{Experimental Evaluation}
\label{sec:experiments}

We use synthetic 
and real datasets in order to compare the performance of MESSI with that of competitors from the literature and baselines we developed. 

We demonstrate that, under the same settings, MESSI is able to construct the index up 
to 4.2x faster, and answer similarity search queries up to 11.2x faster than the competitors. 
Overall, MESSI exhibits robust performance across datasets and settings, and 
enables for the first time the exploration of very large data series collections 
at interactive speeds, and leads to complex analytics that execute more than 1 order of magnitude faster than before.

\subsection{Setup}
\label{setupsec}
\noindent\textbf{[Environment]} 
We used a server with two Intel Xeon E5-2650 v4 2.2Ghz CPUs and 256GB RAM; each one of the two CPUs comprises 12 cores/24 hyper-threads.
All algorithms were implemented in C, and compiled using GCC v6.2.0 on Ubuntu Linux v16.04.

\noindent\textbf{[Algorithms]} 
We compared MESSI to the following algorithms: \\
(i) ParIS+~\cite{parisplus}, 
the state-of-the-art modern hardware data series index.\\
(ii) ParIS+TS, our extension of ParIS+, where we implemented in a parallel fashion the traditional tree-based 
exact search algorithm~\cite{shieh2008sax}.
In brief, this algorithm traverses the tree, and concurrently (1) inserts 
in the priority queue the nodes (inner nodes or leaves) that cannot be pruned based on the lower bound distance,
and (2) pops from the queues nodes for which it calculates the real distances to the candidate series~\cite{shieh2008sax}. 
In contrast, MESSI (a) first makes a \emph{complete pass} over the index using lower bound distance 
computations and then proceeds with the real distance computations; 
(b) it only considers the \emph{leaves} of the index for insertion in the priority queue(s); and (c) performs a \emph{second} filtering step 
using the lower bound distances when popping elements from the priority queue (and before computing 
the real distances). The performance results we present later justify the choices we have made in MESSI, 
and demonstrate that a straight-forward implementation of tree-based exact search leads to sub-optimal 
performance. \\
(iii) UCR Suite-P, our parallel implementation of 
the state-of-the-art optimized serial scan technique, UCR Suite~\cite{rakthanmanon2012searching}, which implements all the known optimizations for exact data series similarity search.
In UCR Suite-P, every thread is assigned a part of the in-memory data series array, and all threads  
concurrently and independently process their own parts, performing the real distance calculations in SIMD, 
and only synchronize at the end to produce the final result.
(We do not consider the non-parallel UCR Suite version in our experiments, since it is almost 300x slower.)

In all cases, the algorithms operated exclusively in main memory (the datasets were already loaded in memory, as well).
The code for all algorithms used in this paper is available online~\cite{sourcescode}.

\noindent\textbf{[Datasets]} 
In order to evaluate the performance of the proposed approach, 
we use several synthetic datasets for a fine grained analysis, 
and two real datasets from diverse domains.
Unless otherwise noted, the series have a size of 256 points, 
which is a standard length used in the literature, 
and allows us to compare our results to previous work.
We used synthetic datasets of sizes 50GB-200GB (with a default size of 100GB).
For the synthetic datasets, we used 
a random walk data series generator that works as follows: 
a random number is first drawn from a Gaussian distribution N(0,1), 
and then at each time point a new number is drawn from this distribution 
and added to the value of the last number. 
This kind of data generation has been extensively used in the past 
(and has been shown to model real-world financial data)~\cite{yi2000fast,shieh2008sax,wang2013data,isax2plus,zoumpatianos2016ads}.
We used the same process to generate 100 query series. 


For our first real dataset, \emph{Seismic}, 
we used the IRIS Seismic Data Access repository~\cite{iris} to gather 100M series representing seismic waves 
from various locations,  
for a total size of 100GB.
The second real dataset, \emph{SALD}, includes neuroscience MRI data series~\cite{url:SALD}, 
for a total of 200M series of size 128, of size 100 GB.
In both cases, we used as queries 100 series out of the datasets (chosen using our synthetic series generator).


We repeated all experiments 10 times and we report the average values. 
We omit the error bars, since all runs gave results that were very similar (less than 3\% difference).
The queries were always run in a sequential fashion, one after the other, 
in order to simulate an exploratory analysis scenario, 
where users formulate new queries after having seen the results of the previous one.

\begin{figure*}[tb]
	\begin{minipage}[b]{0.33\textwidth}
		\includegraphics[page=1,width=0.8\columnwidth]{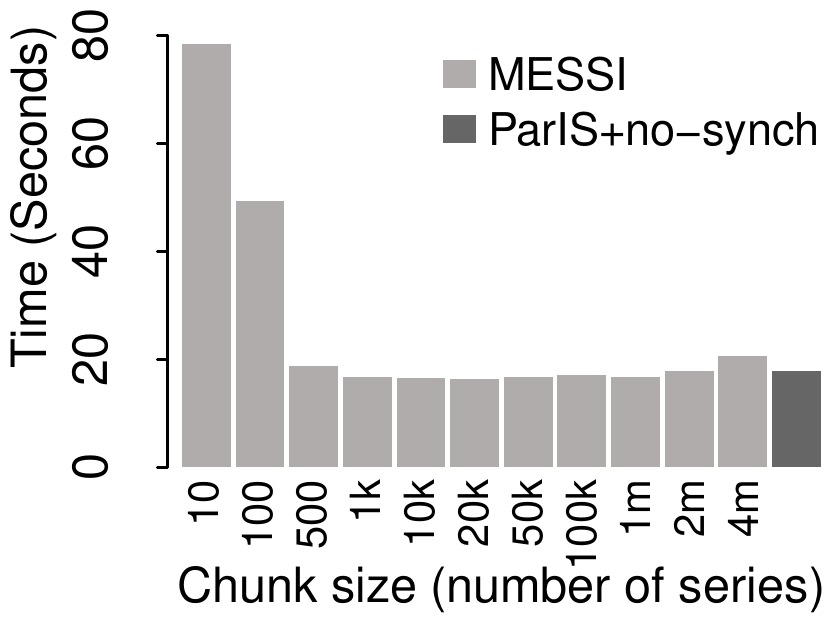}
		\caption{Index creation, vs. chunk size}		\label{fig:chunksize}
	\end{minipage}
	\begin{minipage}[b]{0.33\textwidth}
		\includegraphics[page=1,width=0.8\columnwidth]{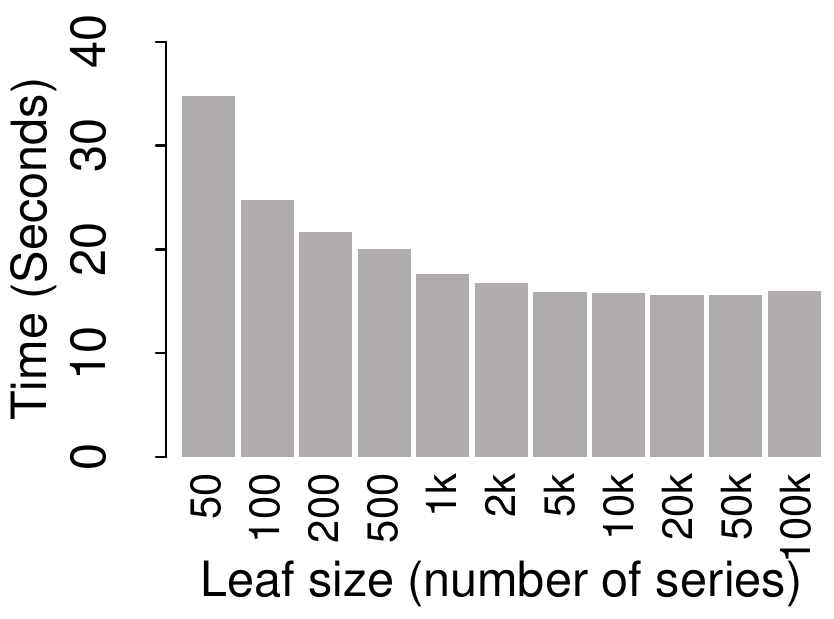}
		\caption{Index creation, vs. leaf size}
		\label{fig:indleaf}
	\end{minipage}
	\begin{minipage}[b]{0.33\textwidth}	
		\includegraphics[page=1,width=0.8\columnwidth]{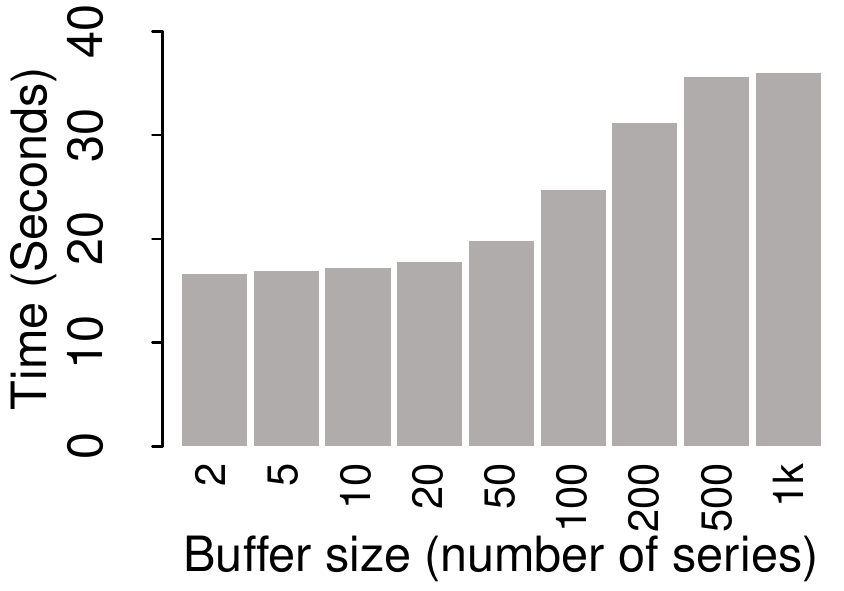}
		\caption{Index creation, vs. initial iSAX buffer size}
		\label{fig:initialisaxbuffersize}
	\end{minipage}
\end{figure*}

	\begin{figure}[tb]
	\centering
	\includegraphics[page=1,width=0.8\columnwidth]{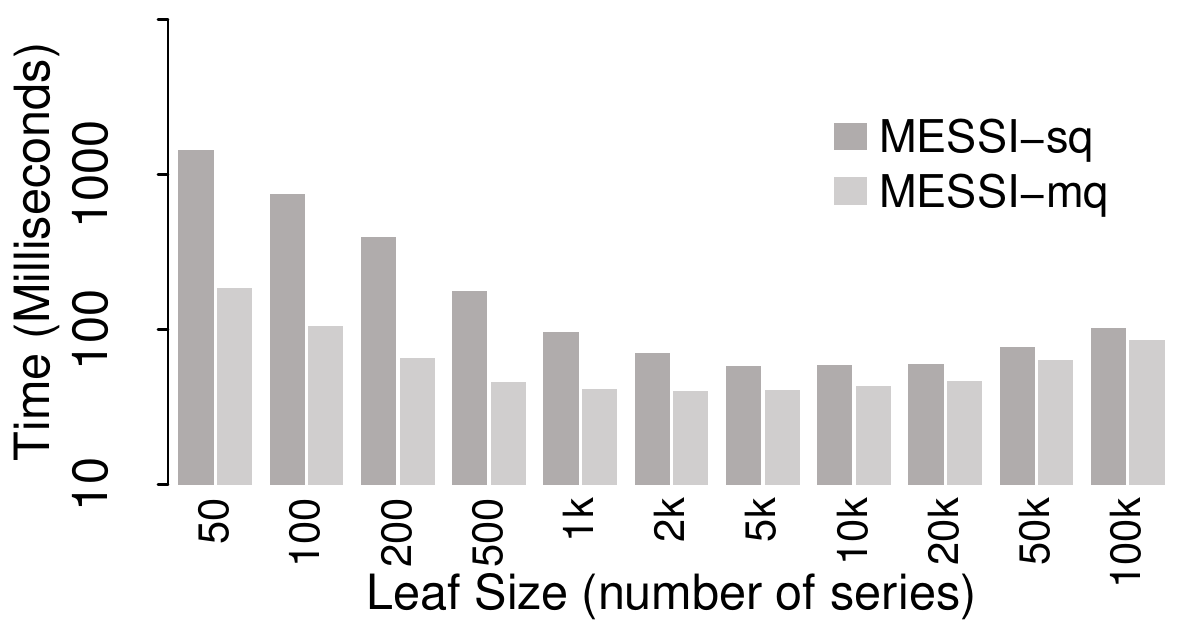}
	\caption{Query answering, vs. leaf size}
	\label{fig:dleaf}
\end{figure}

\subsection{Parameter Tuning Evaluation}
\label{parameter}
In all our experiments, we use 24 index workers and 48 search workers.
We have chosen the chunk size to be 20MB (corresponding to 20K series of length 256 points).
Each part of any iSAX buffer, initially holds a small constant number of data series, 
but its size changes dynamically depending on how many data
series it needs to store. 
The capacity of each leaf of the index tree
is 2000 data series (2MB). 
For query answering, MESSI-mq utilizes 24 priority queues (whereas MESSI-sq
utilizes just one priority queue). 
In either case, each priority queue is implemented 
using an array whose size changes dynamically based on how many elements must be stored in it.
Below we present the experiments that justify the choices for these parameters.

	\begin{figure}[tb]
	\centering
	\hspace*{-0.4cm}
	\includegraphics[page=1,width=\columnwidth]{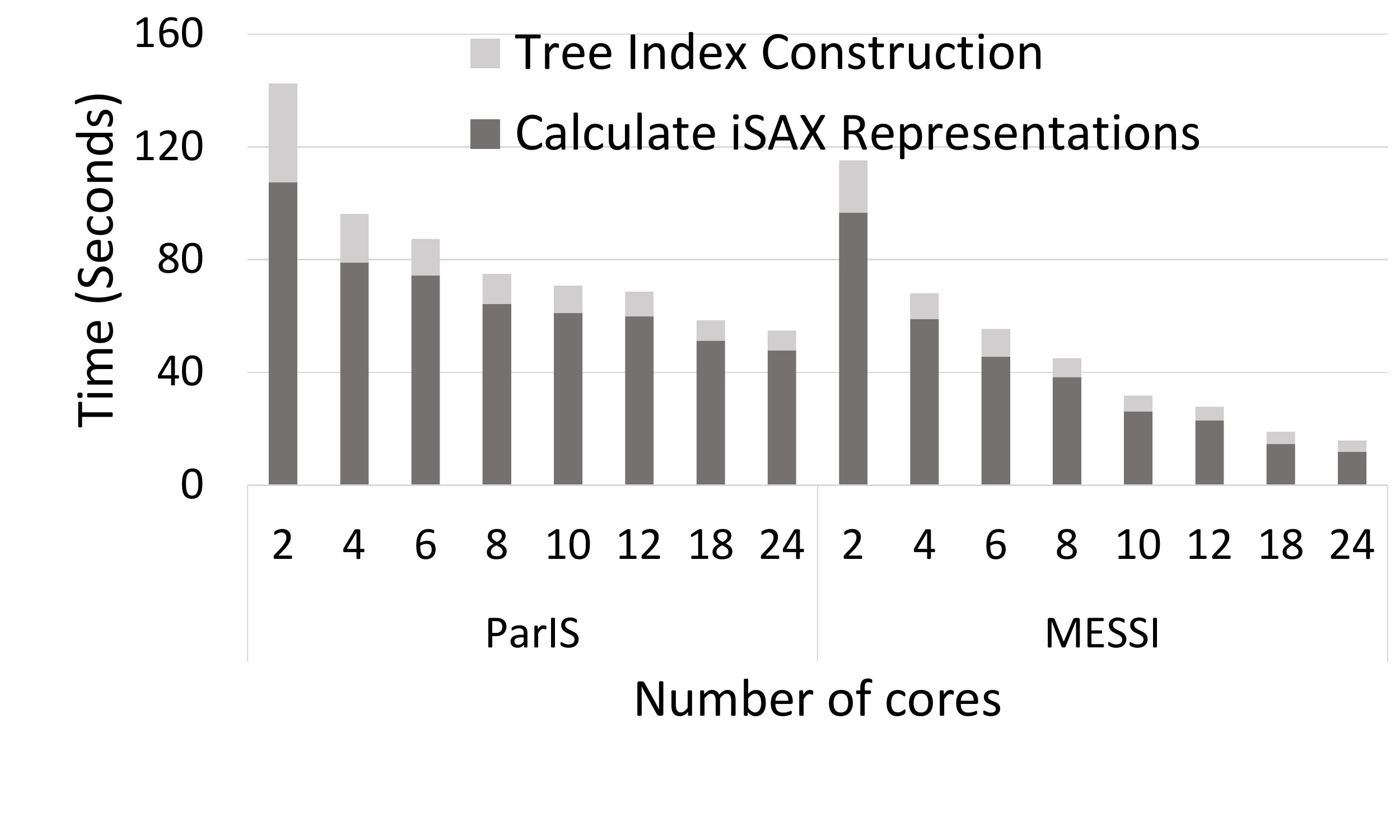}
	\caption{Index creation, varying number of cores}
	\label{fig:pRecBuf}
\end{figure}

\begin{figure*}[tb]
			\centering
	\begin{minipage}[b]{0.31\textwidth}
		\centering
		\includegraphics[page=1,width=\columnwidth]{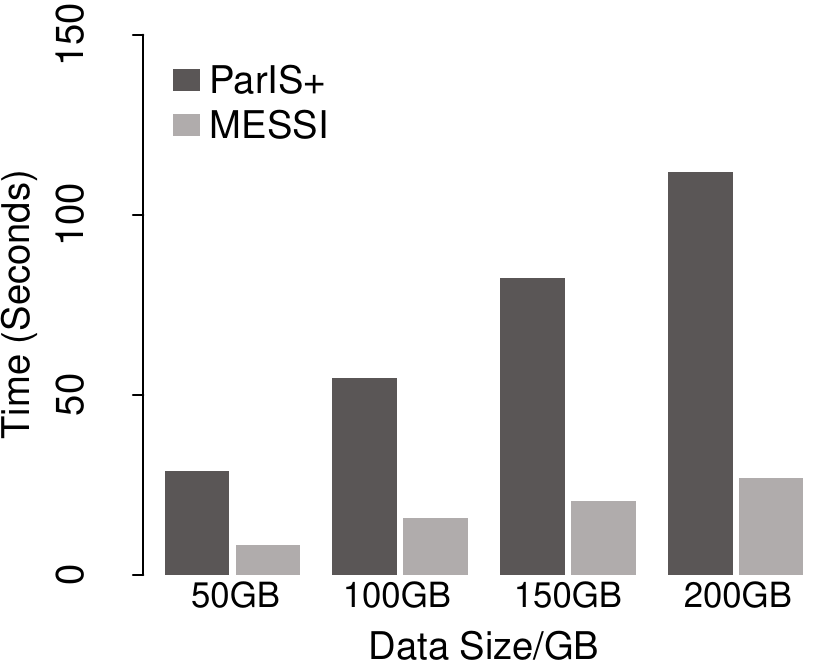}
		\caption{Index creation, vs. data size}
		\label{fig:incvarysize}
	\end{minipage}
	\begin{minipage}[b]{0.34\textwidth}
		\centering
		\includegraphics[page=1,width=\columnwidth]{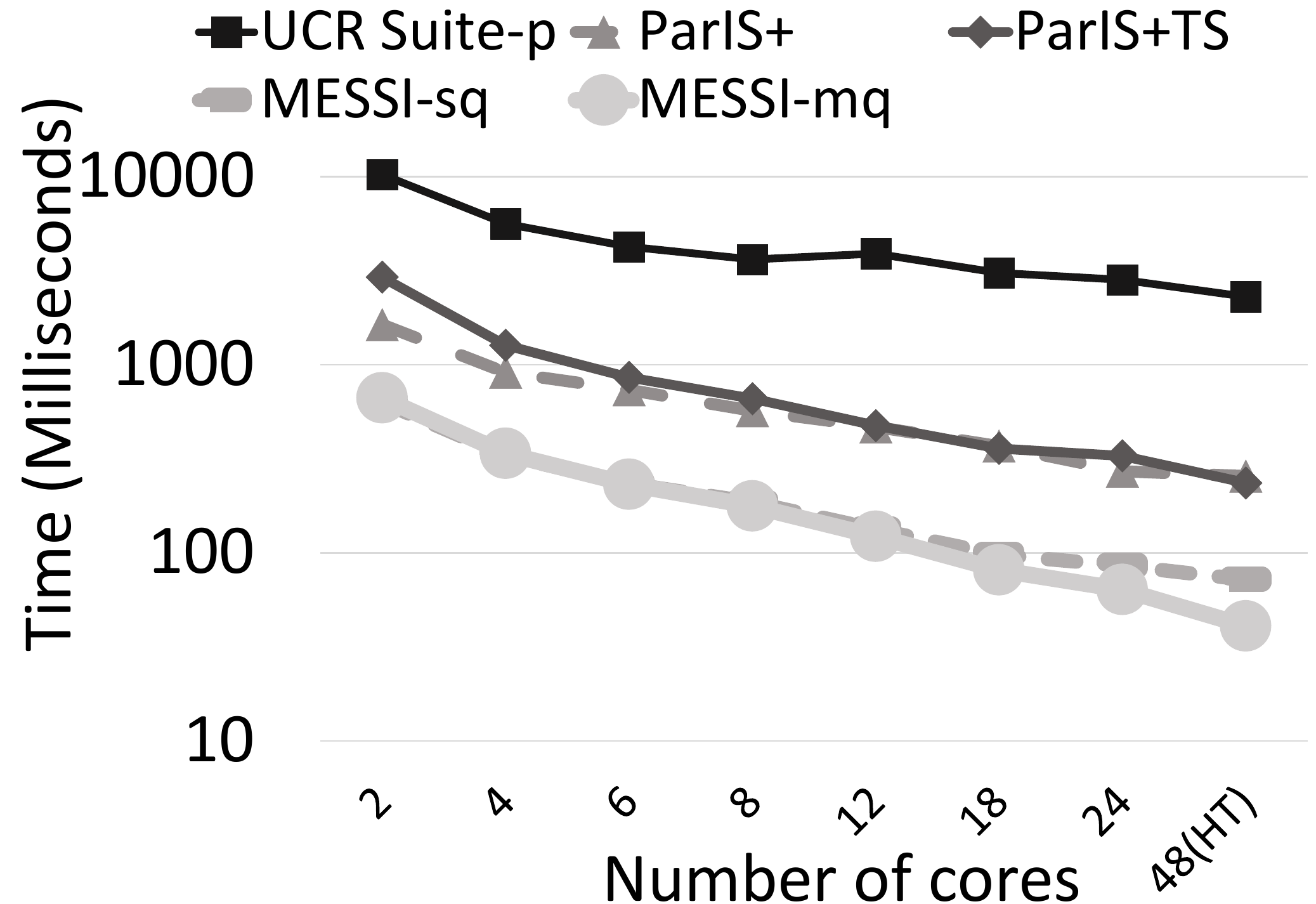}
		\caption{Query answering, vs. number of cores}
		\label{fig:varycore}
	\end{minipage}
	\begin{minipage}[b]{0.34\textwidth}
		\centering
		\includegraphics[page=1,width=0.98\columnwidth]{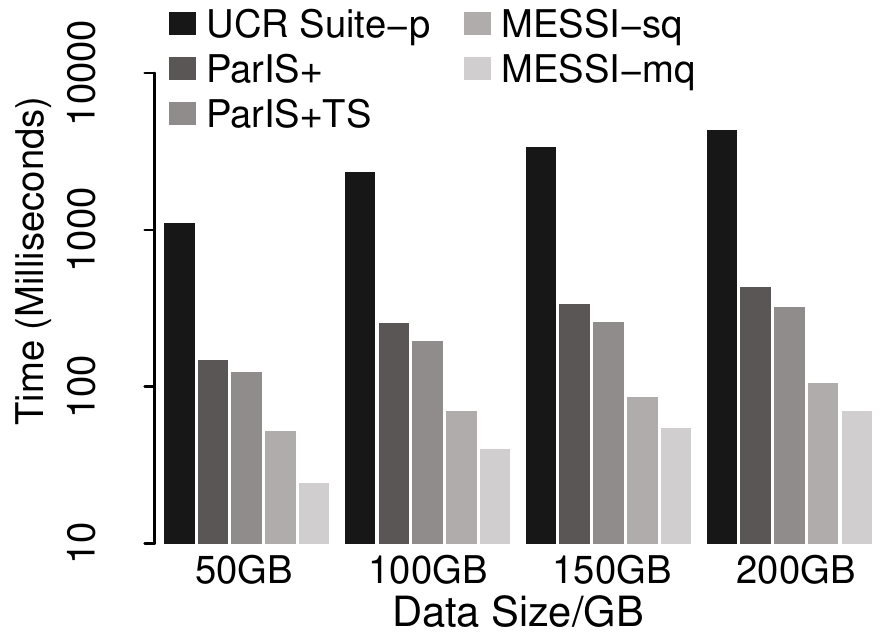}
		\caption{Query answering, vs. data size}
		\label{fig:scalquery}
	\end{minipage}

\end{figure*}

Figure~\ref{fig:chunksize} illustrates the time it takes MESSI to build the tree index for different chunk sizes on a random dataset of 100GB. 
The required time to build the index decreases when the chunk size is small and does not have any big influence in performance after the value of 1K (data series). 
Chunk sizes smaller than 1K result in high contention when accessing the fetch\&increment object used to assign chunks to index workers.
In our experiments, we have chosen a size of 20K, as this gives slightly better performance. 

Figures~\ref{fig:indleaf} and~\ref{fig:dleaf} show the impact of varying the index tree leaf size on the time cost of index creation and query answering, respectively. 
As we see in Figure~\ref{fig:indleaf}, the larger the leaf size is, the faster index creation becomes.
However, once the leaf size becomes 5K or more, this time improvement
is insignificant. On the other hand, Figure~\ref{fig:dleaf} shows that the query answering time takes its minimum value when the leaf size is set to 2K (data series). 
So, we have chosen this value for our experiments.  

Figure~\ref{fig:dleaf} indicates that the influence of varying the leaf size is significant for query answering. Note that when the leaf size is small, there are more leaf nodes in the index tree and therefore, it is highly probable that more nodes will be inserted in the queues, and vice versa. 
As the leaf size increases, the number of real distance calculations performed to process each one of the leaves in the queue is larger. 
This causes
load imbalance among the different search workers that process the priority queues.    
For these reasons, we see that at the beginning the time goes down as the leaf size increases, it reaches its minimum value for leaf size 2K series, and then it goes up again as the leaf size further increases. 

Figure~\ref{fig:initialisaxbuffersize} shows the influence of the initial iSAX buffer size during index creation. This initialization cost is not negligible given that we allocate $2^w$ iSAX buffers, each consisting of $24$ parts (recall that 24 is the number of index workers in the system).  
As expected, smaller initial sizes for the buffers result in better performance. 
We have chosen the initial size of each part of the iSAX buffers to be a small constant number of data series. 
(We also considered a design that collects statistics and allocates the iSAX buffers right from the beginning, but was slower.)

We finally justify the choice of using more than one priority queues
for query answering. 
%
As Figure~\ref{fig:varycore} shows, MESSI-mq and MESSI-sq have similar performance when the number of threads is smaller than 24. 
However, as we go from 24 to 48 cores, the synchronization cost for accessing the single priority queue in MESSI-sq has negative impact in performance.
Figure~\ref{fig:diffq}  presents the breakdown of the query answering time for these two algorithms. 
The figure shows that in MESSI-mq, the time needed to insert and remove nodes from the list is significantly reduced. 
As expected, the time needed for the real distance calculations and for the tree traversal are about the same in both algorithms. 
This has the effect that the time needed for the distance calculations becomes the dominant factor. 
The figure also illustrates the percentage of time that goes on each of these tasks.    

Finally, Figure~\ref{fig:dqueue} shows the impact of the number of priority queues on query answering performance.
As the number of priority queues increases, the time goes down, and is minimized for 24 queues. 
So, we have chosen this value for our experiments.  
We also note that each one of these queues handles almost the same number of elements. 
Our experiments showed that the standard deviation of the number of elements in the queues was always less than 0.8\% of the mean number of elements over all the queues used.

\begin{figure}[tb]
	\centering
	\includegraphics[page=1,width=0.9\columnwidth]{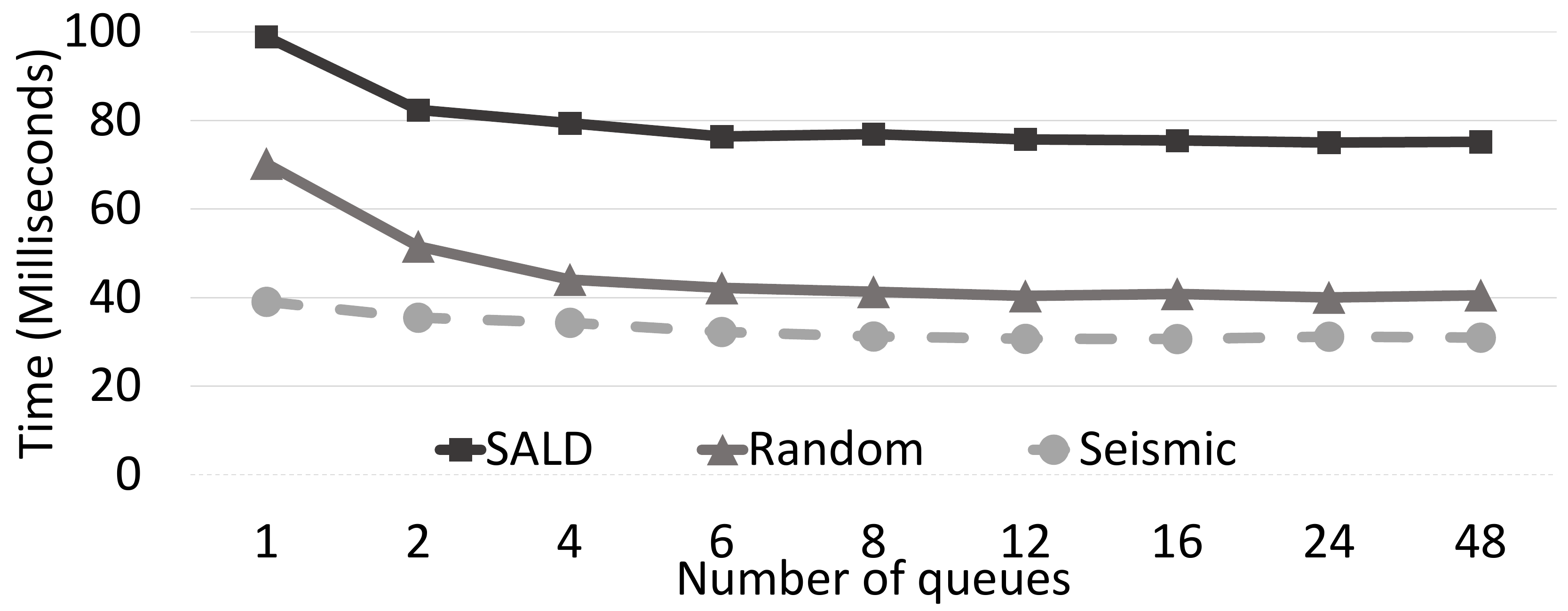}
	\caption{Query answering, vs. number of queues}
	\label{fig:dqueue}
\end{figure}

\begin{figure}[tb]
	\centering
	\includegraphics[page=1,width=\columnwidth]{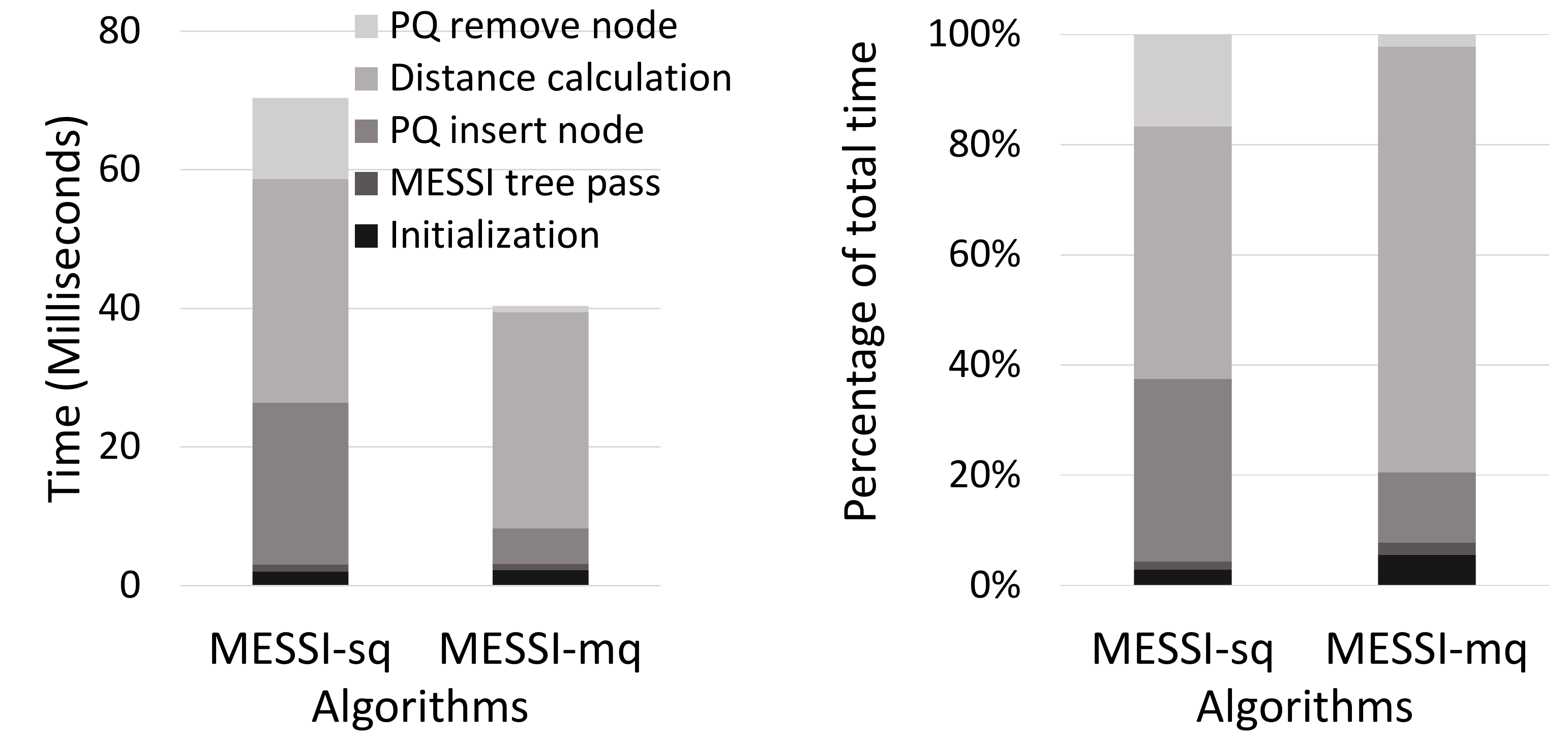}
	\caption{Query answering with different queue type}
	\label{fig:diffq}
\end{figure}

\begin{figure}[tb]
	\begin{minipage}[tb]{0.37\columnwidth}
		\includegraphics[page=1,width=0.91\columnwidth]{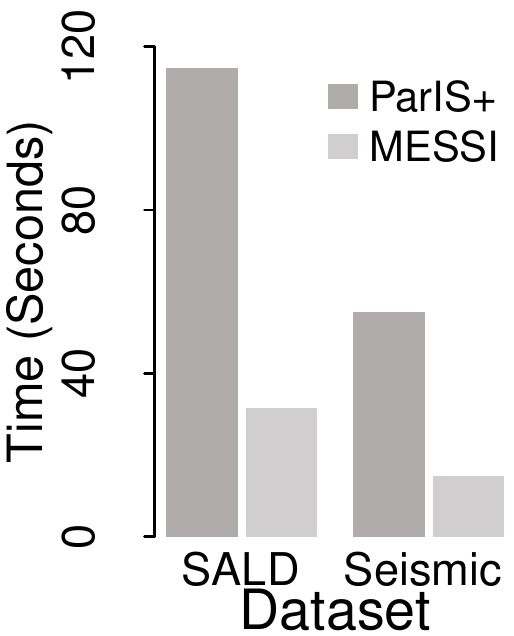}
		\caption{Index creation for real datasets}
		\label{fig:realinc}
	\end{minipage}
	\hspace*{0.1cm}
	\begin{minipage}[tb]{0.62\columnwidth}
		\includegraphics[page=1,width=\columnwidth]{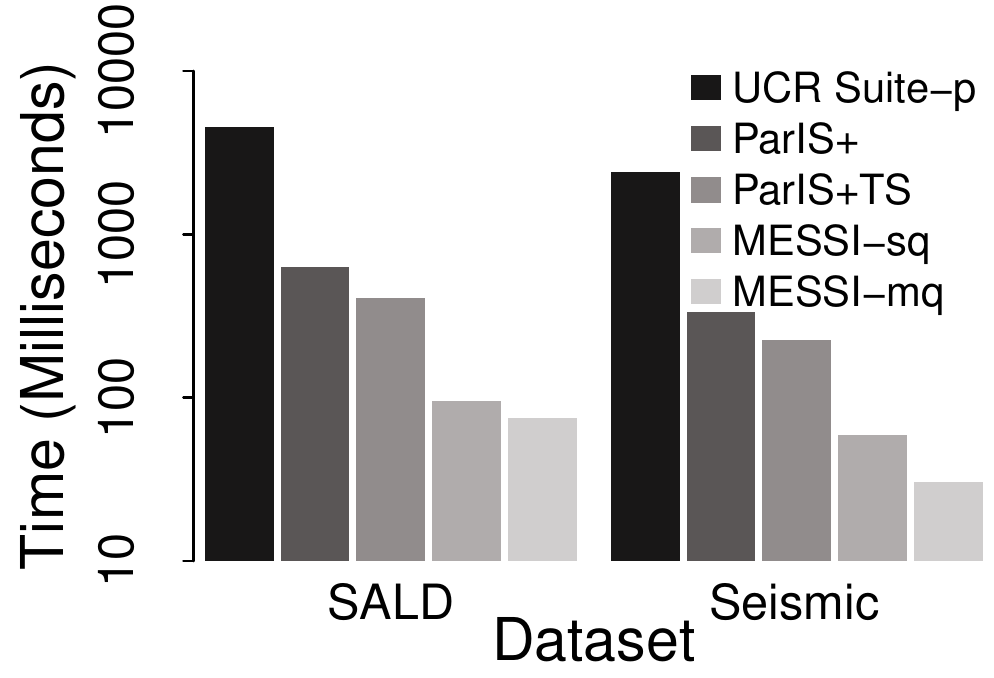}
		\caption{Query answering for real datasets}
		\label{fig:realqa}
	\end{minipage}
\end{figure}

\subsection{Comparison to Competitors}
\noindent{\bf [Index Creation]}
Figure~\ref{fig:pRecBuf} compares the index creation time of MESSI with  that 
of ParIS+ as the number of cores increases for a dataset of 100GB. 
The time MESSI needs for index creation 
is significantly smaller than that of ParIS+. 
Specifically, MESSI is 3.5x faster than ParIS+. 
The main reasons for this are on the one hand 
that MESSI exhibits lower contention cost when accessing the iSAX buffers
in comparison to the corresponding cost paid by ParIS+ to fill in the Receiving Buffers,
and on the other hand, that MESSI achieves better load balancing when performing the computation
of the iSAX summaries from the raw data series.
Note that due to synchronization cost, 
the performance improvement that both algorithms exhibit decreases 
as the number of cores increases; this trend is more prominent in ParIS+, while MESSI manages to exploit to a larger degree the available hardware.

In Figure~\ref{fig:incvarysize}, we depict the index creation time as the dataset size 
grows from 50GB to 200GB. 
We observe that MESSI performs up to 4.2x faster than ParIS+ (for the 200GB dataset), 
with the improvement becoming larger with the dataset size.


\begin{figure}[tb]
	\centering
\begin{minipage}[b]{\columnwidth}
	\subfigure[Lower bound dist. calculations
	\label{fig:ndist1}]{
		\includegraphics[page=1,width=0.49\columnwidth]{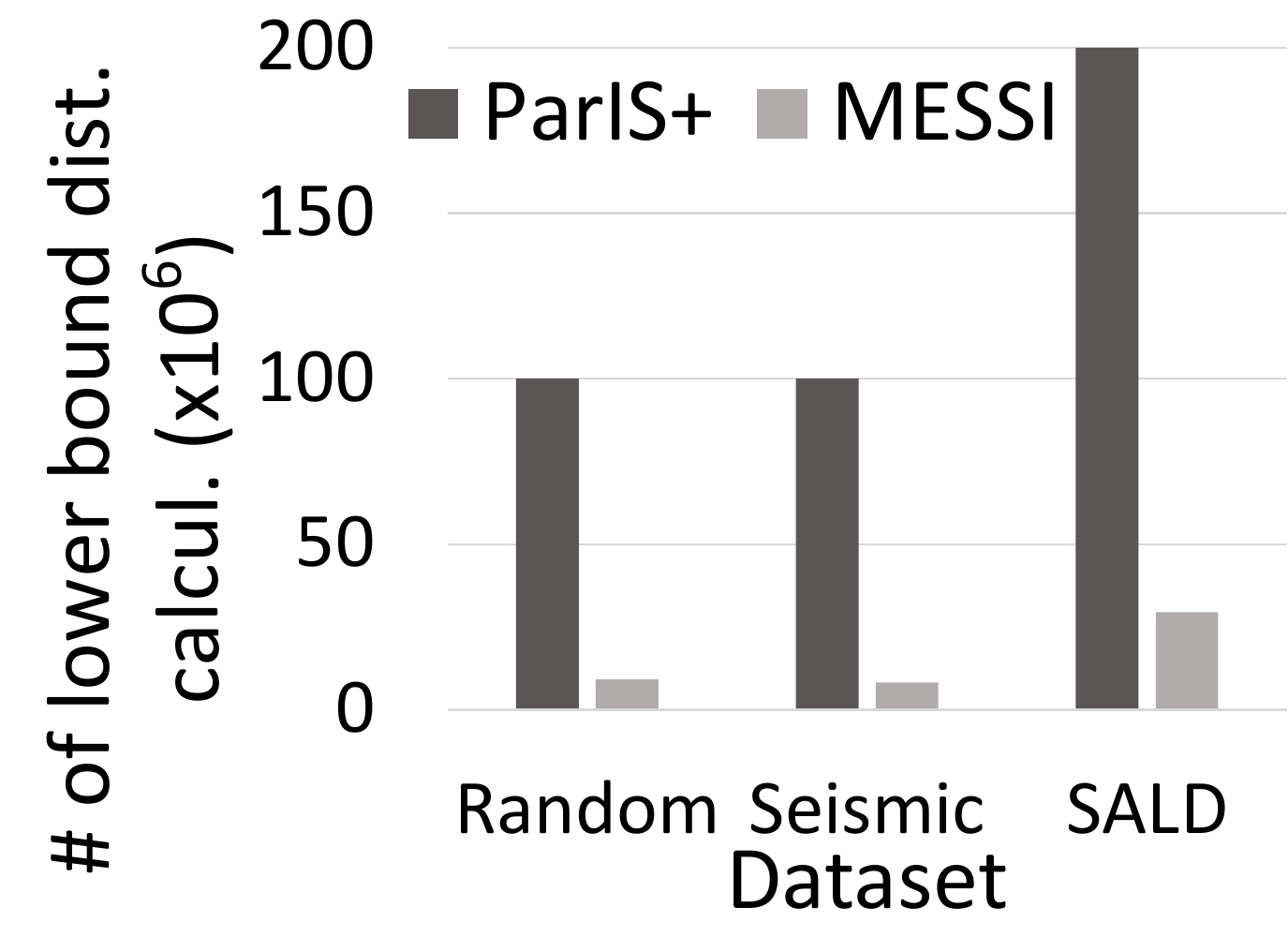}
	}
	\subfigure[Real distance calculations 
	\label{fig:ndist2}]{
		\includegraphics[page=2,width=0.49\columnwidth]{picture/calculationnumber2}
	}
	\caption{Number of distance calculations}
	\label{fig:ndist}
\end{minipage}

\end{figure}



\noindent{\bf [Query Answering]}
%
Figure~\ref{fig:varycore} 
compares the performance of the MESSI query answering algorithm
to its competitors, as the number of cores increases, for a random dataset of 100GB
(y-axis in log scale). 
The results show that both MESSI-sq and MESSI-mq perform much better
than all the other algorithms. 
Note that the performance of MESSI-mq is better than that of MESSI-sq,
so when we mention MESSI in our comparison below we refer to MESSI-mq. 
MESSI is 55x faster than UCR Suite-P and 6.35x faster than ParIS+ 
when we use 48 threads (with hyperthreading). In contrast to ParIS+, 
MESSI applies pruning when performing the lower bound distance calculations
and therefore it executes this phase much faster. 
Moreover, the use of the priority queues result in 
even higher pruning power. As a side effect, MESSI also performs 
less real distance calculations than ParIS+. 
Note that UCR Suite-P does not perform any pruning, thus resulting 
in a much lower performance than the other algorithms. 

Figure~\ref{fig:scalquery} shows that
this superior performance of MESSI is observed across  
various data set sizes: MESSI is up to 61x faster than 
UCR Suite-p (for 200GB), up to 6.35x faster than ParIS+ (for 100GB), 
and up to 7.4x faster than ParIS+TS (for 50GB).

In the following experiment, we studied the query answering performance as we varied the length of the data series.
We measured the query answering time for data series, whose length ranged from 128 to 2048 points. 
The five datasets with different data series length that we used in this experiment are random, and in order to factor out the effect of dataset size (following previous work~\cite{lernaeanhydra,lernaeanhydra2}), they are all 100GB in size. 
This means that the datasets with longer series contain a smaller number of series overall. 
In all cases, the iSAX summaries were built using 16 segments. 
	
Figure~\ref{fig:varylength} shows that the query answering performance of all algorithms increases with the length of the series.
This is to be expected, since the total number of series, and therefore distances that should be computed, is decreasing. 
We observe that as we increase the series length, the lower bounds become looser. 
Consequently, the proportion of lower bound and real distance calculations increase, as reported in Figures~\ref{fig:varylengthlbdp} and~\ref{fig:varylengthrdp}, respectively. 
These results also show that ParIS+ spends time to perform lower bound calculations for all series in the dataset in order to save on the real distance calculations, while ParIS+TS ends up performing a large number of both lower bound and real distance calculations. 
On the other hand, the query answering strategy of MESSI proves very efficient in terms pruning, leading to a small number of lower bound and real distance calculations.
We also observe that when compared to ParIS+TS, MESSI is much more efficient in handling the priority queues. 
Figure~\ref{fig:varylengthqueuetime} shows the time spent on priority queue insertion and deletion operations for the two algorithms.
The time to handling queue become stable for both 2 algorithms. 
The slow performance of ParIS+TS is due to the fact that it inserts in the queue not only the leaf nodes (like MESSI), but also the inner nodes. 
All the above performance characteristics make MESSI the overall winner in terms of query answering time, across all data series lengths we tried (Figure~\ref{fig:varylength}).

%
%
%
%
%

\begin{figure}[tb]
	\centering
	\includegraphics[page=1,width=\columnwidth]{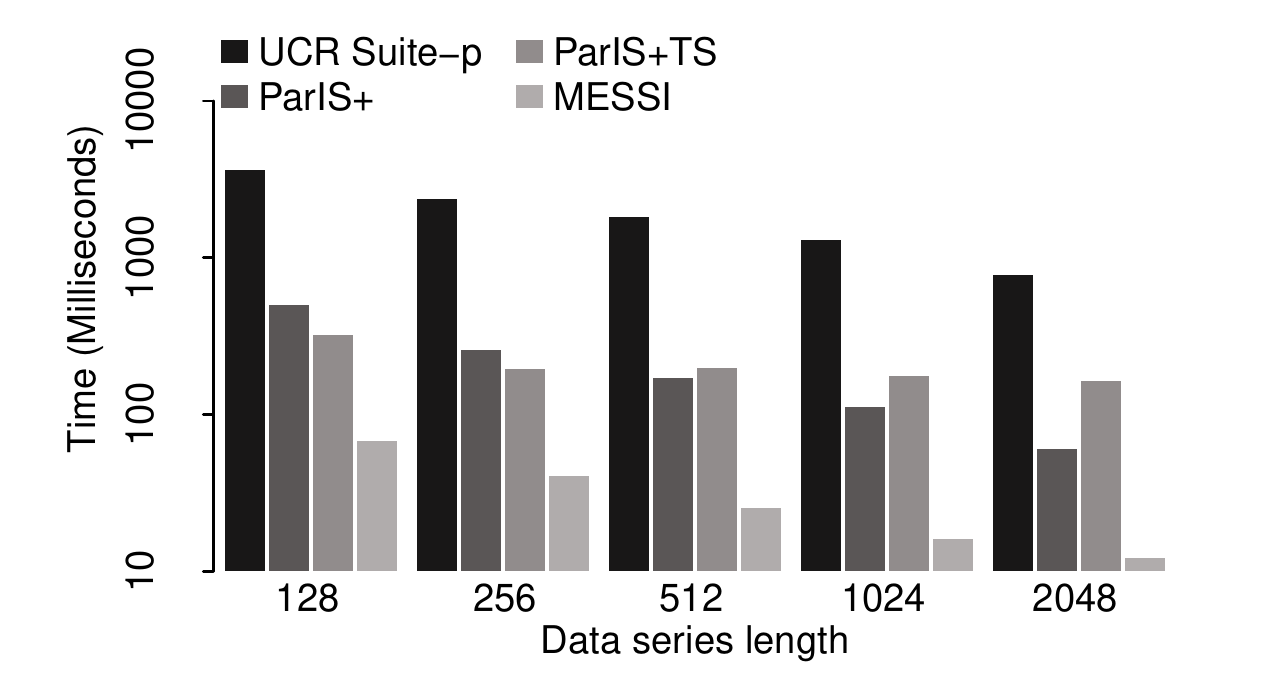}
	\caption{Query answering, vs. data series length.}
	\label{fig:varylength}
\end{figure}

\begin{figure}[tb]
	\centering
	\includegraphics[page=1,width=\columnwidth]{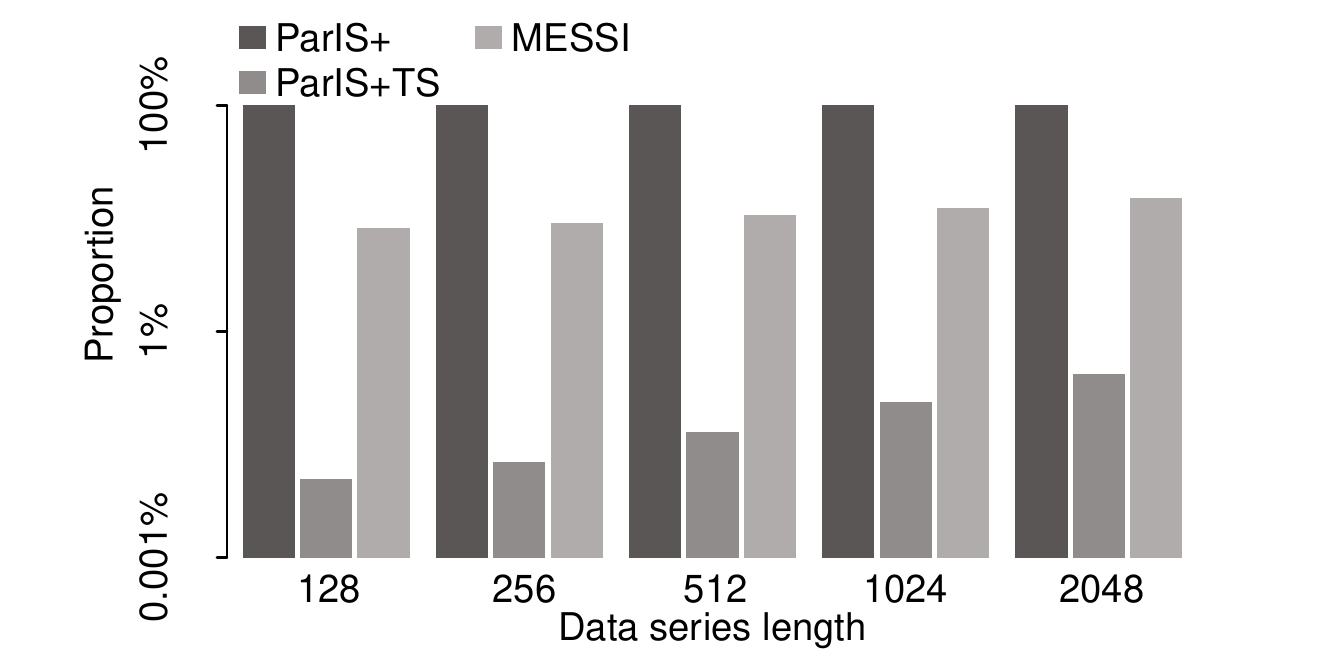}
	\caption{Proportion of lower bound distance calculations, vs. data series length.}
	\label{fig:varylengthlbdp}
\end{figure}

\begin{figure}[tb]
	\centering
	\includegraphics[page=1,width=\columnwidth]{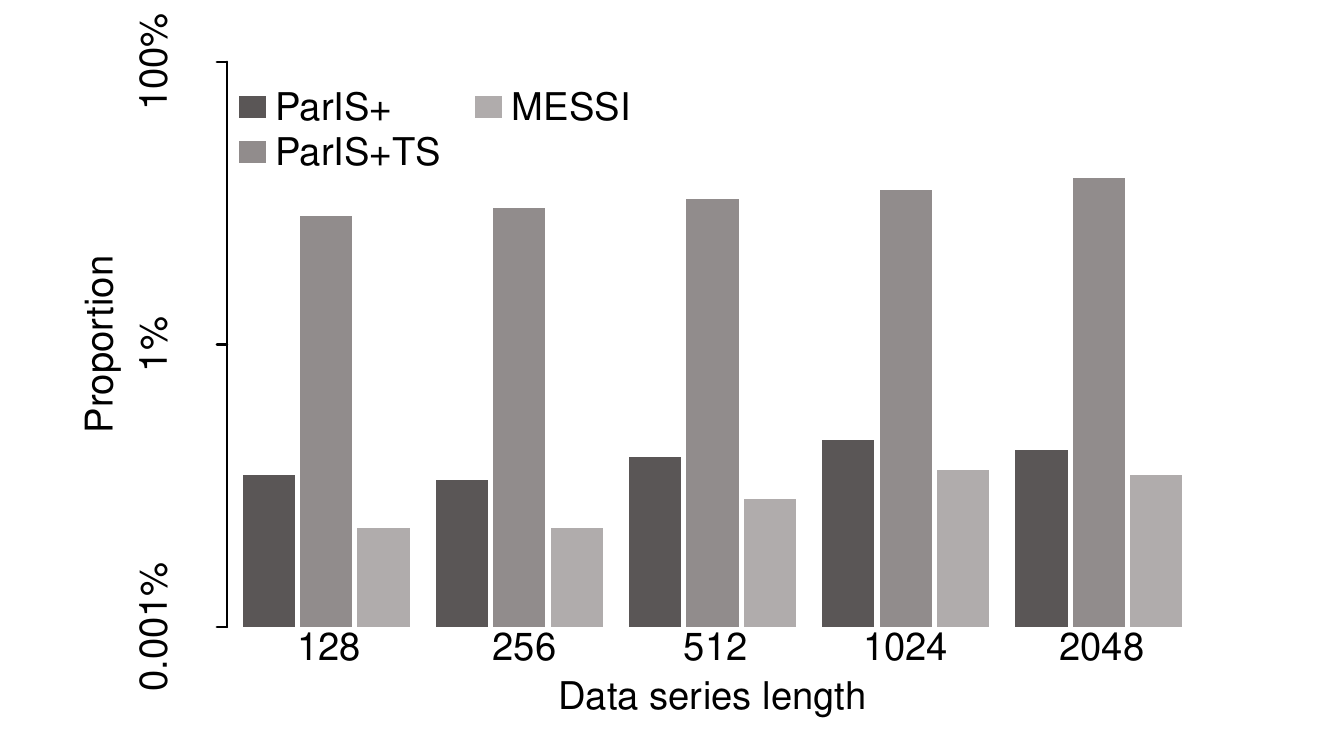}
	\caption{Proportion of real distance calculations, vs. data series length.}
	\label{fig:varylengthrdp}
\end{figure}

\begin{figure}[tb]
	\centering
	\includegraphics[page=1,width=\columnwidth]{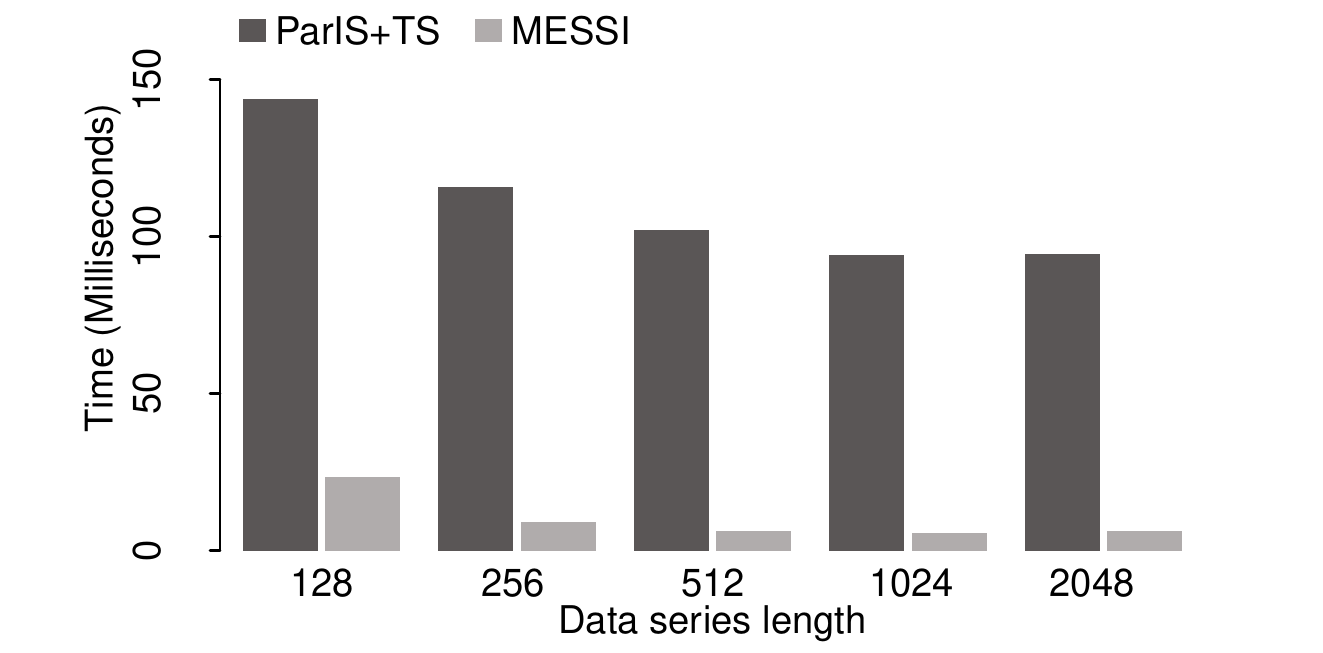}
	\caption{Time spent in priority queue insertions and deletions, vs. data series length.}
	\label{fig:varylengthqueuetime}
\end{figure}

\noindent{\bf [Performance Benefit Breakdown]}
We now evaluate each of the design choices of MESSI in isolation. 
This evaluation serves as a methodological analysis that can help understand the benefit of individual design decisions, and if/how they apply to other indices.

Note that some of our design decisions stem from the fact that in our index 
the root node has a large number of children. 
Thus, the same design ideas are applicable to the iSAX family of indices (e.g., iSAX2+~\cite{isax2plus}, ADS+~\cite{zoumpatianos2016ads}, \\ULISSE~\cite{ulisse}). 
Other indices however~\cite{lernaeanhydra}, use a binary tree (e.g., DSTree), or a tree with a very small fanout (e.g., SFA trie, M-tree), so new design techniques are required for efficient parallelization. 
However, some of our techniques (e.g., 
the use of SIMD, priority queues, and some of the data structures designed to reduce the syncrhonization cost) 
can be applied to all other indices.

We first examine index creation (refer to Figure~\ref{fig:stepbystepindex}). 
The main performance benefit comes from removing the synchronization cost of ParIS+ when filling up the receiving buffers.
In accessing the buffers, we completely eliminated 
the synchronization cost by splitting each such buffer into as many chunks as the number of worker threads.
Then each thread inserts and removes elements without encountering any contention, leading to a 2.5x speedup (shown as ParIS+no-synch in the graph) when compared to ParIS+.
To achieve better load balancing, 
MESSI splits the array into smaller chunks and uses a Fetch\&Add object to 
assign chunks to threads, resulting to a further performance improvement of 11\%.  	

Then, we examine the query answering performance (refer to Figure~\ref{fig:stepbystepquery}).
The leftmost bar (ParIS+SISD) shows the performance 
of ParIS+ when SIMD is \emph{not} used. 
By employing SIMD,
ParIS+ becomes 60\% faster than ParIS+SISD. 
We then measure the performance for ParIS+TS, which is about 10\% faster than ParIS+.
This improvement comes form the fact that using 
the index tree (instead of the SAX array that ParIS+ uses) to prune the search space and determine the data series for which 
a real distance calculation must be performed, significantly reduces the number of lower bound distance 
calculations. 
ParIS+ calculates lower bound distances for all the 
data series in the collection, and pruning is performed only when calculating
real distances, whereas in ParIS+TS pruning also occurs when calculating lower bounds. 



Next, we apply the following technique to reduce the number of real distance
calculations we perform: for each node extracted from the priority queue (where ParIS+TS stores the nodes it needs to process),
we first calculate the lower bound distance and only if it is smaller than the BSF,
the algorithm calculates the real distance. 
We call this algorithm ParIS+TS-LB, which  
results in a further performance improvement 
of 13\%.
MESSI-sq improves upon ParIS+TS-LB in that it inserts in the priority queue
only leaf nodes. This reduces the number of nodes that are inserted in the
priority queue and therefore also the size of the queue, as well as the contention 
incurred when accessing it. 
Given that in a priority queue, all threads need to synchronize at the root node,
this contention is usually rather high. Figure~\ref{fig:stepbystepquery} shows that MESSI-sq is 65\% faster 
than ParIS+TS-LB.
MESSI-mq further reduces the synchronization cost by maintaining more than one queues
and having different threads choose on which queues to work on. This makes MESSI-mq 42\% faster
than MESSI-sq.

In Table~\ref{table4}, we report a breakdown of the number of operation executions that the different algorithms perform. 
These numbers help explain the observations (and design choices) mentioned above. 
Observe that ParIS+TS performs many real distance calculations, because it does not use the second-level filter opportunity offered by the full-length iSAX representation of each data series. 
ParIS+TS-LB improves on this aspect. 
When compared to ParIS+, ParIS+TS-LB only performs 9\% of the lower bound distance calculations, because it uses the index tree and the priority queue in order to prune. 
This also means that ParIS+TS-LB performs less than 50\% of the real distance calculations. 
However, ParIS+TS-LB still executes many insert/delete node operations on the priority queue. 
MESSI sq/hq are much more efficient in handling the priority queue, since it only inserts leaf nodes in the queue. 
\begin{table*}[tb]
	\centering
	\caption{Query answering algorithms comparison: number of times an operation is executed (average over 100 queries).}\label{table4}
\hspace*{-0.3cm}	
	\begin{tabular}{|c|c|c|c|c|c|}
		\hline
		& {\bf ParIS+} & {\bf ParIS+TS} & {\bf ParIS+TS-LB} & {\bf MESSI-sq} & {\bf MESSI-mq} \\
		\hline
		PQ ins. node &n/a&69,117&	69,134&	14,620&	14,611\\
		\hline
		PQ del. node&n/a&20,051	&20,111&	11,152&	10,747\\
		\hline
		LBD calcul. &100 M&	69,117	&9,173,401&	9,175,400	&9,170,162\\
		\hline
		RD calcul. &112,321	&9,183,312	&52,139	&54,207	&53,919\\
		\hline
	\end{tabular}
\end{table*}

\begin{table}[tb]
	\centering
	\caption{Index expansion rate (index size as a percentage of the original data size).}\label{table5}
	\hspace*{-0.3cm}	
	\begin{tabular}{|c|c|c|c|}
		\hline
		&Synthetic&Seismic&SALD\\
		& {100GB} & {100GB} & {100GB} \\
		& {100M series} & {100M series} & {200M series} \\
		\hline
		index expansion rate&5.7\%&	5.1\%& 10.5\%\\
		\hline
	\end{tabular}
\end{table}


\noindent{\bf [Real Datasets]}
Figures~\ref{fig:realinc} and~\ref{fig:realqa} reaffirm that MESSI
exhibits the best performance for both index creation and query answering, 
even when executing on the real datasets, SALD and Seismic (for a 100GB dataset), for the reasons listed in the previous paragraphs. 
Regarding index creation, MESSI is 3.6x faster than ParIS+ on SALD and 3.7x faster than ParIS 
on Seismic, for a 100GB dataset. 
Moreover, for SALD, MESSI query answering is 60x faster than UCR Suite-P 
and 8.4x faster than ParIS+, whereas for Seismic,
it is 80x faster than UCR Suite-P, and almost 11x faster than ParIS+.

Figures~\ref{fig:ndist1} and~\ref{fig:ndist2} illustrate the number of lower bound and real distance calculations,
respectively, performed by the different query algorithms on the three datasets. 
ParIS+ calculates the distance between the iSAX summaries of every single data series 
and the query series (because, as we discussed in Section~\ref{sec:prelim}, it implements the SIMS strategy for query answering). 
In contrast, MESSI performs pruning even during the lower bound distance calculations, 
resulting in much less time for executing this computation. 
Moreover, this results in a significantly reduced number of series whose 
real distance to the query must be calculated. 

The use of the priority queues lead to even less real distance calculations, 
because they help the BSF to converge faster to its final value. 
MESSI performs no more than 15\% of the lower bound distance calculations
performed by ParIS+. 

In the next experiment, we report results with query workloads of increasing difficulty (similarly to earlier work~\cite{DBLP:journals/vldb/ZoumpatianosLIP18}). 
For these workloads, we select series at random from the collection, add to each point Gaussian noise ($\mu=0$, $\sigma=0.01$-$0.1$), and use these as our queries. 
Finally, we also select series at random and remove them from the collection, and use these as our \emph{Real} workload. 

Figure~\ref{fig:RDnumber} shows that the pruning proportion of all algorithms increases as we increase the level of noise in the query workloads, while \emph{Real} is even more difficult: for the Seismic dataset, we can only prune 40-55\% of the real distance calculations. 
Nevertheless, MESSI achieves in all cases the best pruning, thanks to its use of the priority queue, where the BSF is always updated as early as possible. 
The query answering time performance for these workloads is depicted in Figure~\ref{fig:noisequery2}. 
The results show that as the queries get harder, ParIS+ becomes worse that UCR Suite-p. 
ParIS+ pays the penalty of having to generate and process the candidate list, which grows very large when pruning is small.
ParIS+TS has an advantage in this respect, because it only needs to handle the priority queue (of the non-pruned nodes), which is smaller in size, resulting in an overhead that is much less than ParIS+. 
MESSI is always better than all competitors: it performs 3.5x-100x faster than UCR Suite-p on the Seismic dataset, and 16x-135x faster on SALD. 
(ParIS+ was much slower, and we terminated its execution after 10K milliseconds per query.)

In Table~\ref{table5}, we report the MESSI index expansion rate (i.e., the index size as a percentage of the original data size) for the synthetic and real datasets in our study.
We observe that for our 100GB datasets, the MESSI index occupies $\sim$5GB of space for Synthetic and Seismic, and $\sim$10GB for SALD.
Note that the series in SALD have a length of 128 points (compared to the 256 of Synthetic and Seismic); hence, this dataset contains double the number of series than the other two datasets. 
This means that the index contains double the number of iSAX summaries.
Overall, we conclude that the MESSI index expansion rate is small, rendering
MESSI a space efficient index.

\begin{figure}[tb]
		\begin{minipage}[b]{\columnwidth}
		\centering
		\includegraphics[width=0.8\columnwidth]{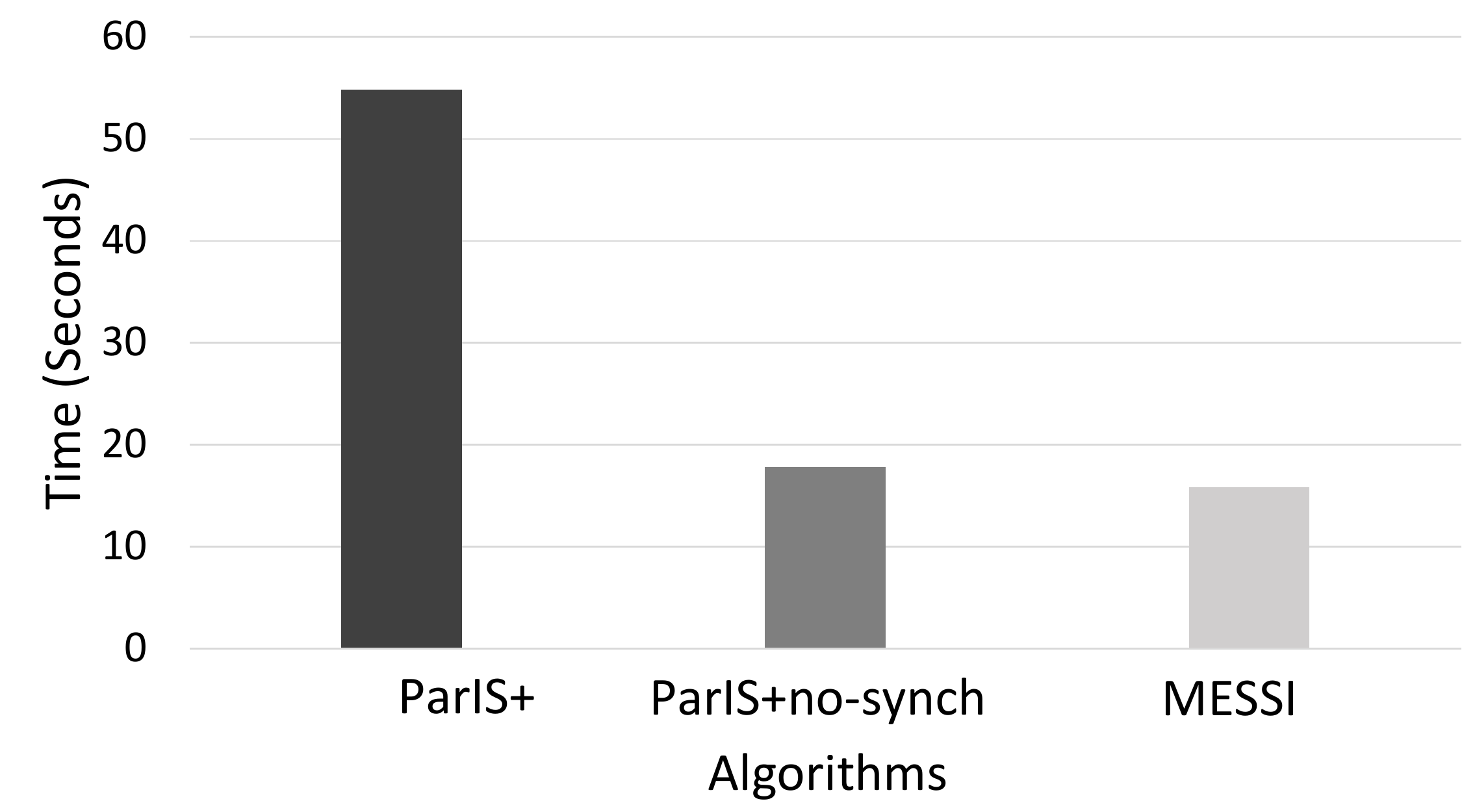}
		\caption{Index creation time}
		\label{fig:stepbystepindex}
	\end{minipage}
\end{figure}

	\begin{figure}[tb]
	\centering
	\includegraphics[width=0.9\columnwidth]{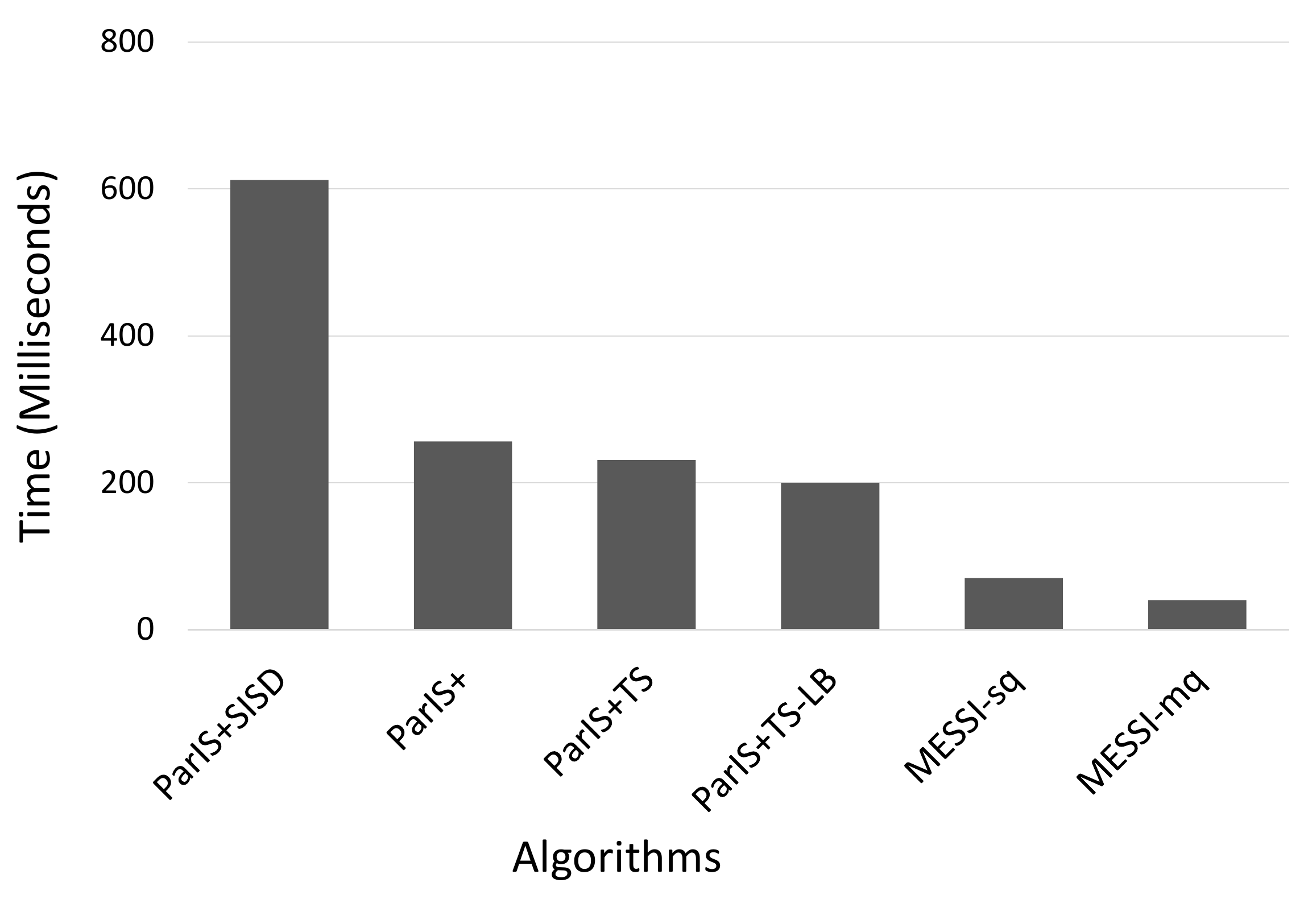}
	\caption{Query answering time}
	\label{fig:stepbystepquery}
\end{figure}

\begin{figure*}[tb]
	\centering

	\includegraphics[page=1,width=1.6\columnwidth]{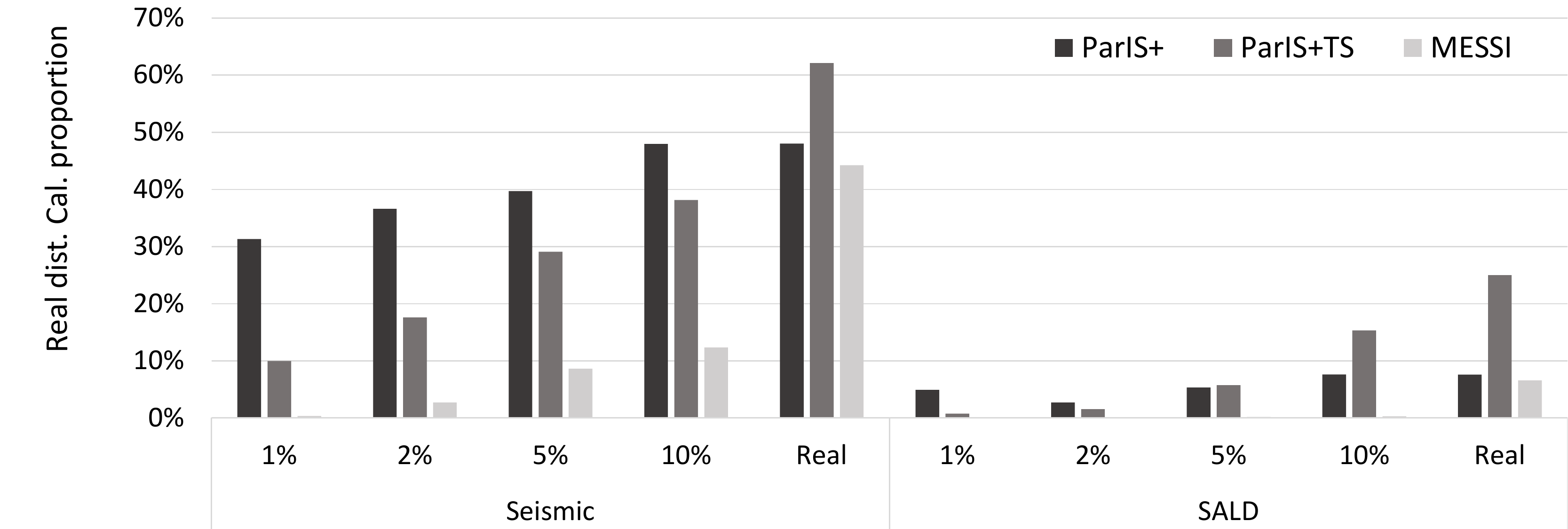}
	\caption{Number of real distance calculations: real datasets, various query workloads.}
	\label{fig:RDnumber}
\end{figure*}

\begin{figure*}[tb]
	\centering
	\hspace*{-0.2cm}
\includegraphics[page=1,width=1.6\columnwidth]{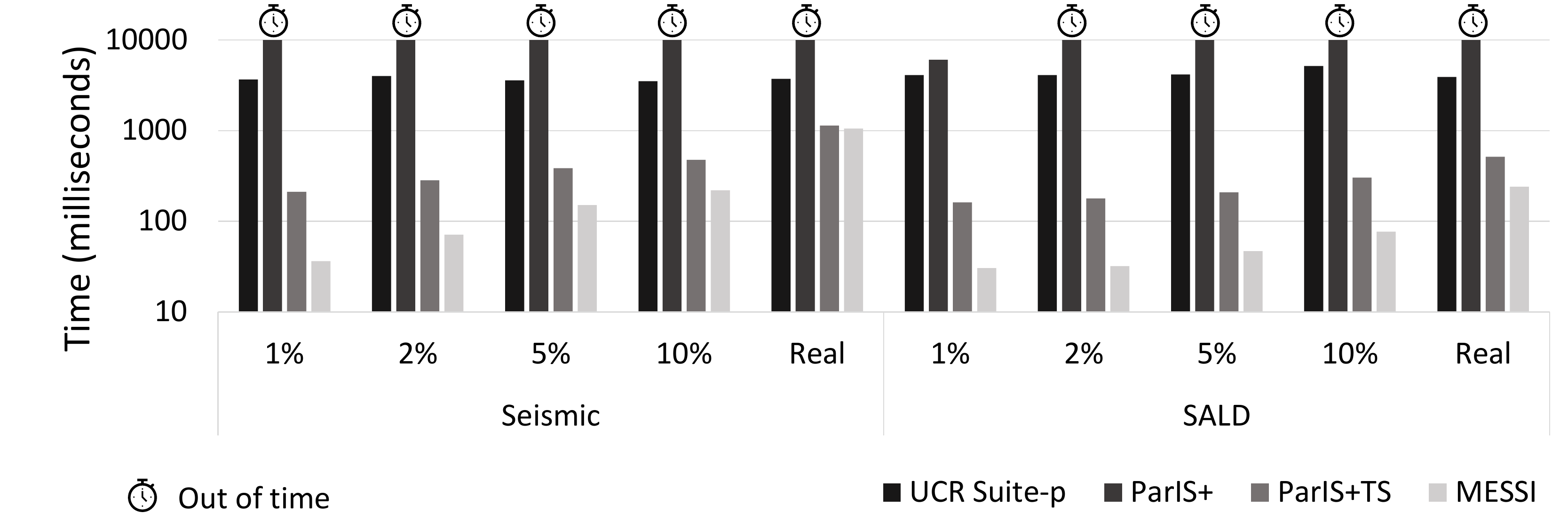}
\caption{Query answering time: real datasets, various query workloads.}
\label{fig:noisequery2}
\end{figure*}

\noindent{\bf [MESSI DTW]}
Figures~\ref{fig:dtwwindow} and~\ref{fig:dtwscalability} compare the performance of MESSI and UCR Suite for the case of DTW distance. 
Overall, MESSI is up to 2.5 
orders of magnitude faster than UCR Suite-p, a parallel version of UCR Suite that uses SIMD and supports DTW.
(For comparison, we also report the performance of the single-core implementation of UCR Suite, which is 1-2 orders of magnitude 
slower than UCR Suite-p.) 

The experimental results on a 100GB dataset show that as we increase the warping window size from 1\% to 20\% of the data series length, the query answering time of MESSI increases as well: the LB\_Keogh envelope of the query becomes wider, and consequently, pruning in the index is smaller (refer to Figure~\ref{fig:dtwwindow}). 
However, MESSI is in all cases at least 9x faster than UCR Suite-p, 
while for the most common warping window sizes of 5-10\%~\cite{keogh2005exact}, the speedup is between 35-170x. 
Figure~\ref{fig:dtwscalability} shows query answering time when varying the dataset size (warping window size: 10\%). 
As we increase the size of the data series collection from 50GB to 200GB, MESSI remains 25-35x faster than UCR Suite-p. 

\begin{figure*}[tb]
		\centering
\begin{minipage}[b]{0.9\columnwidth}
	\centering
	\includegraphics[width=0.65\columnwidth]{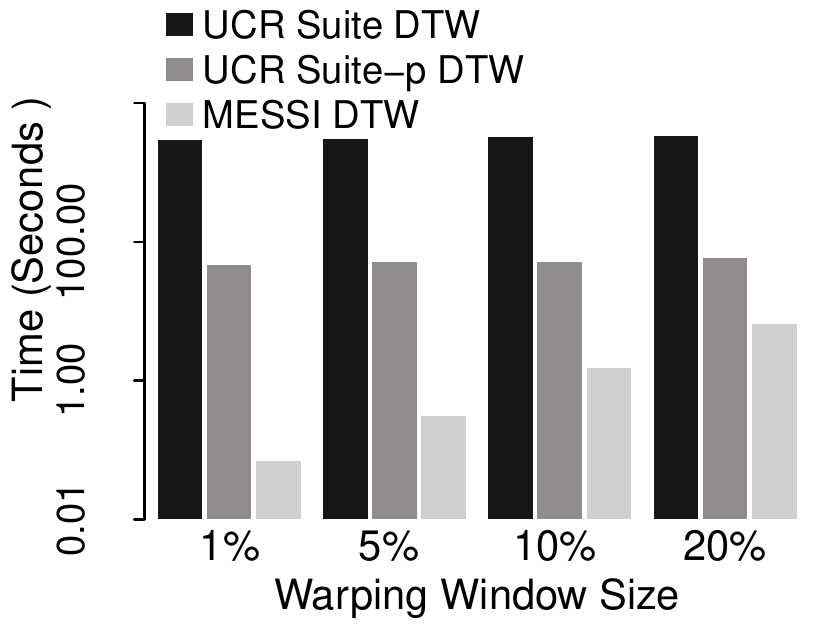}
	\caption{DTW 
		time (synthetic data, varying warping window size, 100GB dataset).}
	\label{fig:dtwwindow}
\end{minipage}
\hspace*{0.9cm}
\begin{minipage}[b]{0.9\columnwidth}
	\centering
		\includegraphics[width=0.65\columnwidth]{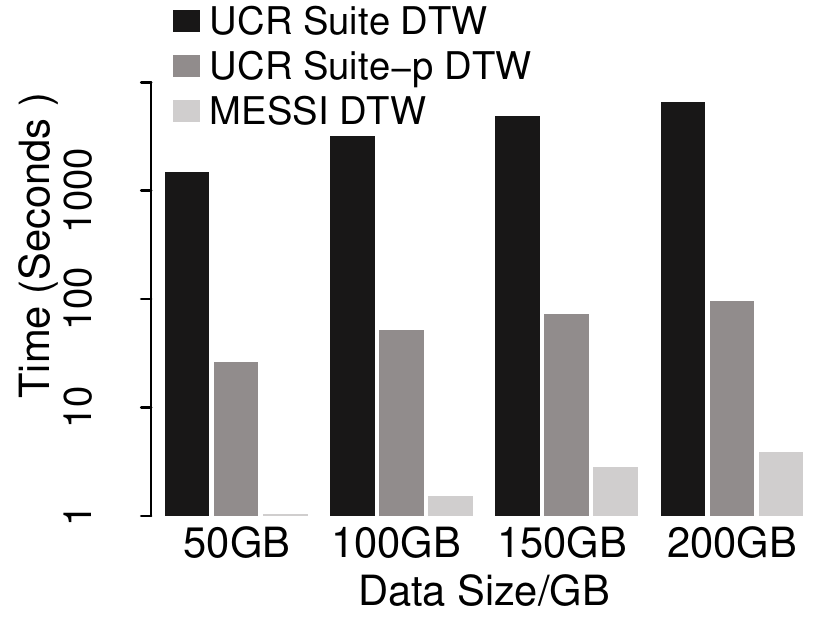}
	\caption{DTW 
		time (synthetic data, varying data size, 10\% warping window).}
		\label{fig:dtwscalability}
\end{minipage}
\end{figure*}

\begin{figure*}[tb]
	\begin{minipage}[t]{0.3\textwidth}
		\centering
		\includegraphics[page=1,width=0.8\textwidth]{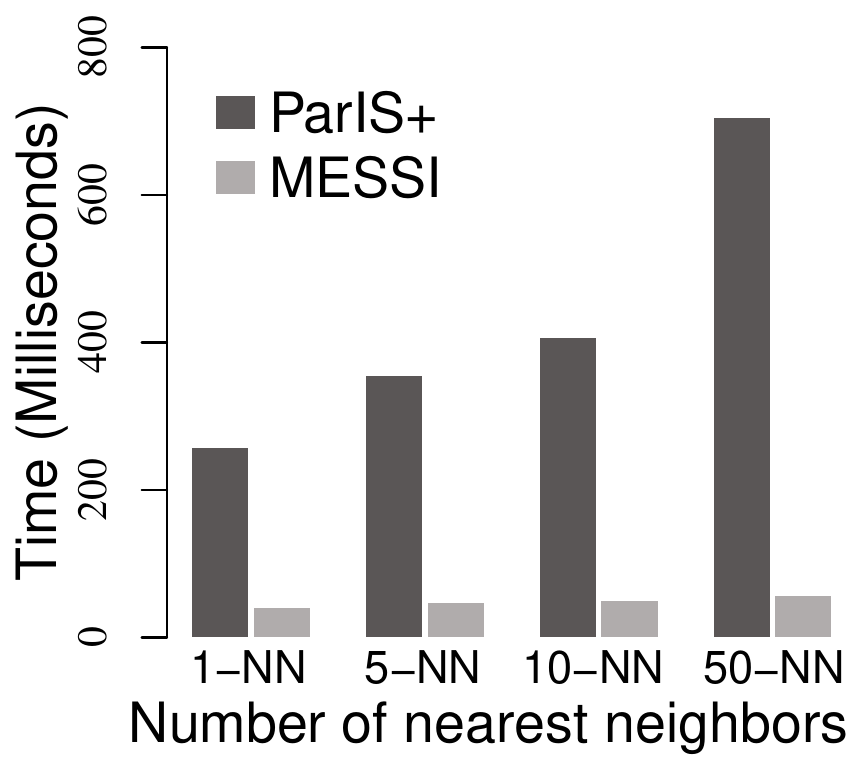}
		\caption{Time for a k-NN Classifier 
			to classify one object using the Euclidean distance (100GB dataset).}
		\label{fig:knn}
	\end{minipage}
	\hspace*{0.7cm}
	\begin{minipage}[t]{0.3\textwidth}
		\centering
		\includegraphics[width=0.8\textwidth]{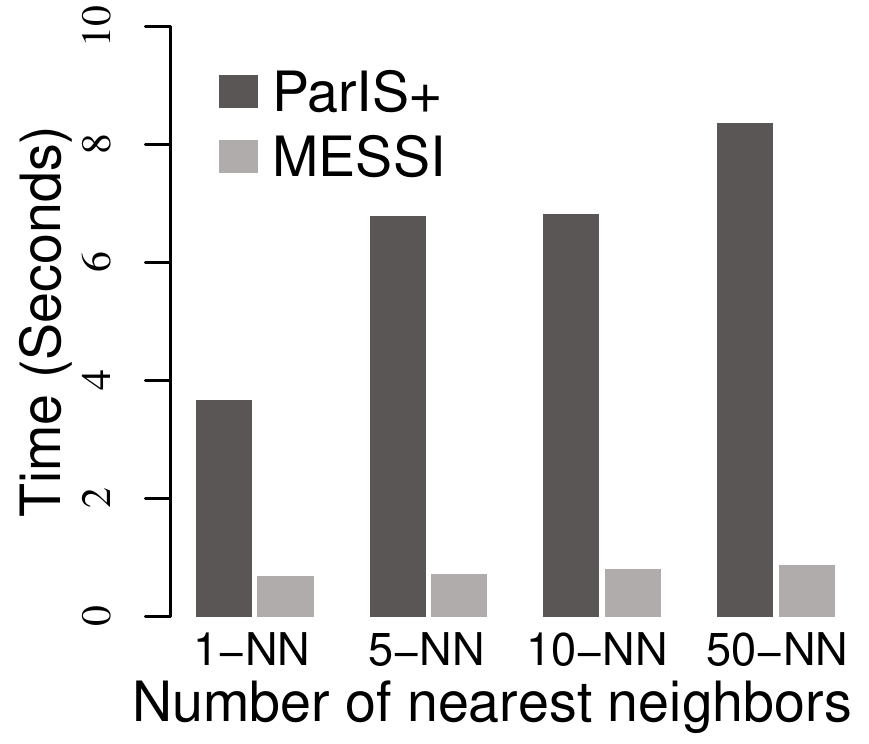}
		\caption{Time for a k-NN Classifier to classify one object using the DTW distance using a warping window size of 5\% of the series length (100GB dataset).}
		\label{fig:dtw5knn}
	\end{minipage}
	\hspace*{0.7cm}
	\begin{minipage}[t]{0.3\textwidth}
		\centering
		\includegraphics[page=1,width=0.8\textwidth]{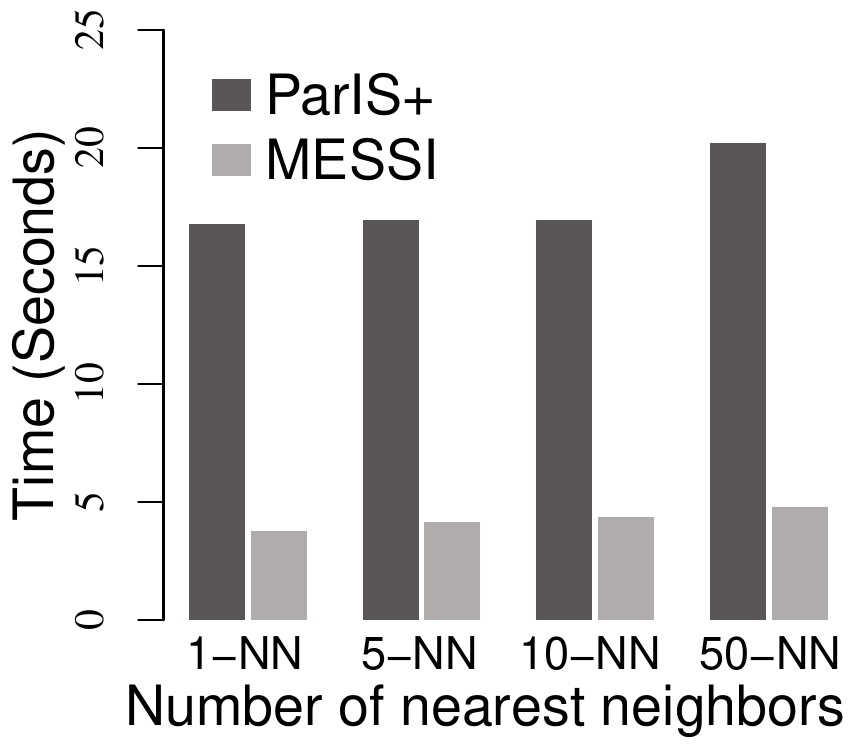}
		\caption{Time for a k-NN Classifier to classify one object using the DTW distance using a warping window size of 10\% of the series length (100GB dataset).}
		\label{fig:dtw10knn}
	\end{minipage}
\end{figure*}

\begin{table}[tb]
	\centering
	\makeatletter\def\@captype{table}\makeatother \caption{Update Frequency of the BSF array (Euclidean distance).}
	{
		\begin{tabular}{|c|c|c|c|c|}
			\hline
			& {\bf 1-NN} & {\bf 5-NN} & {\bf 10-NN} & {\bf 50-NN}\\
			\hline
			number of BSF& & & & \\
			updates/query&11.9&20.9&45.6&258.1\\
			\hline
			BSF update time& & & &\\ 
			$\mu$sec/query&0.5&5.1&19.1&186.5\\
			\hline
			BSF update time& & & &\\ query time \% &0.001\% & 	0.01\% & 	0.04\% &	0.3\%  
			\\
			\hline
		\end{tabular}	
	} 
	\vspace*{0.1cm}
	
	\label{table2}
\end{table}

\subsection{Complex Analytics Task: Classification}

In the following experiment, we tested MESSI on a complex analytics task. 
In particular, we evaluated its performance in a classification task, and measured 
the benefit it would bring to a k-NN Classifier. 
This classifier assigns a new object to the majority class of the k NN of that object (a data series, in our case).

The results, depicted in Figure~\ref{fig:knn}, report the performance of MESSI and ParIS+ for different values of $k$ on a 100GB dataset (100M series of size 256 values, generated with our synthetic data generator). 
The results show that a k-NN Classifier using MESSI can finish a classification task
up to 13x faster than when using ParIS+, which can reduce the total processing time 
for classifying 100K objects from 1 day down to 93min.

We note that the purpose of this experiment was to measure the time performance of executing a k-NN classification task. 
Even though we did not study a real classification
problem, the results are useful in that they report the expected time performance of using MESSI
in such a task with a large data series collection. 

We evaluated the overhead of executing k-NN queries.
MESSI implements k-NN by simply maintaining a sorted array of the best k distances seen so far (i.e., BSF is now an array of k elements).
The elements of the array are initialized by performing a single approximate search (like in 1-NN):
we choose the k series with the smallest distances to the query and initialize the BSF array with their distances. 
Whenever a smaller distance than the biggest element of this array is calculated, the array is updated. 
This process does not result in more operations on the priority queues, or more tree traversals. 
Table~\ref{table2} shows that the additional time
for executing k-NN instead of 1-NN is negligible (times reported in microseconds). 

Finally, we repeated the previous k-NN classification experiment using the DTW distance, which is slower, but can lead to more accurate classifications~\cite{DBLP:journals/datamine/BagnallLBLK17}. 
Figures~\ref{fig:dtw5knn} and~\ref{fig:dtw10knn} show the results (the y-axis is expressed in seconds). 
Observe that when compared to the previous experiment, the execution time (as expected) is now approximately 10x and 30x larger for the 5\% and 10\% warping window sizes, respectively.
Nevertheless, MESSI is up to 9.5x faster than ParIS+ for the 5\% warping window size, and up to 4.5x faster for the 10\% warping.
Therefore, MESSI can considerably reduce the k-NN classification time for large sequence collections. 
Tables~\ref{tabledtw5} and~\ref{tabledtw10} report the number of BSF updates and the time needed to update the BSF during k-NN similarity search queries using the DTW distance with 5\% and 10\% warping window sizes, respectively. 
Similarly to the case of Euclidean distance, we observe that the overhead as the number k of nearest neighbors increases is negligible, even for $k=50$. 
Moreover, since the DTW is computationally more expensive than Euclidean, the percentage of the total query answering time dedicated to updating the BSF shrinks significantly (now expressed as \emph{per thousand}).

In Tables~\ref{lbcost} and~\ref{realcost}, we report the execution time of lower bound and real distance calculations for both the Euclidean and DTW distance measures. 
The results show that the average time cost per lower-bounding calculation is 6.6x slower for DTW than for Euclidean (Table~\ref{lbcost}).
This is due to the fact that DTW needs to pay the cost of computing the LB\_Keogh envelope (for every query).
As expected, the time cost difference between Euclidean and DTW is much larger for the real distance calculation, with DTW being 16x slower than Euclidean (Table~\ref{realcost}).

\begin{table}[tb]
	\centering
	\makeatletter\def\@captype{table}\makeatother \caption{Update Frequency of the BSF array (DTW distance, 5\% warping).}
	{
		\begin{tabular}{|c|c|c|c|c|}
			\hline
			& {\bf 1-NN} & {\bf 5-NN} & {\bf 10-NN} & {\bf 50-NN}\\
			\hline
			number of BSF& & & & \\
			updates/query&22.7&	83.9&	160.2&	672.4\\
			\hline
			BSF update time& & & &\\ 
			$\mu$sec/query&4.9&19.9&		50.2&473.3
			\\
			\hline
			BSF update time& & & &\\ query time \textperthousand 
			&0.007\textperthousand  
			&0.03\textperthousand
			&0.06\textperthousand
			&0.5\textperthousand
			\\
			\hline
		\end{tabular}	
	} 
	\vspace*{0.1cm}
	
	\label{tabledtw5}
\end{table}

\begin{table}[tb]
	\centering
	\makeatletter\def\@captype{table}\makeatother \caption{Update Frequency of the BSF array (DTW distance, 10\% warping).}
	{
		\begin{tabular}{|c|c|c|c|c|}
			\hline
			& {\bf 1-NN} & {\bf 5-NN} & {\bf 10-NN} & {\bf 50-NN}\\
			\hline
			number of BSF& & & & \\
			updates/query&45.8&	124.7&	221.4&	854.1\\
			\hline
			BSF update time& & & &\\ 
			$\mu$sec/query&11.9 &31.5&72.5&574.2
			\\
			\hline
			BSF update time& & & &\\ query time \textperthousand 
			&0.003\textperthousand&
			0.008\textperthousand&
			0.002\textperthousand&
			0.1\textperthousand
			\\
			\hline
		\end{tabular}	
	} 
	\vspace*{0.1cm}
	
	\label{tabledtw10}
\end{table}

\begin{table}[tb]
	\centering
	\caption{Time cost of lower bound distance calculations.}\label{lbcost}
	\begin{tabular}{ccc}
		\toprule
		Distance measure& SISD&SIMD \\
		\midrule
		Euclidean & 107.5ns&31.4ns\\
		DTW& 122.7ns&30.5ns\\ 
		\bottomrule
	\end{tabular}
\end{table}

\begin{table}[tb]
	\centering
	\caption{Time cost of real distance calculations.}
	\label{realcost}
	\begin{tabular}{ccc}
		\toprule
		Distance measure& Query& Time (milliseconds)\\
		\midrule
		Euclidean& 1-NN& 40.3\\
		& 50-NN& 56.3\\
		DTW& 1-NN& 679.3\\
		& 50-NN& 879.7\\
		\bottomrule
	\end{tabular}
\end{table}


\section{Related Work}
\label{sec:rel}

Various dimensionality reduction techniques exist 
for data series, which can then 
be scanned and filtered~\cite{DBLP:conf/kdd/KashyapK11,Li1996} or indexed and pruned~\cite{evolutionofanindex, wang2013data, shieh2008sax,Shieh2009, zoumpatianos2016ads, peng2018paris, DBLP:journals/pvldb/KondylakisDZP18, DBLP:conf/ssd/Chatzigeorgakidis19, ulissevldb, DBLP:journals/vldb/KondylakisDZP19, DBLP:conf/gis/Chatzigeorgakidis19, coconutpalm, DBLP:conf/edbt/Chatzigeorgakidis21}
during query answering, including deep-learned methods~\cite{seanet}; for a complete discussion of such techniques, we refer the reader to two recent tutorials on the subject~\cite{bigdatatutorial,DBLP:conf/edbt/EchihabiZP21}.

We follow the same approach of indexing the series based on their summaries,
though our work is the first to exploit the parallelization opportunities offered by modern hardware, 
in order to accelerate in-memory index construction and similarity search for data series. 
The work closest to ours is Paris/ParIS+~\cite{peng2018paris, parisplus}, which exploits modern hardware, but was designed for disk-resident datasets (see also Section~\ref{sec:prelim}).

FastQuery is an approach used to accelerate search operations in scientific data~\cite{chou2011fastquery}, 
based on the construction of bitmap indices.
In essence, the iSAX summarization used in our approach is an equivalent solution, 
though, specifically designed for sequences (which have high dimensionalities).

The interest in using SIMD instructions for improving the performance of data management solutions is not new~\cite{zhou2002implementing}.
However, it is only more recently that relatively complex algorithms were extended in order to take advantage of this hardware characteristic. 
Polychroniou et al.~\cite{polychroniou2015rethinking}  introduced design principles for efficient vectorization of in-memory database operators (such as selection scans, hash tables, and partitioning). 
For data series in particular, previous work has used SIMD for Euclidean distance computations~\cite{tang2016exploit}. 
Following~\cite{peng2018paris}, in our work we use SIMD both for the computation of Euclidean distances, 
as well as for the computation of lower bounds, 
which involve branching operations.

Multi-core CPUs offer thread parallelism through multiple cores and simultaneous multi-threading (SMT).
Thread-Level Parallelism (TLP) methods, 
like multiple independent cores and hyper-threads are used to increase efficiency~\cite{gepner2006multi}.

A recent study proposed a high performance temporal index similar to time-split B-tree (TSB-tree), 
called TSBw-tree, which focuses on transaction time databases~\cite{LometN15}. 
Binna et al.~\cite{binna2018hot}, present the Height Optimized Trie (HOT), 
a general-purpose index structure for main-memory database systems, while Leis et al.~\cite{leis2013adaptive} describe an in-memory adaptive Radix indexing technique that is designed for modern hardware. 
Xie et al.~\cite{xie2018comprehensive}, 
study and analyze five recently proposed indices, 
i.e., FAST, Masstree, BwTree, ART and PSL 
and identify the effectiveness of common optimization techniques, 
including hardware dependent features such as SIMD, NUMA and HTM. 
They argue that there is no single optimization strategy that fits all situations, due to the differences in the dataset and workload characteristics. 
Moreover, they point out the significant performance gains that the use of modern hardware features 
bring to in-memory indices. 

We note that the indices described above are not suitable for data series (or very high-dimensional data), which is the focus of our work, and which pose very specific data management challenges with their hundreds, or thousands of dimensions (i.e., the length of the sequence).
%
%
%
Techniques specifically designed for modern hardware and in-memory operation have also been studied in the context of adaptive indexing~\cite{alvarez2014main}, and data mining~\cite{tatikonda2008adaptive}.

%

Piatov et al. propose DeltaTop, a fast time series subsequence matching method~\cite{piatov2019interactive}, where the query sequence is itself part of the dataset (i.e., self-join). 
Their method uses a prefix-sum Euclidean distance matrix to accelerate subsequence matching, and supports search in multi-variate time series. 
The authors provide a parallel (multi-core) implementation of their method.
Compared to our work, we observe that they solve the problem of (self-join) subsequence similarity matching, only for non Z-normalized sequences using the Euclidean distance. 
In contrast, we solve the whole-matching problem~\cite{lernaeanhydra}, supporting non Z-normalized and Z-normalized sequences, using both the Euclidean and the DTW distances.
Even when the two approaches are compared in the specific setting of subsequence similarity search\footnote{MESSI can be adapted to support subsequence matching as follows: given a long series (in which we need to identify the most similar subsequence to the query), we extract subsequences from the long series by sliding a window window (of the same length as the query) over the entire length of the series, and then index all these subsequences.} on non Z-normalized data using the Euclidean distance, we observed that DeltaTop's high index creation time and memory cost did not allow it to scale to sequence collections with more than 100K points (i.e., considerably smaller than the datasets we used for evaluating MESSI).

Finally, KV-Match~\cite{wu2019kv} and its improvement, L-Match~\cite{DBLP:journals/access/FengWWW20}, are index structures that can support similarity search in a distributed setting.
Nevertheless, we note that they were developed for subsequence matching on disk-resident data, while the focus of MESSI is on whole-matching for in-memory data.

\section{Conclusions} 
\label{sec:conclusions}
Data series are a very common data type, with increasingly larger collections being generated by applications in many and diverse domains.
In many exploration and analysis pipelines, similarity search is a key operation, which is nevertheless challenging to efficiently support over large data series collections.

In this work, we proposed MESSI, 
a data series index designed for in-memory operation by 
exploiting the parallelism opportunities of modern hardware. 
MESSI is up to 4x faster in index construction and up to 11x faster in query answering than the state-of-the-art solution, and is the first technique to answer 
answering 
exact similarity search queries on 100GB datasets in $\sim$50msec.
This level of performance enables for the first time interactive 
analytics on very large data series collections.


{Finally, we note that the ideas presented in this work are applicable to other indices that have a root node with a large fanout degree. 
This is true for other iSAX-based indices. 
For example, we could parallelize in a way similar to MESSI the ULISSE index~\cite{DBLP:journals/vldb/LinardiP20}, which supports queries of variable length, as well as the DPiSAX index~\cite{dpisaxjournal}, which is a distributed index operating on top of Spark (but currently not supporting parallel execution within each node of the Spark cluster).
It is an interesting open problem to study whether there exist 
efficient parallelization techniques for indexing schemes whose tree index does not satisfy this large fanout property that would result in better perfomance than MESSI.


\vspace*{0.5cm}
\noindent{\bf [Acknowledgements]}
Work supported by  Investir l’Avenir, Univ. of Paris IDEX Emergence en Recherche ANR-18-IDEX-000, CSC, FMJH PGMO, EDF, Thales, HIPEAC 4.
Partly performed when P. Fatourou visited LIPADE and B. Peng visited CARV, FORTH~ICS.

\bibliographystyle{spmpsci} 
\bibliography{parisinmemory}

\end{document}